\documentclass{article}
\usepackage[final]{neurips_2022}

\usepackage[utf8]{inputenc} 
\usepackage[T1]{fontenc}    
\usepackage{hyperref}       
\usepackage{url}            
\usepackage{booktabs}       
\usepackage{amsmath, amssymb, amsthm}
\usepackage{amsfonts}       
\usepackage{nicefrac}       
\usepackage{microtype}      
\usepackage{algorithm}
\usepackage[noend]{algpseudocode}
\usepackage{enumitem}
\usepackage{tikz}
\usetikzlibrary{shapes.geometric, arrows}
\usetikzlibrary{shapes.multipart}
\usepackage{xcolor}
\usepackage{times}

\title{Robust Model Selection and Nearly-Proper Learning for GMMs}

\author{%
Allen Liu \\
MIT\\
Cambridge, MA 02139 \\
\texttt{cliu568@mit.edu} \\
\And
Jerry Li \\
Microsoft Research\\
Redmond, WA 98052 \\
\texttt{jerrl@microsoft.com} \\
\And
Ankur Moitra \\
MIT\\
Cambridge, MA 02139 \\
\texttt{moitra@mit.edu} \\
}

\theoremstyle{plain}
\newtheorem{lemma}{Lemma}[section]
\newtheorem*{lemma*}{Lemma}
\newtheorem{corollary}[lemma]{Corollary}
\newtheorem*{corollary*}{Corollary}

\newtheorem{claim}[lemma]{Claim}
\newtheorem{fact}[lemma]{Fact}

\theoremstyle{plain}
\newtheorem{theorem}[lemma]{Theorem}
\newtheorem*{theorem*}{Theorem}
\newtheorem*{problem*}{Problem}
\newtheorem{definition}[lemma]{Definition}

\newtheorem*{remark}{Remark}

\newcommand\TV{\textsf{TV}}
\newcommand\comment[1]{}
\newcommand\abs[1]{\vert{#1}\vert}
\newcommand\norm[1]{\Vert{#1}\Vert}
\newcommand\R{\mathbb{R}}

\newcommand\Z{\mathbb{Z}}
\newcommand\C{\mathbb{C}}

\newcommand{\poly}{{\rm poly}}

\newcommand\eps{\epsilon}

\newcommand{\trunc}{\textsf{trunc}}

\newcommand\dims{m}

\newcommand\mcl{\mathcal}

\newif\ifrandom
\randomtrue

\newcommand{\wt}{\widetilde}

\newcommand{\wh}{\widehat}

\begin{document}

\maketitle

\begin{abstract}

In learning theory, a standard assumption is that the data is generated from a finite mixture model. But what happens when the number of components is not known in advance? The problem of estimating the number of components, also called {\em model selection}, is important in its own right but there are essentially no known efficient algorithms with provable guarantees let alone ones that can tolerate adversarial corruptions. 
In this work, we study the problem of robust model selection for univariate Gaussian mixture models (GMMs). Given $\poly(k/\eps)$ samples from a distribution that is $\eps$-close in TV distance to a GMM with $k$ components, we can construct a GMM with $\wt{O}(k)$ components that approximates the distribution to within $\wt{O}(\eps)$ in $\poly(k/\eps)$ time.  Thus we are able to approximately determine the minimum number of components needed to fit the distribution within a logarithmic factor.   Prior to our work, the only known algorithms for learning arbitrary univariate GMMs either output significantly more than $k$ components (e.g. $k/\eps^2$ components for kernel density estimates) or run in time exponential in $k$. 
Moreover, by adapting our techniques we obtain similar results for reconstructing Fourier-sparse signals.

\end{abstract}

\section{Introduction}

Many works in learning theory operate under the assumption that the data is generated from a finite mixture model, and furthermore that the number of components is known in advance.  But what happens when the number of components is not known in advance? The problem of estimating the number of components is called {\em model selection} and has been intensively studied in statistics for over fifty years \citep{neyman1966use}. Indeed, in many scientific applications, it is the central issue. Consider the motivation given by \citet{chen2004testing}: In genetics, we might have a continuous-valued trait, like height, that can be measured across a population and we want to understand its genetic basis. But is the underlying genetic mechanism simple or complex? Is it controlled by just a few genes or are there many more genes waiting to be discovered that each have a small effect on it?

From a statistical perspective, what makes model selection challenging is that the standard analysis of the likelihood ratio test breaks down because of lack of regularity and non-identifiability \citep{hartigan1985failure}. Despite many attempts \citep{ghosh1984asymptotic, lo2001testing, huang2017model} and rejoinders \citep{jeffries2003note}, even understanding the asymptotic distribution of the likelihood ratio statistics has remained a long-standing challenge in the field \citep{kasahara2015testing}. From an algorithmic standpoint, the problem is even more difficult. 

In this work, we study the problem of robust model selection for one-dimensional Gaussian mixture models with $k$ components ($k$-GMMs for short).
A natural approach for this problem is via \emph{agnostic proper learning}, where the task is to, given samples from an unknown distribution, output the best $k$-GMM approximation to this distribution in TV distance.
An efficient agnostic proper learning algorithm, combined with standard tools from hypothesis testing, would immediately yield an algorithm for model selection.

Unfortunately, while there are many efficient algorithms for learning one-dimensional GMMs, they all fall into one of several categories: $(1)$ They assume some strong separation conditions on the components so that the samples can be clustered based on which component they were generated from. $(2)$ They solve the harder problem of learning the parameters of the components, which information-theoretically requires the number of samples to be exponential in $k$ \citep{moitra2010settling}. $(3)$ They employ brute-force search \citep{daskalakis2014faster, acharya2014near} or solve a system of polynomial inequalities \citep{li2017robust}, and run in time exponential in $k$. $(4)$ They learn an approximation that is either not a GMM, e.g. a piece-wise polynomial approximation \citep{chan2013efficient, acharya2017sample} or output a GMM where the number of components is much larger than $k$~\citep{wu2018improved, devroye2012combinatorial,bhaskara2015sparse}. $(5)$ They assume that the components in the GMM have the same or similar variances and means not too far apart so that there is a good approximation to the density with just a logarithmic number of components~\citep{Wu2018optimal,polyanskiy2020self}.
In all cases, these guarantees are insufficient for efficient model selection, and/or yield a trivial approximation to the number of components in a GMM except in restricted settings. 
In this work, we ask: Are there efficient algorithms for learning arbitrary one-dimensional GMMs that output an approximation with $\widetilde{O}(k)$ components? Relatedly: Are there efficient algorithms for approximating the number of components in a GMM? We give efficient algorithms whose running time and sample complexity are polynomial in $k$ for both of these problems, and also the related problem of reconstructing Fourier-sparse signals with an unknown number of frequencies.

\subsection{Learning and model selection for GMMs}

Our main result is a new robust learning algorithm for one-dimensional GMMs. We show:

\begin{theorem}\label{thm:main-gmm}
Let $k, \eps > 0$ be parameters and let $f$ be a distribution such that $d_{\TV}(\mcl{M}, f) \leq \eps$ for some unknown mixture of Gaussians $\mcl{M} = w_1 G_1 + \dots + w_k G_k$.  Assume that we are given $\wt{O}(k /\eps^2)$ samples from $f$.  Then there is an algorithm that runs in $\poly(k /\eps)$ time and with probability $0.9$ (over the random samples), outputs a mixture of $\wt{O}(k)$ Gaussians,  $\wt{\mcl{M}}$, such that 
\[
d_{\TV}(\wt{\mcl{M}}, f ) \leq \wt{O}(\eps) \,. 
\]
\end{theorem}

\noindent In contrast to other known learning algorithms (discussed earlier), our learning algorithm works for arbitrary GMMs, runs in polynomial time and uses a polynomial number of samples, and while it does not output a GMM with exactly $k$ components, it does the next best thing: it outputs a GMM with at most a polylogarithmic factor more components. 


As a corollary, we also give an algorithm for robust approximate model selection for GMMs. The connection to model selection is that when our algorithm fails to find a GMM with $\widetilde{O}(k)$ components that fits the data we can be assured that there must more than $k$ components to begin with. Notice in particular that improper approximations by themselves do not suffice for the model selection problem, as a good improper approximation could exist even if the distribution is far from any GMM with $\widetilde{O} (k)$ components.

\begin{theorem}\label{coro:gmm-hypothesis-test}
Let $k, \eps >0$ be parameters we are given.  Let $\mcl{F}_1$ be the family of distributions that are $\eps$-close to a $k$-GMM with $k$ components (in TV distance).   Let $\mcl{F}_2$ be the family of distributions that are not $\wt{O}(\eps)$-close to any GMM with $\wt{O}(k)$ components.  There is an algorithm that given $\poly(k /\eps)$ samples from a known distribution $\mcl{D}$, runs in $\poly(k /\eps)$ time, and outputs $1$ if $\mcl{D} \in \mcl{F}_1$ and outputs $2$ if $\mcl{D} \in \mcl{F}_2$ both with failure probability at most $0.2$.
\end{theorem}
\begin{remark}
Even if the distribution $\mcl{D}$ is completely unknown and we are only given samples from it, the above result still holds as long as $\mcl{D}$ is somewhat well behaved (note that such an assumption is necessary as hypothesis testing with respect to total variation distance without any assumptions on $\mcl{D}$ is impossible).  In particular we can use piecewise polynomial approximation \citep{chan2013efficient} or kernel density estimates \citep{terrell1992variable} to learn a distribution $\mcl{D}'$ that is close to $\mcl{D}$ that we have an explicit form for and then run the hypothesis test using $\mcl{D}'$.
\end{remark}

\subsection{Fourier sparse interpolation}

Our techniques also immediately apply to the problem of Fourier sparse interpolation, where the goal is to interpolate a signal based on noisy measurements of it at a few points \citep{chen2016fourier}.
We say that a function $\mcl{M}$ is $(k, C)$ simple if it can be written in the form
$$ \mcl{M}(t) = \sum_{j=1}^k a_j e^{2 \pi i \theta_j t} \; ,$$
where additionally $\sum_j |a_j| \leq C$.
In other words, a function is $(k, C)$ simple if it is $k$-sparse in the Fourier domain, and its Fourier coefficients are bounded in $\ell_1$ by $C$.

We consider the following problem.
We get query access to a function $f(t) = \mcl{M}(t) + \eta(t)$
at any point in the interval $[-1, 1]$, where $\mcl{M}$ is $(k, C)$ simple and has all frequencies in the interval $[-F, F]$, and $\eta(t)$ is noise that we will assume is bounded in $L_2$ norm. The goal is to compute a Fourier-sparse approximation $\widetilde{\mcl{M}}(t)$ that is close to $f(t)$, in the sense that its error is comparable to that of $\mcl{M}(t)$. Recently \citet{chen2016fourier} showed how to construct an approximation $\widetilde{\mcl{M}}(t)$ that satisfies
$$\|f(t) - \widetilde{\mcl{M}}(t)\|_2 \leq \| \eta(t) \|_2 + \eps \|\mcl{M}(t)\|_2$$
where the $L_2$ norm is taken over the interval $[-1, 1]$. Their algorithm works for any $\eps > 0$ and uses $\mbox{poly}(k, \log 1/\eps) \log F$ measurements. Moreover the $\widetilde{\mcl{M}}(t)$ that they output is $\mbox{poly}(k, \log 1/\eps)$-Fourier sparse. Similarly to the GMM setting, a natural goal is to perform robust interpolation but with tighter bounds on the number of frequencies. We show:

\begin{theorem}\label{thm:fourier-main}
Let $f, \mcl{M}$ be as above where $\mcl{M}$ is $(k,1)$-simple. Then for any desired accuracy $\eps > 0$ and constant $c > 0$, in $\poly(k, \log 1/\eps)  \log F  $ queries and $\poly(k /c , \log 1/  \eps)  \log^2 F$ time, we can output a function $\wt{\mcl{M}}$ such that with probability $1 - 2^{-\Omega(k)}$,
\begin{enumerate}
    \item $\wt{\mcl{M}}$ is $\widetilde{O}(k)$-Fourier sparse with $\norm{\wh{\wt{\mcl{M}}}}_1 \leq \widetilde{O}(k)$
    \item 
    $\int_{-1 + c}^{1 - c} |\wt{\mcl{M}} - f|^2 \leq  \widetilde{O}\left(\eps^2 + \int_{-1}^1 |f - \mcl{M}|^2 \right) $
\end{enumerate}
\end{theorem}
\begin{remark}
Note the constraints $\norm{\wh{\mcl{M}}}_1 $ and $\norm{\wh{\wt{\mcl{M}}}}_1$ translate into bounds on the sizes of the coefficients of the exponentials in $\mcl{M}$ and $\wt{\mcl{M}}$ respectively. 
\end{remark}


The natural open question left by our work is to improve the sparsity bounds, both for interpolation/learning and model selection. In principle it could be possible that there are efficient algorithms for these problems, however it now seems somewhat unlikely. Even without noise, learning a Gaussian mixture model with $k$ components without a separation condition in time $\mbox{poly}(k, 1/\epsilon)$ is open. From our work (see Section~\ref{sec:overview}), we see that even in the well-conditioned case this is equivalent to finding a non-trivially sparse solution to a system of polynomial equations where there seems to be no structure that makes algorithmic search better than brute-force possible.  Moreover, this question has already been open for many years, but there hasn't been any progress on proper learning.  Thus, we conjecture that both the learning and model selection problems are computationally hard if we are not allowed to relax the number of components. 

\subsection{Related work}
There is a vast literature on the three problems we consider. Here we will give a more detailed review of related work. 

\paragraph{Learning Mixtures of Gaussians and Model Selection} Since the pioneering work of \citet{pearson1894contributions}, mixtures of Gaussians have become one of the most ubiquitous and well-studied generative models in both theory and practice.
Numerous problems have been studied on the context of learning mixtures of Gaussians, including clustering~\citep{dasgupta1999learning,vempala2004spectral,achlioptas2005spectral,dasgupta2007probabilistic,AroraK07,kumar2010clustering,awasthi2012improved,mixon2017clustering,hopkins2018mixture,kothari2018robust,DiakonikolasKS18}, learning in the presence of adversarial noise in high dimensional settings~\citep{DiakonikolasKS18,hopkins2018mixture,kothari2018robust,bakshi2020robustly,diakonikolas2020robustly,kane2021robust,liu2020settling,liu2021learning}, parameter estimation~\citep{kalai2010efficiently,belkin2015polynomial,moitra2010settling,hardt2015tight}, learning in smoothed settings~\citep{hsu2013learning,anderson2014more,bhaskara2014smoothed,ge2015learning}, and density estimation~\citep{devroye2012combinatorial,chan2014efficient,acharya2017sample}.

Of particular interest to us is the line of work on proper learning~\citep{feldman2006pac,acharya2014near,li2017robust,ashtiani2018nearly}, where the goal is to output a mixture of $k$-Gaussians which is close in total variation to the underlying ground truth.
Unfortunately, while the sample complexity of these algorithms is usually polynomial, the runtime for all known approaches is exponential in $k$.
In contrast, our runtimes are polynomial, albeit for a relaxed version of the problem, where the output is allowed to be a mixture of $k'$ Gaussians, for $k' > k$.

For this ``semi-proper'' regime, efficient algorithms are known, albeit either only for restricted settings, or with significantly worse quantitative results than we achieve.
In the ``well-conditioned'' case, where the means are close together, and the variances of all the components are comparable, the aforementioned work of~\citep{Wu2018optimal,polyanskiy2020self} demonstrates that the nonparametric MLE can efficiently obtain an estimate using only logarithmically many pieces.
However, the nonparametric MLE is not suited for the general setting, where the means could be far apart, and variances could be very different, and will not converge in general.
Moreover, while nonparametric MLE is robust to perturbations in KL, it is not robust to perturbations in total variation distance, as we consider here.

For the general case, by using kernel density estimates, one can achieve $\eps$ approximation using $k' = O(k / \eps^C)$ for some constant $C$~\citep{devroye2012combinatorial}.
Similarly~\citet{bhaskara2015sparse} achieves $\eps$ error using $k' = O(k / \eps^3)$ pieces.
That is, for both of these approaches, they require a number of pieces which scales polynomially with $1 / \eps$.
In comparison, our dependence on $\eps$ in terms of the number of pieces is logarithmic.

As discussed previously, there are strong connections between proper learning and model selection
~\citep{neyman1966use,hartigan1985failure,ghosh1984asymptotic, lo2001testing,jeffries2003note,kasahara2015testing,huang2017model}.
Related notions have been considered in distribution testing~\citep{parnas2006tolerant,ValiantV10a,ValiantV10b,ValiantV11a,JiaoHW16,JiaoVHW17,HanJW16} and testing properties of boolean functions~\citep{diakonikolas2007testing, iyer2021junta}.

\paragraph{Continuous Time Sparse Fourier Transforms} Sparse Fourier transforms in the continuous setting, also known as sparse Fourier transforms off the grid, has been the subject of intensive study.
Indeed, the first algorithm for this problem dates back to \citet{pronyessai}.
Modern algorithms include MUSIC~\citep{schmidt1982signal},  ESPRIT~\citep{roy1986esprit}, maximum likelihood estimators \citep{bresler1986exact}, convex programming based methods \citep{candes2014towards} and the matrix pencil method \citep{moitra2015super}.

Most of these works, especially those that work in a noisy setting, require a frequency gap. Moreover they require more than $k$ samples (their bound usually depends on the frequency gap), even if the underlying signal is $k$-sparse in the Fourier domain.
A recent line of work has focused on the problem of improving the sample complexity \--- in particular getting bounds which only depend on $k$ with runtimes that are polynomial in $k$ \citep{fannjiang2012coherence,duarte2013spectral,tang2013compressed,tang2014near,boufounos2015s,huang2015super,price2015robust}. The setting where there is no gap and there is noise is particularly challenging. One approach is to relax the definition of a frequency gap, and require it only between ``clusters" of frequencies \citep{batenkov2020conditioning}. 
Another line of work~\citep{avron2019universal,chen2019active} shows how to output a hypothesis which is $k$-sparse without any gap assumptions and with sample complexity which is polynomial in $k$. However these methods run in exponential time.
As we previously discussed, the most relevant works to us are~\citet{chen2016fourier} and ~\citet{chen2019active}, which give an algorithm whose running time and sample complexity are polynomial in $k$ that works without any gap assumptions, but for a relaxation where we are allowed to output a $\widetilde{O} (k^2)$-Fourier sparse signal.

\section{Technical overview}
\label{sec:overview}

We now give an overview of our approach.  We will focus on just the GMM case in this overview.  Our approach for sparse Fourier interpolation follows a very similar outline.  We first present our techniques assuming that we have explicit access to $f$.  In Section~\ref{sec:abstract-away-samples} we show how to reduce to this case when we are only given samples.  In other words, the problem is as follows: we are given a function $f$, and we want to find a sparse approximation to $f$ as a nonnegative sum of Gaussians, i.e. we want to write
\[
f \sim a_1 G_1 + \dots + a_n G_n
\]
with $n$ small, where each $G_i$ is a Gaussian. 

\subsection{Well-conditioned case}

We first solve the ``well-conditioned" case.  Roughly, we say that a GMM is well-conditioned if the variances of the components are all constant scale and the means are all not too far from zero.  Formally, we have the following definition:
\begin{definition}
We say a Gaussian $G = N(\mu, \sigma^2)$ is $\delta$-well-conditioned if $\abs{\mu} \leq \delta$ and $|\sigma^2  - 1| \leq \delta$.
Furthermore we say a mixture of Gaussians $\mcl{M} = w_1 G_1 + \dots + w_kG_k$  is $\delta$-well-conditioned if all of the components $G_1, \dots , G_k$ are $\delta$-well-conditioned.
\end{definition}
Naturally, our techniques also apply to a shared scaling and/or translation of the components, but we will ignore this for now. Earlier work of~\citep{Wu2018optimal, polyanskiy2020self} proved an important structural result that a well-conditioned GMM can be $\eps$-approximated by a mixture with $O(\log 1/\eps)$ components. However we will want a robust and algorithmic version: In particular, instead of requiring the distribution to be exactly a well-conditioned GMM, we will only require that it be close in total variation distance. Even in this setting, with some level of model misspecification, we want an efficient algorithm for constructing an approximating GMM with few components. 
To this end, a key result, proved in Section~\ref{sec:GMM-well-conditioned}, is:

\begin{lemma}
Let $\eps > 0$ be a parameter.  Assume we are given access to a distribution $f$ such that $d_{\TV}(f , \mcl{M}) \leq \eps$ where $\mcl{M} = w_1G_1 + \dots + w_kG_k$ is a $0.5$-well-conditioned mixture of Gaussians.  Then we can compute, in $\poly(1/\eps)$ time, a mixture $\wt{\mcl{M}}$ of at most  $O(\log 1/\eps)$ Gaussians such that $d_{\TV}(\mcl{M}, \wt{\mcl{M}}) \leq \wt{O}(\eps)$.
\end{lemma}



Our approach departs from the moment matching framework of  \cite{Wu2018optimal, polyanskiy2020self}. Instead we take the probability density function of any well-conditioned Gaussian $G_j$. We can expand it as a Taylor series around $0$ of the form
\[
G_j(x) = c_{G_j}^{(0)} + \frac{c_{G_j}^{(1)}x}{1!} + \frac{c_{G_j}^{(2)}x^2}{2!} + \dots 
\]
for some coefficients $c_{G_j}^{(i)}$.  
We can then associate it with the vector $c_{G_j} = (c_{G_j}^{(0)}, \dots , c_{G_j}^{(\ell-1)} ) $ of length $\ell = O(\log 1/\eps)$.
Then, for any well-conditioned mixture $\mcl{M} = a_1 G_1 + \dots + a_k G_n$, we can associate it with the corresponding convex combination of the vectors of its components, i.e., we define $c_{\mcl{M}} = a_1 c_{G_1} + \ldots + \ldots a_k c_{G_k} \in \R^{\ell}$.

The point of this is the following implication: if two well-conditioned mixtures get mapped to vectors which are close, then these two mixtures must be close in total variation distance.
The intuition is that when we write down the $L_1$ distance between the two mixtures, because the Taylor coefficients of Gaussians decay exponentially fast, the contribution of terms with degree $l > O(\log 1/\eps)$ to the integral becomes negligible.

Now, we can associate the set of well-conditioned mixtures with a convex body in $O(\log 1/\eps)$-dimensions, where the vertices of the convex body are given by single Gaussians. 
Consequently we can use Caratheodory's theorem to argue that any point within this body can be approximated as an $O(\log 1/\eps)$-sparse convex combination of the vertices, or equivalently, any well-conditioned mixture can we approximated by a mixture of $O(\log 1/\eps)$ well-conditioned Gaussians. 

It remains to demonstrate how to actually find this sparse mixture of Gaussians.
Naively, the number of vertices is infinite, as there are infinitely many well-conditioned Gaussians.
However, it is not too hard to show that if we consider a slight coarsening of this body by only taking the vertices to be the vectors associated to the well-conditioned Gaussians which belong in some $\poly (1 / \eps)$-sized net, then the quality of our solution only degrades by constant multiplicative factors.
At this point we can appeal to standard results in convex optimization to find the desired sparse approximation.
We defer the details of this argument to Section~\ref{sec:GMM-well-conditioned}.

\subsection{Localization}

After solving the well-conditioned case, the next step is to reduce the general case to the well-conditioned case via localization.  
We begin with an important definition.
\begin{definition}[Gaussian Multiplier]
For parameters $\mu, \sigma$, we define 
\[
M_{\mu, \sigma^2}(x) = e^{-\frac{(x - \mu)^2}{2 \sigma^2}}
\]
i.e. it is a Gaussian scaled so that its maximum value is $1$.
\end{definition}
Gaussian multipliers will be crucial in the localization step.  Now assume that $f$ can be written as some unknown $k$-sparse combination, say
\[
f = a_1G_1 + \dots + a_k G_k 
\]
We can then modify $f$, e.g. by multiplying by a Gaussian multiplier $M_{\mu, \sigma^2}$.
Heuristically, this operation changes the coefficients $a_1, \ldots, a_k$ in a predictable way.  Namely, the coefficients $a_j$ of Gaussians $G_j$ that are far from $N(\mu, \sigma^2)$ are exponentially attenuated based on the distance to $N(\mu, \sigma^2)$.  This effectively ''localizes'' the mixture.  More formally, 
\begin{claim}\label{claim:mult-by-gaussian-intro}
We have the identity
 \[
 M_{\mu, \sigma^2}(x)N(\mu_1, \sigma_1^2) =   \frac{1}{\sqrt{1 + \frac{\sigma_1^2}{\sigma^2}}}e^{-\frac{(\mu_1 - \mu)^2}{2(\sigma_1^2 + \sigma^2)}} N\left( \frac{\mu \sigma_1^2 + \mu_1 \sigma^2}{\sigma_1^2 + \sigma^2} , \frac{\sigma_1^2 \sigma^2}{\sigma_1^2 + \sigma^2} \right) \,.
 \]
 \end{claim}
 \begin{proof}
 We prove the above through direct computation.
 \begin{align*}
  M_{\mu, \sigma^2}(x)G_1(x) &= e^{-\frac{(x - \mu)^2}{2\sigma^2} -\frac{(x - \mu_1)^2}{2\sigma_1^2}} \frac{1}{\sigma_1 \sqrt{2\pi }} = \frac{1}{\sigma_1 \sqrt{2\pi }} \cdot e^{-\frac{1}{2}\left( \left(\frac{1}{\sigma^2} + \frac{1}{\sigma_1^2} \right)x^2 - 2 \left( \frac{\mu}{\sigma^2} + \frac{\mu_1}{\sigma_1^2} \right)x + \frac{\mu^2}{\sigma^2} + \frac{\mu_1^2}{\sigma_1^2}\right)} \\ &= \frac{1}{\sigma_1 \sqrt{2\pi }} \cdot  \exp \left( -\frac{1}{2}\left( \sqrt{\frac{1}{\sigma^2} + \frac{1}{\sigma_1^2}}x - \frac{\frac{\mu}{\sigma^2} + \frac{\mu_1}{\sigma_1^2}}{\sqrt{\frac{1}{\sigma^2} + \frac{1}{\sigma_1^2}}} \right)^2 - \frac{1}{2} \cdot \frac{(\mu_1 - \mu)^2}{\sigma_1^2 + \sigma^2}  \right) \\ &= \frac{1}{\sqrt{1 + \frac{\sigma_1^2}{\sigma^2}}}e^{-\frac{(\mu_1 - \mu)^2}{2(\sigma_1^2 + \sigma^2)}} N\left( \frac{\mu \sigma_1^2 + \mu_1 \sigma^2}{\sigma_1^2 + \sigma^2} , \frac{\sigma_1^2 \sigma^2}{\sigma_1^2 + \sigma^2} \right) \,.
 \end{align*}
 \end{proof}
The hope is that this will leave us with only components that are not too far from each other \--- exactly the well-conditioned case which we already know how to solve.  If the variances of all of the components are comparable, then this is indeed the case.  However, additional complications arise when one of the components $G_i = N(\mu, \sigma_i^2)$ has variance $\sigma_i \ll \sigma$ because this component will still have much smaller variance than the others after localizing.  Nevertheless, we show that we can carefully localize at different scales, using smaller variance Gaussian multipliers to localize around smaller variance components so that all of the localized mixtures are well-conditioned.     

The main remaining question is to select a good family of localizations so that we can then fully reconstruct the original mixture from the localized mixtures.  Each localized mixture will cost us  $O(\log 1/\eps)$ components, and therefore we must use at most $\wt{O}(k)$ different localizations.  When all of the variances of the Gaussians are not too dissimilar, we can do so by leveraging the following structural result, which states that one can $\eps$-approximate the constant function using a sum of evenly spaced Gaussians with variance $1$ and spacing $(\log 1/\eps)^{-1/2}$ (or smaller).  The intuition behind this observation is that the Fourier transform of a Gaussian is also a Gaussian, which has exponential tail decay.
\begin{lemma}\label{lem:approx-constant-intro}
Let $0 < \eps < 0.1$ be a parameter.  Let $c$ be a real number such that $0 < c \leq (\log 1/\eps)^{-1/2}$.  Define
\[
f(x) = \sum_{j = - \infty}^{\infty} \frac{c}{\sqrt{2\pi}} M_{ c j \sigma , \sigma^2}(x) \,.
\]
Then $1 - \eps \leq f(x) \leq 1 + \eps$ for all $x$.
\end{lemma}
\begin{proof}
WLOG $\sigma = 1$.  Now the function $f$ is $c$-periodic and even, so we may consider its Fourier expansion 
\[
f(x) = a_0 + 2 a_1 \cos \left( \frac{2\pi  x}{c}\right) + 2a_2 \cos \left( \frac{4\pi  x}{c}\right)+ \dots 
\]
and we will now compute the Fourier coefficients.  First note that 
\[
a_0 = \frac{1}{c}\int_0^c f(x) dx = \frac{1}{\sqrt{2\pi}}\sum_{j = -\infty}^{\infty} \int_{c(j+1)}^{cj} M_{0, 1}(x) dx = 1 \,.
\]
Next, for any $j \geq 1$, 
\begin{align*}
a_j = \frac{1}{c} \int_0^c f(x) \cos \left( \frac{2\pi j x}{c}\right) dx = \frac{1}{\sqrt{2\pi}} \sum_{j = -\infty}^{\infty} \int_{c(j+1)}^{cj} M_{0, 1}(x) \cos \left( \frac{2\pi j x}{c}\right) dx \\ = \frac{1}{\sqrt{2\pi}} \int_{-\infty}^\infty \frac{1}{2}\left( e^{\frac{-x^2}{2} + \frac{2\pi i j x}{c}} + e^{\frac{-x^2}{2} - \frac{2\pi i j x}{c}} \right) dx = e^{-\frac{2\pi^2 j^2}{c^2}} 
\end{align*}
where in the above we use the notation $i = \sqrt{-1}$.  Using the assumption that $c \leq (\log 1/\eps)^{-1/2}$, it is clear that
\[
\sum_{j=1}^{\infty} e^{-\frac{2\pi^2 j^2}{c^2}} \leq \frac{\eps}{2}
\]
so we deduce that for any $x$,
\[
|f(x) - 1| \leq 2(|a_1| + |a_2| + \cdots ) = 2\sum_{j=1}^{\infty} e^{-\frac{2\pi^2 j^2}{c^2}}  \leq \eps \,.
\]
In other words, the function $f$ is between $1 - \eps$ and $1+\eps$ everywhere and we are done.
\end{proof}

In light of the above lemma, we can use a set of evenly spaced Gaussian multipliers and simply sum the different localized mixtures.  Note that it suffices to use $\wt{O}(k)$ different localizations because we only need to sum over the Gaussian multipliers that have some nontrivial overlap with one of the $k$ true components (since for Gaussian multipliers that are far from all of the components, the localized mixture will be approximately $0$).

To handle the fully general case, when the variances of the Gaussians are unbounded, we need a generalization of the previous lemma that allows us to $\eps$-approximate the indicator function of an interval with a sum of $O(\log^2 1/\eps)$ Gaussians.  The proof of this generalization is in Section~\ref{sec:approx-with-Gaussians}.
\begin{definition}[Significant Interval]
For a Gaussian multiplier $M_{\mu, \sigma^2}$, we say the $C$-significant interval of $M$ is $[\mu - C \sigma, \mu +  C \sigma  ]$.  We will use the same terminology for a Gaussian $N(\mu, \sigma^2)$.
\end{definition}
\begin{theorem}
Let $l$ be a positive real number and $0 < \eps < 0.1$ be a parameter.  There is a function $f$ with the following properties
\begin{enumerate}
    \item $f$ can be written a linear combination of Gaussian multipliers
    \[
    f(x) = w_1M_{\mu_1, \sigma_1^2}(x) + \dots  + w_nM_{\mu_n, \sigma_n^2}(x)
    \]
    where $n = O(\log^2 1/\eps)$ and $0 \leq w_1, \dots , w_n \leq 1$ 
    \item The $10 \sqrt{\log 1/\eps}$-significant intervals of all of the $M_{\mu_i, \sigma_i^2}$ are contained in the interval $[-(1 + \eps)l, (1 + \eps)l]$
    \item $0 \leq f(x) \leq 1 + \eps$ for all $x$
    \item $1 - \eps \leq f(x) \leq 1 + \eps$ for all $x$ in the interval $[-l, l]$
    \item $0 \leq f(x) \leq \eps$ for $x \geq (1+ \eps)l $ and $x \leq -(1 + \eps)l$
    
\end{enumerate}
\end{theorem}
We combine this structural result with a dynamic program which allows us to efficiently choose the scales at which to localize. Putting all of these pieces together yields our full algorithm, assuming we have access to the pdf of the unknown function.  We show how to eliminate the need for pdf access below and present our full algorithm in complete detail in Section~\ref{sec:GMM-full}.  


\subsection{Abstracting away the samples}\label{sec:abstract-away-samples}

In the previous sections, we have assumed that we have access to the underlying pdf function $f$.
Typically, however, we only have sample access to the unknown distribution.
To rectify this, we will use the improper learner in \citep{chan2013efficient} (see Theorem 37) whose output is a piecewise polynomial. 
We can then only work with this piecewise polynomial, which is an explicit function that we can then perform explicit computations with. 
\begin{definition}
A function $f$ is $t$-piecewise degree $d$ if there is a partition of the real line into intervals $I_1, \dots , I_t$ and polynomials $q_1(x), \dots , q_t(x)$ of degree at most $d$ such that for all $i \in [t]$, $f(x) = q_i(x)$ on the interval $I_i$.
\end{definition}
The work in \citep{chan2013efficient} guarantees to learn a piecewise polynomial $f'$ that is close to $\mcl{M}$ in $L^1$ distance when given $\wt{O}(k/\eps^2)$ samples (and they also show that this sample complexity is essentially optimal).

\begin{theorem}[\cite{chan2013efficient}]\label{thm:improper-GMMs}
Let $\mcl{M} = w_1G_1 + \dots + w_kG_k$ be an unknown mixture of Gaussians and $f$ a distribution such that $d_{\TV}(f, \mcl{M}) \leq \eps$.  There is an algorithm that, given $\wt{O}(k/\eps^2)$ samples from $f$, runs in $\poly(k/\eps)$ time and returns an $O(k)$-piecewise degree $O(\log 1/\eps)$ function $f'$ such that with $0.9$ probability (over the random samples),
\[
\norm{f' - f}_1 \le O(\eps) \,.
\]
\end{theorem}

For technical reasons, we will need a few simple post-processing steps after using Theorem \ref{thm:improper-GMMs}.  We can ensure that the output hypothesis $f'$ is always nonnegative by splitting each polynomial into positive and negative parts and zeroing out the negative parts (since this will not increase the $L^1$ error).  Finally, we can re-normalize so that the output $f'$ is actually a distribution.  This renormalization at most doubles the $L^1$ error.  Thus we have:
\begin{corollary}\label{coro:improper-learner}
Let $\mcl{M} = w_1G_1 + \dots + w_kG_k$ be an unknown mixture of Gaussians and $f$ a distribution such that $d_{\TV}(f, \mcl{M}) \leq \eps$.  There is an algorithm that, given $\wt{O}(k/\eps^2)$ samples from $\mcl{D}$, runs in $\poly(k/\eps)$ time and returns an $O(k \log 1/\eps)$-piecewise degree $O(\log 1/\eps)$ function $f'$ such that $f'$ is a distribution and with $0.9$ probability (over the random samples),
\[
d_{\TV}(f, f') \le O(\eps) \,.
\]
\end{corollary}

\subsection{Hypothesis testing for model selection}

We now show how our result for model order selection, Theorem \ref{coro:gmm-hypothesis-test}, follows immediately from combining Theorem \ref{thm:main-gmm} with a standard procedure for testing the TV-distance between two distributions from samples (see \cite{yatracos1985rates}). 
\begin{claim}\label{claim:testL1}
Let $\mcl{D}_1, \mcl{D}_2$ be two distributions for which we have explicitly computable  density functions.  Let $\eps, \tau > 0$ be parameters.  Assume that we are given $O(1/\eps^2 \cdot \log 1/\tau)$ samples from $\mcl{D}_1$ and can efficiently sample from $\mcl{D}_2$.  Then in $\poly(1/\eps \log 1/\tau)$ time, we can compute $d$ such that with probability $1 - \tau$,
\[
| d - d_{\TV}(\mcl{D}_1, \mcl{D}_2) | \leq \eps \,.
\]
\end{claim}
\begin{proof}[Proof of Theorem \ref{coro:gmm-hypothesis-test}]
We can run the algorithm in Theorem \ref{thm:main-gmm} with parameters $k, \eps$ to obtain an output distribution $\wt{\mcl{M}}$ that is a mixture of $\wt{O}(k)$ Gaussians.  We can then use Claim \ref{claim:testL1} with parameters $\eps, 0.01$ to measure the TV-distance between $\wt{\mcl{M}}$ and $\mcl{D}$ (note that we have explicit access to the pdf of $\mcl{D}$) and output $1$ or $2$ depending on if our estimate of the TV distance is less than $\wt{O}(\eps)$.  Combining the guarantees of Theorem \ref{thm:main-gmm} and Claim \ref{claim:testL1} ensures that our output satisfies the desired properties.
\end{proof}

\subsection{Sparse Fourier}

We now briefly describe how our techniques can be used for sparse Fourier reconstruction.
Recall that the problem is to, given query access to a function $f$ on $[-1, 1]$ which is approximately $k$-Fourier sparse, approximate it with an $\wt{O}(k)$-Fourier sparse function. 
As before, we first abstract away the query access, by leveraging the following result from~\cite{chen2016fourier}:
\begin{theorem}[Theorem 1.1 in \cite{chen2016fourier}]\label{thm:sparse-fourier-old}
Let $f$ be a function defined on $[-1,1]$ and assume we are given query access to $f$.  Let $\mcl{M}$ be a function that is $(k,1)$-simple and has frequencies in the interval $[-F, F]$.  Then for any desired accuracy $\eps$, in $\poly(k, \log 1/\eps)  \log F  $ samples and $\poly(k, \log 1/\eps)  \log^2 F$ time, we can output a function $f'$ such that with probability $1 - 2^{-\Omega(k)}$,
\begin{enumerate}
    \item $f'$ is $(\poly(k, \log 1/\eps), \exp( \poly(k, \log 1/\eps)) )$- simple
    \item 
    \[
    \int_{-1}^1 | f' - f|^2 \leq O\left( \eps^2 + \int_{-1}^1 |f - \mcl{M}|^2 \right) \,.
    \]
\end{enumerate}
\end{theorem}
\begin{remark}
While the bound on the coefficients of $f'$ is not explicitly stated in Theorem 1.1 in \cite{chen2016fourier}, it immediately follows from the proof.
\end{remark}

Our algorithm for postprocessing this into a $\wt{O}(k)$-Fourier sparse signal follows roughly the same steps as in the Gaussian case.
First, we show that in a certain ``well-conditioned'' regime, namely, when the frequencies are not too dissimilar, there is a signal using $O(\log 1 / \eps)$ frequencies which approximates the function. 
To handle the general case, we use localizations based on carefully chosen kernels to reduce every signal to a sum of well-conditioned signals (at least, approximately).

One important distinction between the GMM and sparse Fourier reconstruction setting we highlight is that in the latter, the goal is usually to have runtimes which scale logarithmically with $1 / \eps$, whereas in the GMM setting, $\poly (1 / \eps)$ sample complexity and thus runtime is unavoidable.
However, our naive method of solving the well-conditioned case required constructing a net of $\poly (1 / \eps)$ many Gaussians, and thus required $\poly (1 / \eps)$ runtime.
To circumvent this difficulty, we demonstrate that in fact this can be improved, and that by being more careful, and choosing the (much smaller) set of vertices based on the Chebyshev points, we can in fact improve this runtime significantly.
See Sections~\ref{sec:fourier-well-conditioned} and~\ref{sec:sparse-fourier-full} for a full treatment of our algorithm.

\subsection{Paper organization}

The remainder of the paper will be devoted to proving Theorem~\ref{thm:main-gmm}, our main result for GMMs and Theorem~\ref{thm:fourier-main}, our main result for sparse Fourier reconstruction.  Due to space constraints, the remaining parts are deferred to the appendix.  We first present the proof of our result for GMMs.  In Section~\ref{sec:GMM-well-conditioned}, we deal with the well-conditioned case.  In Section~\ref{sec:approx-with-Gaussians}, we present some tools for localization which we will then use in Section~\ref{sec:GMM-full} to prove our full result for GMMs.  We then present the proof of our main result for sparse Fourier reconstruction which follows a very similar outline. We deal with the well-conditioned case in Section \ref{sec:fourier-well-conditioned} and then the general case in Section~\ref{sec:sparse-fourier-full}.  Appendix \ref{appendix:basic} contains several basic tools that will be used throughout the paper.


\begin{ack}
AL was supported in part by an NSF Graduate Research Fellowship and a Fannie and John Hertz Foundation Fellowship.  AM was supported in part by a Microsoft Trustworthy AI Grant, NSF CAREER Award CCF-1453261, NSF Large CCF1565235, a David and Lucile Packard Fellowship and an ONR Young Investigator Award.
\end{ack}

\section{Ethics and broader impact}\label{sec:ethics}
Our work is purely theoretical and we do not think there are any ethical issues or potential negative societal impacts.

\bibliographystyle{plainnat}
\bibliography{bibliography}

\begin{thebibliography}{76}
\providecommand{\natexlab}[1]{#1}
\providecommand{\url}[1]{\texttt{#1}}
\expandafter\ifx\csname urlstyle\endcsname\relax
  \providecommand{\doi}[1]{doi: #1}\else
  \providecommand{\doi}{doi: \begingroup \urlstyle{rm}\Url}\fi

\bibitem[Acharya et~al.(2014)Acharya, Jafarpour, Orlitsky, and
  Suresh]{acharya2014near}
Jayadev Acharya, Ashkan Jafarpour, Alon Orlitsky, and Ananda~Theertha Suresh.
\newblock Near-optimal-sample estimators for spherical gaussian mixtures.
\newblock \emph{arXiv preprint arXiv:1402.4746}, 2014.

\bibitem[Acharya et~al.(2017)Acharya, Diakonikolas, Li, and
  Schmidt]{acharya2017sample}
Jayadev Acharya, Ilias Diakonikolas, Jerry Li, and Ludwig Schmidt.
\newblock Sample-optimal density estimation in nearly-linear time.
\newblock In \emph{Proceedings of the Twenty-Eighth Annual ACM-SIAM Symposium
  on Discrete Algorithms}, pages 1278--1289. SIAM, 2017.

\bibitem[Achlioptas and McSherry(2005)]{achlioptas2005spectral}
Dimitris Achlioptas and Frank McSherry.
\newblock On spectral learning of mixtures of distributions.
\newblock In \emph{International Conference on Computational Learning Theory},
  pages 458--469. Springer, 2005.

\bibitem[Anderson et~al.(2014)Anderson, Belkin, Goyal, Rademacher, and
  Voss]{anderson2014more}
Joseph Anderson, Mikhail Belkin, Navin Goyal, Luis Rademacher, and James Voss.
\newblock The more, the merrier: the blessing of dimensionality for learning
  large gaussian mixtures.
\newblock In \emph{Conference on Learning Theory}, pages 1135--1164. PMLR,
  2014.

\bibitem[Arora and Kale(2007)]{AroraK07}
Sanjeev Arora and Satyen Kale.
\newblock A combinatorial, primal-dual approach to semidefinite programs.
\newblock In \emph{Proceedings of the 39th Annual {ACM} Symposium on Theory of
  Computing, San Diego, California, USA, June 11-13, 2007}, pages 227--236,
  2007.

\bibitem[Ashtiani et~al.(2018)Ashtiani, Ben-David, Harvey, Liaw, Mehrabian, and
  Plan]{ashtiani2018nearly}
Hassan Ashtiani, Shai Ben-David, Nicholas~JA Harvey, Christopher Liaw, Abbas
  Mehrabian, and Yaniv Plan.
\newblock Nearly tight sample complexity bounds for learning mixtures of
  gaussians via sample compression schemes.
\newblock In \emph{Proceedings of the 32nd International Conference on Neural
  Information Processing Systems}, pages 3416--3425, 2018.

\bibitem[Avron et~al.(2019)Avron, Kapralov, Musco, Musco, Velingker, and
  Zandieh]{avron2019universal}
Haim Avron, Michael Kapralov, Cameron Musco, Christopher Musco, Ameya
  Velingker, and Amir Zandieh.
\newblock A universal sampling method for reconstructing signals with simple
  fourier transforms.
\newblock In \emph{Proceedings of the 51st Annual ACM SIGACT Symposium on
  Theory of Computing}, pages 1051--1063, 2019.

\bibitem[Awasthi and Sheffet(2012)]{awasthi2012improved}
Pranjal Awasthi and Or~Sheffet.
\newblock Improved spectral-norm bounds for clustering.
\newblock In \emph{Approximation, Randomization, and Combinatorial
  Optimization. Algorithms and Techniques}, pages 37--49. Springer, 2012.

\bibitem[Bakshi et~al.(2020)Bakshi, Diakonikolas, Jia, Kane, Kothari, and
  Vempala]{bakshi2020robustly}
Ainesh Bakshi, Ilias Diakonikolas, He~Jia, Daniel~M Kane, Pravesh~K Kothari,
  and Santosh~S Vempala.
\newblock Robustly learning mixtures of $ k $ arbitrary gaussians.
\newblock \emph{arXiv preprint arXiv:2012.02119}, 2020.

\bibitem[Batenkov et~al.(2020)Batenkov, Demanet, Goldman, and
  Yomdin]{batenkov2020conditioning}
Dmitry Batenkov, Laurent Demanet, Gil Goldman, and Yosef Yomdin.
\newblock Conditioning of partial nonuniform fourier matrices with clustered
  nodes.
\newblock \emph{SIAM Journal on Matrix Analysis and Applications}, 41\penalty0
  (1):\penalty0 199--220, 2020.

\bibitem[Belkin and Sinha(2015)]{belkin2015polynomial}
Mikhail Belkin and Kaushik Sinha.
\newblock Polynomial learning of distribution families.
\newblock \emph{SIAM Journal on Computing}, 44\penalty0 (4):\penalty0 889--911,
  2015.

\bibitem[Bhaskara et~al.(2014)Bhaskara, Charikar, Moitra, and
  Vijayaraghavan]{bhaskara2014smoothed}
Aditya Bhaskara, Moses Charikar, Ankur Moitra, and Aravindan Vijayaraghavan.
\newblock Smoothed analysis of tensor decompositions.
\newblock In \emph{Proceedings of the forty-sixth annual ACM symposium on
  Theory of computing}, pages 594--603, 2014.

\bibitem[Bhaskara et~al.(2015)Bhaskara, Suresh, and
  Zadimoghaddam]{bhaskara2015sparse}
Aditya Bhaskara, Ananda Suresh, and Morteza Zadimoghaddam.
\newblock Sparse solutions to nonnegative linear systems and applications.
\newblock In \emph{Artificial Intelligence and Statistics}, pages 83--92. PMLR,
  2015.

\bibitem[Boufounos et~al.(2015)Boufounos, Cevher, Gilbert, Li, and
  Strauss]{boufounos2015s}
Petros Boufounos, Volkan Cevher, Anna~C Gilbert, Yi~Li, and Martin~J Strauss.
\newblock What’s the frequency, kenneth?: Sublinear fourier sampling off the
  grid.
\newblock \emph{Algorithmica}, 73\penalty0 (2):\penalty0 261--288, 2015.

\bibitem[Bresler and Macovski(1986)]{bresler1986exact}
Yoram Bresler and Albert Macovski.
\newblock Exact maximum likelihood parameter estimation of superimposed
  exponential signals in noise.
\newblock \emph{IEEE Transactions on Acoustics, Speech, and Signal Processing},
  34\penalty0 (5):\penalty0 1081--1089, 1986.

\bibitem[Cand{\`e}s and Fernandez-Granda(2014)]{candes2014towards}
Emmanuel~J Cand{\`e}s and Carlos Fernandez-Granda.
\newblock Towards a mathematical theory of super-resolution.
\newblock \emph{Communications on pure and applied Mathematics}, 67\penalty0
  (6):\penalty0 906--956, 2014.

\bibitem[Chan et~al.(2013)Chan, Diakonikolas, Servedio, and
  Sun]{chan2013efficient}
Siu-On Chan, Ilias Diakonikolas, Rocco~A. Servedio, and Xiaorui Sun.
\newblock Efficient density estimation via piecewise polynomial approximation,
  2013.

\bibitem[Chan et~al.(2014)Chan, Diakonikolas, Servedio, and
  Sun]{chan2014efficient}
Siu-On Chan, Ilias Diakonikolas, Rocco~A Servedio, and Xiaorui Sun.
\newblock Efficient density estimation via piecewise polynomial approximation.
\newblock In \emph{Proceedings of the forty-sixth annual ACM symposium on
  Theory of computing}, pages 604--613, 2014.

\bibitem[Chen et~al.(2004)Chen, Chen, and Kalbfleisch]{chen2004testing}
Hanfeng Chen, Jiahua Chen, and John~D Kalbfleisch.
\newblock Testing for a finite mixture model with two components.
\newblock \emph{Journal of the Royal Statistical Society: Series B (Statistical
  Methodology)}, 66\penalty0 (1):\penalty0 95--115, 2004.

\bibitem[Chen and Price(2019)]{chen2019active}
Xue Chen and Eric Price.
\newblock Active regression via linear-sample sparsification.
\newblock In \emph{Conference on Learning Theory}, pages 663--695. PMLR, 2019.

\bibitem[Chen et~al.(2016)Chen, Kane, Price, and Song]{chen2016fourier}
Xue Chen, Daniel~M Kane, Eric Price, and Zhao Song.
\newblock Fourier-sparse interpolation without a frequency gap.
\newblock In \emph{2016 IEEE 57th Annual Symposium on Foundations of Computer
  Science (FOCS)}, pages 741--750. IEEE, 2016.

\bibitem[Dasgupta(1999)]{dasgupta1999learning}
Sanjoy Dasgupta.
\newblock Learning mixtures of gaussians.
\newblock In \emph{40th Annual Symposium on Foundations of Computer Science
  (Cat. No. 99CB37039)}, pages 634--644. IEEE, 1999.

\bibitem[Dasgupta and Schulman(2007)]{dasgupta2007probabilistic}
Sanjoy Dasgupta and Leonard Schulman.
\newblock A probabilistic analysis of em for mixtures of separated, spherical
  gaussians.
\newblock \emph{Journal of Machine Learning Research}, 8\penalty0
  (Feb):\penalty0 203--226, 2007.

\bibitem[Daskalakis and Kamath(2014)]{daskalakis2014faster}
Constantinos Daskalakis and Gautam Kamath.
\newblock Faster and sample near-optimal algorithms for proper learning
  mixtures of gaussians.
\newblock In \emph{Conference on Learning Theory}, pages 1183--1213. PMLR,
  2014.

\bibitem[Devroye and Lugosi(2012)]{devroye2012combinatorial}
Luc Devroye and G{\'a}bor Lugosi.
\newblock \emph{Combinatorial methods in density estimation}.
\newblock Springer Science \& Business Media, 2012.

\bibitem[Diakonikolas et~al.(2007)Diakonikolas, Lee, Matulef, Onak, Rubinfeld,
  Servedio, and Wan]{diakonikolas2007testing}
Ilias Diakonikolas, Homin~K Lee, Kevin Matulef, Krzysztof Onak, Ronitt
  Rubinfeld, Rocco~A Servedio, and Andrew Wan.
\newblock Testing for concise representations.
\newblock In \emph{48th Annual IEEE Symposium on Foundations of Computer
  Science (FOCS'07)}, pages 549--558. IEEE, 2007.

\bibitem[Diakonikolas et~al.(2018)Diakonikolas, Kane, and
  Stewart]{DiakonikolasKS18}
Ilias Diakonikolas, Daniel~M. Kane, and Alistair Stewart.
\newblock List-decodable robust mean estimation and learning mixtures of
  spherical gaussians.
\newblock In \emph{Proceedings of the 50th Annual {ACM} {SIGACT} Symposium on
  Theory of Computing, {STOC} 2018, Los Angeles, CA, USA, June 25-29, 2018},
  pages 1047--1060, 2018.

\bibitem[Diakonikolas et~al.(2020)Diakonikolas, Hopkins, Kane, and
  Karmalkar]{diakonikolas2020robustly}
Ilias Diakonikolas, Samuel~B Hopkins, Daniel Kane, and Sushrut Karmalkar.
\newblock Robustly learning any clusterable mixture of gaussians.
\newblock \emph{arXiv preprint arXiv:2005.06417}, 2020.

\bibitem[Duarte and Baraniuk(2013)]{duarte2013spectral}
Marco~F Duarte and Richard~G Baraniuk.
\newblock Spectral compressive sensing.
\newblock \emph{Applied and Computational Harmonic Analysis}, 35\penalty0
  (1):\penalty0 111--129, 2013.

\bibitem[Fannjiang and Liao(2012)]{fannjiang2012coherence}
Albert Fannjiang and Wenjing Liao.
\newblock Coherence pattern--guided compressive sensing with unresolved grids.
\newblock \emph{SIAM Journal on Imaging Sciences}, 5\penalty0 (1):\penalty0
  179--202, 2012.

\bibitem[Feldman et~al.(2006)Feldman, Servedio, and
  O’Donnell]{feldman2006pac}
Jon Feldman, Rocco~A Servedio, and Ryan O’Donnell.
\newblock Pac learning axis-aligned mixtures of gaussians with no separation
  assumption.
\newblock In \emph{International Conference on Computational Learning Theory},
  pages 20--34. Springer, 2006.

\bibitem[Ge et~al.(2015)Ge, Huang, and Kakade]{ge2015learning}
Rong Ge, Qingqing Huang, and Sham~M Kakade.
\newblock Learning mixtures of gaussians in high dimensions.
\newblock In \emph{Proceedings of the forty-seventh annual ACM symposium on
  Theory of computing}, pages 761--770, 2015.

\bibitem[Ghosh and Sen(1984)]{ghosh1984asymptotic}
Jayanta~K Ghosh and Pranab~Kumar Sen.
\newblock On the asymptotic performance of the log likelihood ratio statistic
  for the mixture model and related results.
\newblock Technical report, North Carolina State University. Dept. of
  Statistics, 1984.

\bibitem[Guruswami and Zuckerman(2016)]{guruswami2016robust}
Venkatesan Guruswami and David Zuckerman.
\newblock Robust fourier and polynomial curve fitting.
\newblock In \emph{2016 IEEE 57th Annual Symposium on Foundations of Computer
  Science (FOCS)}, pages 751--759. IEEE, 2016.

\bibitem[Han et~al.(2016)Han, Jiao, and Weissman]{HanJW16}
Yanjun Han, Jiantao Jiao, and Tsachy Weissman.
\newblock Minimax rate-optimal estimation of divergences between discrete
  distributions.
\newblock In \emph{Proceedings of the 2016 International Symposium on
  Information Theory and Its Applications}, ISITA '16, pages 256--260,
  Washington, DC, USA, 2016. IEEE Computer Society.

\bibitem[Hardt and Price(2015)]{hardt2015tight}
Moritz Hardt and Eric Price.
\newblock Tight bounds for learning a mixture of two gaussians.
\newblock In \emph{Proceedings of the forty-seventh annual ACM symposium on
  Theory of computing}, pages 753--760, 2015.

\bibitem[Hartigan(1985)]{hartigan1985failure}
JA~Hartigan.
\newblock A failure of likelihood asymptotics for normal mixtures.
\newblock In \emph{Proc. Barkeley Conference in Honor of J. Neyman and J.
  Kiefer}, volume~2, pages 807--810, 1985.

\bibitem[Hopkins and Li(2018)]{hopkins2018mixture}
Samuel~B Hopkins and Jerry Li.
\newblock Mixture models, robustness, and sum of squares proofs.
\newblock In \emph{Proceedings of the 50th Annual ACM SIGACT Symposium on
  Theory of Computing}, pages 1021--1034, 2018.

\bibitem[Hsu and Kakade(2013)]{hsu2013learning}
Daniel Hsu and Sham~M Kakade.
\newblock Learning mixtures of spherical gaussians: moment methods and spectral
  decompositions.
\newblock In \emph{Proceedings of the 4th conference on Innovations in
  Theoretical Computer Science}, pages 11--20, 2013.

\bibitem[Huang and Kakade(2015)]{huang2015super}
Qingqing Huang and Sham~M Kakade.
\newblock Super-resolution off the grid.
\newblock \emph{arXiv preprint arXiv:1509.07943}, 2015.

\bibitem[Huang et~al.(2017)Huang, Peng, and Zhang]{huang2017model}
Tao Huang, Heng Peng, and Kun Zhang.
\newblock Model selection for gaussian mixture models.
\newblock \emph{Statistica Sinica}, pages 147--169, 2017.

\bibitem[Iyer et~al.(2021)Iyer, Tal, and Whitmeyer]{iyer2021junta}
Vishnu Iyer, Avishay Tal, and Michael Whitmeyer.
\newblock Junta distance approximation with sub-exponential queries.
\newblock In \emph{Electron. Colloquium Comput. Complex.}, volume~28, page~4,
  2021.

\bibitem[Jeffries(2003)]{jeffries2003note}
Neal~O Jeffries.
\newblock A note on ‘testing the number of components in a normal mixture’.
\newblock \emph{Biometrika}, 90\penalty0 (4):\penalty0 991--994, 2003.

\bibitem[Jiao et~al.(2016)Jiao, Han, and Weissman]{JiaoHW16}
Jiantao Jiao, Yanjun Han, and Tsachy Weissman.
\newblock Minimax estimation of the $\ell_1$ distance.
\newblock In \emph{Proceedings of the 2016 IEEE International Symposium on
  Information Theory}, ISIT '16, pages 750--754, Washington, DC, USA, 2016.
  IEEE Computer Society.

\bibitem[Jiao et~al.(2017)Jiao, Venkat, Han, and Weissman]{JiaoVHW17}
Jiantao Jiao, Kartik Venkat, Yanjun Han, and Tsachy Weissman.
\newblock Minimax estimation of functionals of discrete distributions.
\newblock \emph{IEEE Transactions on Information Theory}, 61\penalty0
  (5):\penalty0 2835--2885, 2017.

\bibitem[Kalai et~al.(2010)Kalai, Moitra, and Valiant]{kalai2010efficiently}
Adam~Tauman Kalai, Ankur Moitra, and Gregory Valiant.
\newblock Efficiently learning mixtures of two gaussians.
\newblock In \emph{Proceedings of the forty-second ACM symposium on Theory of
  computing}, pages 553--562, 2010.

\bibitem[Kane(2021)]{kane2021robust}
Daniel~M Kane.
\newblock Robust learning of mixtures of gaussians.
\newblock In \emph{Proceedings of the 2021 ACM-SIAM Symposium on Discrete
  Algorithms (SODA)}, pages 1246--1258. SIAM, 2021.

\bibitem[Kasahara and Shimotsu(2015)]{kasahara2015testing}
Hiroyuki Kasahara and Katsumi Shimotsu.
\newblock Testing the number of components in normal mixture regression models.
\newblock \emph{Journal of the American Statistical Association}, 110\penalty0
  (512):\penalty0 1632--1645, 2015.

\bibitem[Kothari et~al.(2018)Kothari, Steinhardt, and
  Steurer]{kothari2018robust}
Pravesh~K Kothari, Jacob Steinhardt, and David Steurer.
\newblock Robust moment estimation and improved clustering via sum of squares.
\newblock In \emph{Proceedings of the 50th Annual ACM SIGACT Symposium on
  Theory of Computing}, pages 1035--1046, 2018.

\bibitem[Kumar and Kannan(2010)]{kumar2010clustering}
Amit Kumar and Ravindran Kannan.
\newblock Clustering with spectral norm and the k-means algorithm.
\newblock In \emph{2010 IEEE 51st Annual Symposium on Foundations of Computer
  Science}, pages 299--308. IEEE, 2010.

\bibitem[Li and Schmidt(2017)]{li2017robust}
Jerry Li and Ludwig Schmidt.
\newblock Robust and proper learning for mixtures of gaussians via systems of
  polynomial inequalities.
\newblock In \emph{Conference on Learning Theory}, pages 1302--1382. PMLR,
  2017.

\bibitem[Liu and Moitra(2020)]{liu2020settling}
Allen Liu and Ankur Moitra.
\newblock Settling the robust learnability of mixtures of gaussians.
\newblock \emph{arXiv preprint arXiv:2011.03622}, 2020.

\bibitem[Liu and Moitra(2021)]{liu2021learning}
Allen Liu and Ankur Moitra.
\newblock Learning gmms with nearly optimal robustness guarantees.
\newblock \emph{arXiv preprint arXiv:2104.09665}, 2021.

\bibitem[Lo et~al.(2001)Lo, Mendell, and Rubin]{lo2001testing}
Yungtai Lo, Nancy~R Mendell, and Donald~B Rubin.
\newblock Testing the number of components in a normal mixture.
\newblock \emph{Biometrika}, 88\penalty0 (3):\penalty0 767--778, 2001.

\bibitem[Mixon et~al.(2017)Mixon, Villar, and Ward]{mixon2017clustering}
Dustin~G Mixon, Soledad Villar, and Rachel Ward.
\newblock Clustering subgaussian mixtures by semidefinite programming.
\newblock \emph{Information and Inference: A Journal of the IMA}, 6\penalty0
  (4):\penalty0 389--415, 2017.

\bibitem[Moitra(2015)]{moitra2015super}
Ankur Moitra.
\newblock Super-resolution, extremal functions and the condition number of
  vandermonde matrices.
\newblock In \emph{Proceedings of the forty-seventh annual ACM symposium on
  Theory of computing}, pages 821--830, 2015.

\bibitem[Moitra and Valiant(2010)]{moitra2010settling}
Ankur Moitra and Gregory Valiant.
\newblock Settling the polynomial learnability of mixtures of gaussians.
\newblock In \emph{2010 IEEE 51st Annual Symposium on Foundations of Computer
  Science}, pages 93--102. IEEE, 2010.

\bibitem[Neyman and Scott(1966)]{neyman1966use}
J~Neyman and E~Scott.
\newblock On the use of c ($\alpha$) tests of composite hypotheses.
\newblock \emph{Bulletin de L’lnstitut International de Statistique},
  41:\penalty0 477--497, 1966.

\bibitem[Parnas et~al.(2006)Parnas, Ron, and Rubinfeld]{parnas2006tolerant}
Michal Parnas, Dana Ron, and Ronitt Rubinfeld.
\newblock Tolerant property testing and distance approximation.
\newblock \emph{Journal of Computer and System Sciences}, 72\penalty0
  (6):\penalty0 1012--1042, 2006.

\bibitem[Pearson(1894)]{pearson1894contributions}
Karl Pearson.
\newblock Contributions to the mathematical theory of evolution.
\newblock \emph{Philosophical Transactions of the Royal Society of London. A},
  185:\penalty0 71--110, 1894.

\bibitem[Polyanskiy and Wu(2020)]{polyanskiy2020self}
Yury Polyanskiy and Yihong Wu.
\newblock Self-regularizing property of nonparametric maximum likelihood
  estimator in mixture models, 2020.
\newblock URL \url{https://arxiv.org/abs/2008.08244}.

\bibitem[Price and Song(2015)]{price2015robust}
Eric Price and Zhao Song.
\newblock A robust sparse fourier transform in the continuous setting.
\newblock In \emph{2015 IEEE 56th Annual Symposium on Foundations of Computer
  Science}, pages 583--600. IEEE, 2015.

\bibitem[Prony(1795)]{pronyessai}
R~Prony.
\newblock Essai experimental et analytique sur les lois de la dilabilite des
  fluides elastiques et sur celles de la force expansive de la vapeur de l'eau
  et de la vapeur de palcool a differentes temperatures.
\newblock \emph{Journal de I'Ecole Poly technique}, pages 24--76, 1795.

\bibitem[Rivlin(2020)]{rivlin2020chebyshev}
Theodore~J Rivlin.
\newblock \emph{Chebyshev polynomials}.
\newblock Courier Dover Publications, 2020.

\bibitem[Roy et~al.(1986)Roy, Paulraj, and Kailath]{roy1986esprit}
Robert Roy, Arogyaswami Paulraj, and Thomas Kailath.
\newblock Esprit--a subspace rotation approach to estimation of parameters of
  cisoids in noise.
\newblock \emph{IEEE transactions on acoustics, speech, and signal processing},
  34\penalty0 (5):\penalty0 1340--1342, 1986.

\bibitem[Schmidt(1982)]{schmidt1982signal}
Ralph~Otto Schmidt.
\newblock \emph{A signal subspace approach to multiple emitter location and
  spectral estimation}.
\newblock Stanford University, 1982.

\bibitem[Tang et~al.(2013)Tang, Bhaskar, Shah, and Recht]{tang2013compressed}
Gongguo Tang, Badri~Narayan Bhaskar, Parikshit Shah, and Benjamin Recht.
\newblock Compressed sensing off the grid.
\newblock \emph{IEEE transactions on information theory}, 59\penalty0
  (11):\penalty0 7465--7490, 2013.

\bibitem[Tang et~al.(2014)Tang, Bhaskar, and Recht]{tang2014near}
Gongguo Tang, Badri~Narayan Bhaskar, and Benjamin Recht.
\newblock Near minimax line spectral estimation.
\newblock \emph{IEEE Transactions on Information Theory}, 61\penalty0
  (1):\penalty0 499--512, 2014.

\bibitem[Terrell and Scott(1992)]{terrell1992variable}
George~R Terrell and David~W Scott.
\newblock Variable kernel density estimation.
\newblock \emph{The Annals of Statistics}, pages 1236--1265, 1992.

\bibitem[Valiant and Valiant(2010{\natexlab{a}})]{ValiantV10a}
Gregory Valiant and Paul Valiant.
\newblock A {CLT} and tight lower bounds for estimating entropy.
\newblock \emph{Electronic Colloquium on Computational Complexity (ECCC)},
  17\penalty0 (179), 2010{\natexlab{a}}.

\bibitem[Valiant and Valiant(2010{\natexlab{b}})]{ValiantV10b}
Gregory Valiant and Paul Valiant.
\newblock Estimating the unseen: A sublinear-sample canonical estimator of
  distributions.
\newblock \emph{Electronic Colloquium on Computational Complexity (ECCC)},
  17\penalty0 (180), 2010{\natexlab{b}}.

\bibitem[Valiant and Valiant(2011)]{ValiantV11a}
Gregory Valiant and Paul Valiant.
\newblock Estimating the unseen: An $n/\log n$-sample estimator for entropy and
  support size, shown optimal via new {CLT}s.
\newblock In \emph{Proceedings of the 43rd Annual ACM Symposium on the Theory
  of Computing}, STOC '11, pages 685--694, New York, NY, USA, 2011. ACM.

\bibitem[Vempala and Wang(2004)]{vempala2004spectral}
Santosh Vempala and Grant Wang.
\newblock A spectral algorithm for learning mixture models.
\newblock \emph{Journal of Computer and System Sciences}, 68\penalty0
  (4):\penalty0 841--860, 2004.

\bibitem[Wu and Xie(2018)]{wu2018improved}
Xuan Wu and Changzhi Xie.
\newblock Improved algorithms for properly learning mixture of gaussians.
\newblock In \emph{National Conference of Theoretical Computer Science}, pages
  8--26. Springer, 2018.

\bibitem[Wu and Yang(2018)]{Wu2018optimal}
Yihong Wu and Pengkun Yang.
\newblock Optimal estimation of gaussian mixtures via denoised method of
  moments, 2018.
\newblock URL \url{https://arxiv.org/abs/1807.07237}.

\bibitem[Yatracos(1985)]{yatracos1985rates}
Yannis~G Yatracos.
\newblock Rates of convergence of minimum distance estimators and kolmogorov's
  entropy.
\newblock \emph{The Annals of Statistics}, pages 768--774, 1985.

\end{thebibliography}

\section*{Checklist}

\begin{enumerate}

\item For all authors...
\begin{enumerate}
  \item Do the main claims made in the abstract and introduction accurately reflect the paper's contributions and scope?
    \answerYes{}
  \item Did you describe the limitations of your work?
    \answerYes{}
  \item Did you discuss any potential negative societal impacts of your work?
    \answerYes{}
  \item Have you read the ethics review guidelines and ensured that your paper conforms to them?
    \answerYes{}
\end{enumerate}

\item If you are including theoretical results...
\begin{enumerate}
  \item Did you state the full set of assumptions of all theoretical results?
    \answerYes{}
        \item Did you include complete proofs of all theoretical results?
    \answerYes{In the appendix}
\end{enumerate}

\item If you ran experiments...
\begin{enumerate}
  \item Did you include the code, data, and instructions needed to reproduce the main experimental results (either in the supplemental material or as a URL)?
    \answerNA{}
  \item Did you specify all the training details (e.g., data splits, hyperparameters, how they were chosen)?
    \answerNA{}
        \item Did you report error bars (e.g., with respect to the random seed after running experiments multiple times)?
    \answerNA{}
        \item Did you include the total amount of compute and the type of resources used (e.g., type of GPUs, internal cluster, or cloud provider)?
    \answerNA{}
\end{enumerate}

\item If you are using existing assets (e.g., code, data, models) or curating/releasing new assets...
\begin{enumerate}
  \item If your work uses existing assets, did you cite the creators?
    \answerNA{}
  \item Did you mention the license of the assets?
    \answerNA{}
  \item Did you include any new assets either in the supplemental material or as a URL?
    \answerNA{}
  \item Did you discuss whether and how consent was obtained from people whose data you're using/curating?
    \answerNA{}
  \item Did you discuss whether the data you are using/curating contains personally identifiable information or offensive content?
    \answerNA{}
\end{enumerate}

\item If you used crowdsourcing or conducted research with human subjects...
\begin{enumerate}
  \item Did you include the full text of instructions given to participants and screenshots, if applicable?
    \answerNA{}
  \item Did you describe any potential participant risks, with links to Institutional Review Board (IRB) approvals, if applicable?
    \answerNA{}
  \item Did you include the estimated hourly wage paid to participants and the total amount spent on participant compensation?
    \answerNA{}
\end{enumerate}

\end{enumerate}

\newpage 

\textbf{\large{Supplementary Materials}}
\appendix
\section{Well-Conditioned Case: Learning GMMs}\label{sec:GMM-well-conditioned}

We now deal with learning well-conditioned GMMs. We begin by formally specifying the properties that we want the components of the mixture to have.  Roughly, we want the components to have comparable variances and the separation between their means cannot be too large compared to the variances.  This means that after applying a suitable linear transformation, the components are all not too far from the standard Gaussian $N(0,1)$.

\begin{definition}
We say a Gaussian where $G = N(\mu, \sigma^2)$ is $\delta$-well-conditioned if 
\begin{itemize}
    \item $\abs{\mu} \leq \delta$
    \item $|\sigma^2  - 1| \leq \delta$
\end{itemize}
We say a mixture of Gaussians $\mcl{M} = w_1 G_1 + \dots + w_kG_k$  is $\delta$-well-conditioned if all of the components $G_1, \dots , G_k$ are $\delta$-well-conditioned.
\end{definition}

We now state our learning result for well-conditioned mixtures.

\begin{lemma}\label{lem:learn-tight-mixture}
Let $\eps > 0$ be a parameter.  Assume we are given access to a distribution $f$ such that $d_{\TV}(f , \mcl{M}) \leq \eps$ where $\mcl{M} = w_1G_1 + \dots + w_kG_k$ is a $0.5$-well-conditioned mixture of Gaussians.  Then we can compute, in $\poly(1/\eps)$ time, a mixture $\wt{\mcl{M}}$ of at most  $O(\log 1/\eps)$ Gaussians such that $d_{\TV}(\mcl{M}, \wt{\mcl{M}}) \leq \wt{O}(\eps)$.
\end{lemma}
\begin{remark}
Note that in the well-conditioned case, the number of components in the mixture that we compute does not depend on $k$.
\end{remark}

Our algorithm for proving Lemma \ref{lem:learn-tight-mixture} can be broken down into two parts.  In the first part, we find a mixture of $\poly(1/\eps)$ Gaussians that approximates $f$.  We then show how to reduce this mixture of $\poly(1/\eps)$ Gaussians to $O(\log 1/\eps)$ Gaussians by using the Taylor series approximation to a Gaussian.

\begin{lemma}\label{lem:gaussian-learnmanycomponents}
Let $\eps > 0$ be a parameter.  Assume we are given access to a distribution $f$ such that $d_{\TV}(f , \mcl{M}) \leq \eps$ where $\mcl{M} = w_1G_1 + \dots + w_kG_k$ is a $0.5$-well-conditioned mixture of Gaussians.  Then we can compute, in $\poly(1/\eps)$ time, a mixture of at most  $O(1/\eps^2)$ Gaussians that is $\wt{O}(\eps)$-close to $\mcl{M}$ in TV distance.
\end{lemma}
\begin{proof}
First, let $\mcl{T}$ be the set of all $0.5$-well-conditioned Gaussians such that $\mu$ and $\sigma^2$ are integer multiples of $0.1\eps$.  Note $|\mcl{T}| = O(1/\eps^2)$.
\\\\
By rounding all of the Gaussians $G_1, \dots , G_k$ to the nearest element of $\mcl{T}$ (this increases our $L^1$ error by at most $\eps$), we may assume that all of the components $G_1, \dots , G_k$ are actually in $\mcl{T}$.  Now note that since $\norm{f - w_1G_1 - \dots - w_kG_k}_1 \leq 2\eps$, we have for all $x$,
\begin{equation}\label{eq:L1bound}
|\wh{f}(x) - w_1\wh{G_1}(x) - \dots - w_k\wh{G_k}(x) | \leq 2\eps 
\end{equation}
where $\wh{G_j}$ denotes taking the Fourier transform of the pdf of the Gaussian $G_j$.  Let $l = \lceil \log 1/\eps \rceil $.  We now have, 
\[
\int_{-l}^l |\wh{f}(x) - w_1\wh{G_1}(x) - \dots - w_k\wh{G_k}(x) |^2  dx  \leq O(l \eps^2) \,.
\]
Now let all of the Gaussians in $\mcl{T}$ be $G_1, \dots , G_m$ where $m = |\mcl{T}|$.  By Lemma \ref{lem:linear-regression} (and splitting into real and imaginary parts), we can compute in $\poly(1/\eps)$ time, nonnegative weights $\wt{w_1}, \dots , \wt{w_{m}}$ with $\wt{w_1} +  \dots + \wt{w_{m}} \leq 1 $ such that 
\[
\int_{-l}^l |\wh{f}(x) - \wt{w_1}\wh{G_1}(x) - \dots - \wt{w_m}\wh{G_m}(x) |^2  dx  \leq O(l \eps^2) 
\]
which by Cauchy Schwarz implies that 
\[
\int_{-l}^l |\wh{f}(x) - \wt{w_1}\wh{G_1}(x) - \dots - \wt{w_m}\wh{G_m}(x) | dx  \leq O(l\eps) \,.
\]
Now note that since all of the Gaussians $G_1, \dots , G_m$ are $0.5$-well-conditioned, their Fourier transforms $\wh{G_j}$ also decay rapidly away from $[-l,l]$ so combining the above with (\ref{eq:L1bound}), we deduce that
\[
\int_{-\infty}^{\infty} |(\wt{w_1}\wh{G_1}(x) - \dots - \wt{w_m}\wh{G_m}(x) ) - w_1\wh{G_1}(x) - \dots - w_k\wh{G_k}(x) | \leq O(l\eps) \,.
\]
From the Fourier transform of the above we then get for all $x$
\[
|\wt{w_1}G_1(x) + \dots  +\wt{w_m}G_m(x) - w_1G_1(x) - \dots - w_mG_m(x)| \leq O(l\eps)
\]
and since all of the Gaussians involved are $0.5$-well-conditioned, they all decay rapidly outside the interval $[-l,l]$ and we conclude
\[
\int_{-\infty}^{\infty} |\wt{w_1}G_1(x) + \dots  +\wt{w_m}G_m(x) - w_1G_1(x) - \dots - w_mG_m(x)| dx \leq O(l^2 \eps) \,.
\]
Finally, note that by the above, we must have $1 - O(l^2 \eps) \leq \wt{w_1} + \dots + \wt{w_m} \leq 1 + O(l^2 \eps)$ so rescaling to an actual mixture i.e. so that the weights $\wt{w_1} + \dots  +\wt{w_m} = 1$, will affect the above error by at most $O(l^2 \eps)$. Thus, we can output this mixture and we are done.
\end{proof}

Next, as an immediate consequence of Lemma \ref{lem:Gaussian-taylorseries}, a $0.5$-well-conditioned Gaussian can be well approximated by its Taylor expansion.

\begin{corollary}\label{coro:gaussian-taylorseries}
Let $G = N(\mu, \sigma^2)$ be a $0.5$-well-conditioned Gaussian.  Let $\eps > 0$ be some parameter and let $l = \lceil \log 1/\eps \rceil$.  Then we can compute a polynomial $P_G(x)$ of degree $(10l)^2$ such that for all $x \in [-l, l]$, 
\[
|G(x) - P_G(x)| \leq O(\eps) \,. 
\]
\end{corollary}
\begin{proof}
This follows immediately from using Lemma \ref{lem:Gaussian-taylorseries} and applying the appropriate linear transformation to the polynomial.
\end{proof}

We can now complete the proof of Lemma \ref{lem:learn-tight-mixture} by using Lemma \ref{lem:gaussian-learnmanycomponents} and then using Corollary \ref{coro:gaussian-taylorseries} and Caratheodory to reduce the number of components.

\begin{proof}[Proof of Lemma \ref{lem:learn-tight-mixture}]
By Lemma \ref{lem:gaussian-learnmanycomponents}, we can compute a mixture 
\[
\wt{\mcl{M}} = \wt{w_1} \wt{G_1} + \dots + \wt{w_m} \wt{G_m}
\]
such that $m = O(1/\eps^2)$ and
\[
\norm{\wt{\mcl{M}} - \mcl{M}}_1 \leq \wt{O}(\eps) \,.
\]
For each Gaussian $\wt{G_j}$, let $P_{\wt{G_j}}(x)$ be the polynomial computed in Lemma \ref{lem:Gaussian-taylorseries}.  Write
\[
P_{\wt{G_j}}(x) = a_{j,0} + a_{j,1}x + \dots + a_{j,(10l)^2} x^{(10l)^2} \,.
\]
Define the vector 
\[
v_j = ( a_{j,0} , a_{j,1},  \dots ,  a_{j,(10l)^2}) \,.
\]
Now the point $\wt{w_1}v_1 + \dots + \wt{w_m}v_m$ is in the convex hull of $v_1, \dots , v_m$.  By Caratheodory (since the space is $(10l)^2 + 1$-dimensional), it must be in the convex hull of some $(10l)^2 + 1$ of the vertices.  Thus, we can compute indices $i_0, \dots , i_{(10l)^2}$ and nonnegative weights $w_0', \dots , w_{(10l)^2}'$ summing to $1$ such that 
\[
\wt{w_1}v_1 + \dots + \wt{w_m}v_m = w_0'v_{i_0} + \dots +w_{(10l)^2}'v_{i_{(10l)^2}} \,.
\]
The above implies that for all $x$,
\[
\wt{w_1} P_{\wt{G_1}}(x) + \dots + \wt{w_m} P_{\wt{G_m}}(x) =  w_0' P_{\wt{G_{i_0}}}(x) + \dots + w_{(10l)^2}' P_{\wt{G_{i_{(10l)^2}}}}(x) \,.
\]
Now by Corollary \ref{coro:gaussian-taylorseries} and the fact that all of the Gaussians are $0.5$-well-conditioned, meaning that they decay rapidly outside of $[-l,l]$, we conclude that if we set 
\[
\mcl{M}' =  w_0'\wt{G_{i_0}} + \dots +  w_{(10l)^2}'\wt{G_{i_{(10l)^2}}}
\]
then 
\[
\norm{\wt{\mcl{M}} - \mcl{M}'}_1 \leq \wt{O}(\eps)
\]
and then we have
\[
\norm{\mcl{M} - \mcl{M}'}_1  \leq \norm{\wt{\mcl{M}} - \mcl{M}}_1 + \norm{\wt{\mcl{M}} - \mcl{M}'}_1 \leq \wt{O}(\eps)
\]
as desired.
\end{proof}

 We can slightly improve Lemma \ref{lem:learn-tight-mixture} to work even when we do not have a precise estimate of $d_{\TV}(f ,\mcl{M})$ since we can just repeatedly decrease our target accuracy until we cannot improve our accuracy further.  Recall that we can use Claim \ref{claim:testL1} to test the $L^1$ distance between two distributions.  We now have the following (slight) improvement of Lemma \ref{lem:learn-tight-mixture}.
\begin{corollary}\label{coro:learn-tight-mixture}
Let $\eps > 0$ be a parameter.  Let $\mcl{M} = w_1G_1 + \dots + w_kG_k$ be an unknown $0.5$-well-conditioned mixture of Gaussians.  Assume we are given access to a distribution $f$.  Then we can compute, in $\poly(1/\eps)$ time, a mixture $\wt{\mcl{M}}$ of at most  $O(\log 1/\eps)$ Gaussians such that with high probability, 
\[
d_{\TV}(f, \wt{\mcl{M}}) \leq \eps^2 + \poly(\log 1/\eps) d_{\TV}(f, \mcl{M}) \,. 
\]
\end{corollary}
\begin{proof}
We can simply start from $\eps' = 1$ and run the algorithm in Lemma \ref{lem:learn-tight-mixture} with parameter $\eps'$ and then estimate $d_{\TV}(f, \wt{\mcl{M}})$ using Claim \ref{claim:testL1}.  If $d_{\TV}(f, \wt{\mcl{M}}) \leq \eps' \poly(\log 1/\eps) $ then we can decrease $\eps'$ by a factor of $0.9$ and repeat.  Repeating this process and taking the smallest accuracy $\eps' \geq \eps^3$ for which the above check succeeds, we get (from the guarantee of Lemma \ref{lem:learn-tight-mixture}) that 
\[
d_{\TV}(f, \wt{\mcl{M}}) \leq \eps^2  + \poly(\log 1/\eps) d_{\TV}(f, \mcl{M})
\]
and we are done.
\end{proof}

\section{Function Approximations Using Gaussians}\label{sec:approx-with-Gaussians}

In this section, we present several results about approximating functions as a sum of Gaussians.  These results will be key building blocks in the localization steps of both of our algorithms.  The main result of this section, Theorem \ref{thm:approx-interval}, allows us to $\eps$-approximate the indicator function of an interval as a sum of $\poly(\log 1/\eps)$-Gaussians.


First, it will be convenient to renormalize Gaussians so that their maximum value is $1$.  After renormalization, we call them Gaussian multipliers.

\begin{definition}[Gaussian Multiplier]\label{def:gaussian-multiplier}
For parameters $\mu, \sigma$, we define 
\[
M_{\mu, \sigma^2}(x) = e^{-\frac{(x - \mu)^2}{2 \sigma^2}}
\]
i.e. it is a Gaussian scaled so that its maximum value is $1$.
\end{definition}

We also introduce the some additional terminology.

\begin{definition}[Significant Interval]\label{def:sig-interval}
For a Gaussian multiplier $M_{\mu, \sigma^2}$, we say the $C$-significant interval of $M$ is $[\mu - C \sigma, \mu +  C \sigma  ]$.  We will use the same terminology for a Gaussian $N(\mu, \sigma^2)$.
\end{definition}

It will be used repeatedly that for a Gaussian (or Gaussian multiplier), $1- \eps$-fraction of its mass is contained in its $O(\sqrt{\log 1/\eps})$-significant interval.  We now state the main result of this section about approximating the indicator function of an interval as a weighted sum of Gaussian multipliers.

\begin{theorem}\label{thm:approx-interval}
Let $l$ be a positive real number and $0 < \eps < 0.1 $ be a parameter.  There is a function $f$ with the following properties
\begin{enumerate}
    \item $f$ can be written a linear combination of Gaussian multipliers
    \[
    f(x) = w_1M_{\mu_1, \sigma_1^2}(x) + \dots  + w_nM_{\mu_n, \sigma_n^2}(x)
    \]
    where $n = O((\log 1/\eps)^2)$ and $0 \leq w_1, \dots , w_n \leq 1$ 
    \item The $10 \sqrt{\log 1/\eps}$-significant intervals of all of the $M_{\mu_i, \sigma_i^2}$ are contained in the interval $[-(1 + \eps)l, (1 + \eps)l]$
    \item $0 \leq f(x) \leq 1 + \eps$ for all $x$
    \item $1 - \eps \leq f(x) \leq 1 + \eps$ for all $x$ in the interval $[-l, l]$
    \item $0 \leq f(x) \leq \eps$ for $x \geq (1+ \eps)l $ and $x \leq -(1 + \eps)l$
    
\end{enumerate}
\end{theorem}

\subsection{Approximating a Constant Function}
Recall Lemma~\ref{lem:approx-constant-intro} (restated below) that allows us to approximate a constant function using an infinite sum of evenly spaced Gaussian multipliers.  

\begin{lemma}[Restated from Lemma~\ref{lem:approx-constant-intro}]\label{lem:approx-constant}
Let $0 < \eps < 0.1$ be a parameter.  Let $c$ be a real number such that $0 < c \leq (\log 1/\eps)^{-1/2}$.  Define
\[
f(x) = \sum_{j = - \infty}^{\infty} \frac{c}{\sqrt{2\pi}} M_{ c j \sigma , \sigma^2}(x) \,.
\]
Then $1 - \eps \leq f(x) \leq 1 + \eps$ for all $x$.
\end{lemma}

\subsection{Approximating an Interval}

The next step in the proof of Theorem \ref{thm:approx-interval} is to show how to approximate an interval using a finite number of Gaussian multipliers i.e. we need to show how to create the sharp transitions at the ends of the interval.  In light of Lemma \ref{lem:approx-constant}, we can create a function satisfying the last four properties by taking $\wt{O}((1/\eps)^2)$ evenly spaced Gaussians multipliers with standard deviation $\eps^2 l$.  However, this is too many components and we must reduce the number of components to $O(\log^2 1/\eps)$.  The way we do this is by merging most of these components (all but the ones on the ends) into fewer components with larger standard deviation.  We keep iterating this merging process and prove that we can eventually reduce the number of components to $O(\log^2 1/\eps)$.
\\\\
First, the following result is an immediate consequence of Lemma \ref{lem:approx-constant}.  It allows us to approximate a Gaussian with standard deviation $2\sigma$ as a weighted sum of Gaussians with standard deviation $\sigma$.
\begin{corollary}\label{coro:approx-wider-gaussian}
Let $\eps$ be a parameter.  Let $c$ be a real number such that $0 < c \leq 0.5 (\log 1/\eps)^{-1/2}$.  Let
\[
g(x) = \sum_{j = - \infty}^{\infty} \frac{\sqrt{2}c}{\sqrt{3\pi}} e^{-\frac{c^2j^2}{6}} M_{c j \sigma, \sigma^2}(x)  \,.
\]
Then for all $x$,
\[
(1-\eps) M_{0, 4\sigma^2}(x) \leq g(x) \leq (1  +\eps)M_{0, 4\sigma^2}(x) \,.
\]
\end{corollary}
\begin{proof}
Lemma \ref{lem:approx-constant} (with $c \leftarrow \frac{2}{\sqrt{3}}c, \sigma \leftarrow \frac{2}{\sqrt{3}} \sigma$) implies that the function 
\[
f(x) = \sum_{j = - \infty}^{\infty} \frac{\sqrt{2}c}{\sqrt{3\pi}} M_{ \frac{4}{3}c j \sigma , \frac{4}{3}\sigma^2}(x)
\]
is between $1 - \eps$ and $1+ \eps$ everywhere.  Now consider
\begin{align*}
f(x) \cdot M_{0, 4\sigma^2}(x) = \sum_{j = - \infty}^{\infty} \frac{\sqrt{2}c}{\sqrt{3\pi}} M_{ \frac{4}{3}c j \sigma , \frac{4}{3}\sigma^2}(x) M_{0, 4\sigma^2}(x) = \sum_{j = - \infty}^{\infty} \frac{\sqrt{2}c}{\sqrt{3\pi}} e^{ -\frac{x^2 + 3(x - \frac{4}{3}cj \sigma)^2}{8\sigma^2}} \\ = \sum_{j = - \infty}^{\infty} \frac{\sqrt{2}c}{\sqrt{3\pi}} e^{-\frac{c^2j^2}{6}} e^{-\frac{(x - cj \sigma)^2}{2\sigma^2}  } = \sum_{j = - \infty}^{\infty} \frac{\sqrt{2}c}{\sqrt{3\pi}} e^{-\frac{c^2j^2}{6}} M_{c j \sigma, \sigma^2}(x) \,.
\end{align*}

\end{proof}

In the next lemma, we show when given a sum of evenly spaced Gaussians with standard deviation $\sigma$, we can replace almost all of them (except for ones on the ends) with a sum of fewer evenly spaced Gaussians with standard deviation $2\sigma$.

\begin{lemma}\label{lem:widergaussian}
Let $\eps$ be a parameter.  Let $c$ be a real number such that $0 < c \leq 0.01 (\log 1/\eps)^{-1/2}$.  Let $b$ be a positive integer.  Consider the function
\[
f(x) = \sum_{j = 0}^{b}  \frac{c}{\sqrt{2 \pi}} M_{cj\sigma, \sigma^2}(x) \,.
\]
Let $C = \lceil 10^2 c^{-1} \log (1/\eps)^{1/2} \rceil$.  There is a function $g$ of the form
\[
g(x) = \sum_{j = 0}^{2C}  \frac{w_j c}{\sqrt{ 2\pi}} M_{cj\sigma, \sigma^2}(x)  + \sum_{j = b - 2C }^b \frac{w_j c}{\sqrt{2\pi}} M_{cj\sigma, \sigma^2}(x) + \sum_{j = \lfloor  C/2 \rfloor}^{\lceil (b-C) / 2 \rceil } \frac{c}{\sqrt{2\pi}} M_{2cj\sigma, 4\sigma^2}(x)
\]
where the $0 \leq w_0, \dots , w_{2C}, w_{b - 2C}, \dots , w_{k} \leq 1$ are weights and
\[
\norm{f - g}_{\infty} \leq \eps^{10} \,.
\]
\end{lemma}
\begin{proof}
Let $\eps' = \eps^{100}$.  By Corollary \ref{coro:approx-wider-gaussian}, for any real numbers $j, x$,
\[
\left \lvert M_{c j \sigma, 4\sigma^2}(x) -  \sum_{k = -\infty}^{\infty} \frac{\sqrt{2}c}{\sqrt{3\pi}} e^{-\frac{c^2k^2}{6}} M_{c(k + j) \sigma , \sigma^2}(x) \right \rvert \leq \eps'  M_{c j \sigma, 4\sigma^2}(x) \,.
\]
Now we use the above inequality on each term of the last sum in the expression for $g(x)$.
\begin{align*}
\left \lvert \sum_{j = \lfloor  C/2 \rfloor}^{\lceil (b-C) / 2 \rceil } \frac{c M_{2cj\sigma, 4\sigma^2}(x)}{\sqrt{2\pi}}  -  \sum_{j = \lfloor  C/2 \rfloor}^{\lceil (b-C) / 2 \rceil }\sum_{k = -\infty}^{\infty} \frac{c^2}{\pi \sqrt{3}} e^{-\frac{c^2k^2}{6}} M_{c(k + 2j) \sigma , \sigma^2}(x) \right \rvert \\ \leq \eps' \sum_{j = \lfloor  C/2 \rfloor}^{\lceil (b-C) / 2 \rceil }  \frac{c M_{2c j \sigma, 4\sigma^2}(x)}{\sqrt{2\pi}}  \leq 2\eps'
\end{align*}
where the last step follows from Lemma \ref{lem:approx-constant}.  Now we rewrite the second sum in the LHS above.  Let 
\begin{align*}
S(x) &= \sum_{j = \lfloor  C/2 \rfloor}^{\lceil (b-C) / 2 \rceil }\sum_{k = -\infty}^{\infty} \frac{c^2}{\pi \sqrt{3}} e^{-\frac{c^2k^2}{6}} M_{c(k + 2j) \sigma , \sigma^2}(x) \\&= \sum_{l = -\infty}^{\infty} \frac{c^2}{\pi \sqrt{3}}M_{c l\sigma, \sigma^2}(x) \sum_{j = \lfloor  C/2 \rfloor}^{\lceil (b-C) / 2 \rceil } e^{-\frac{c^2(l - 2j)^2}{6}} \,.
\end{align*}
Define 
\[
a_l = \sum_{j = \lfloor  C/2 \rfloor}^{\lceil (b-C) / 2 \rceil } e^{-\frac{c^2(l - 2j)^2}{6}} \,.
\]
First, by applying Lemma \ref{lem:approx-constant} with parameters $c \leftarrow \frac{2}{\sqrt{3}}c, \sigma \leftarrow \sqrt{3}c^{-1}$, we have that for all real numbers $l$,
\[
 \left \lvert \frac{\sqrt{3\pi}}{\sqrt{2}c} - \sum_{-\infty}^{\infty} e^{-\frac{c^2(l - 2j)^2}{6}} \right \rvert \leq \frac{\eps' \sqrt{3\pi}}{\sqrt{2}c} \,.
\]
By the way we chose $C$, we deduce that for all  integers $l$ with $2C \leq  l \leq b - 2C $,
\begin{equation}\label{eq:coeffbound1}
\left \lvert \frac{\sqrt{3\pi}}{\sqrt{2}c}  - a_l \right \rvert \leq (2\eps') \frac{\sqrt{3\pi}}{\sqrt{2}c} 
\end{equation}
for all integers $0 \leq l \leq 2C$, or $b - 2C  \leq l \leq b$,
\begin{equation}\label{eq:coeffbound2}
a_l \leq \frac{\sqrt{3\pi}}{\sqrt{2}c} (1 + 2\eps')
\end{equation}
and finally for all integers $l < 0$ or $l > b$,
\begin{equation}\label{eq:coeffbound3}
a_l \leq \frac{\sqrt{3\pi}}{\sqrt{2}c} (2\eps') \,.
\end{equation}
To obtain these inequalities, we simply use the fact that the terms in the sum
\[
\sum_{-\infty}^{\infty} e^{-\frac{c^2(l - 2j)^2}{6}} 
\]
decay exponentially when $j$ is far from $l/2$ so their total contribution is small.
\\\\
Now we can set $w_0, \dots , w_{2C}, w_{b - 2C}, \dots , w_{k}$ in the expresion for $g(x)$ as follows:  
\[
w_j = \max\left(0,  1 - \frac{\sqrt{2}c}{\sqrt{3\pi}}a_j  \right) \,.
\]
It is clear that all of these weights are between $0$ and $1$.  We now have that
\begin{align*}
\norm{f - g}_{\infty} &\leq 2\eps' + \norm{f(x) - \left(S(x) + \sum_{j = 0}^{2C}  \frac{w_j c}{\sqrt{ 2\pi}} M_{cj\sigma, \sigma^2}(x)  + \sum_{j = b - 2C }^b \frac{w_j c}{\sqrt{2\pi}} M_{cj\sigma, \sigma^2}(x) \right) }_{\infty} 
\end{align*}
The expression inside the norm on the RHS can be rewritten as 
\begin{align*}
&\sum_{l = 2C + 1}^{b - 2C - 1} \left( \frac{c}{\sqrt{2\pi}} - \frac{c^2}{\pi\sqrt{3}} a_l \right)M_{c l\sigma, \sigma^2}(x) +  \sum_{l = 0}^{2C} \left( \frac{c}{\sqrt{2\pi}} (1 - w_l) - \frac{c^2}{\pi\sqrt{3}} a_l \right)M_{c l\sigma, \sigma^2}(x) \\ &+  \sum_{l = b - 2C}^{b} \left( \frac{c}{\sqrt{2\pi}} (1 - w_l) - \frac{c^2}{\pi\sqrt{3}} a_l \right)M_{c l\sigma, \sigma^2}(x) + \sum_{l = -\infty}^{-1} - \frac{c^2}{\pi\sqrt{3}} a_l M_{c l\sigma, \sigma^2}(x)  \\ &+ \sum_{l = b + 1}^{\infty} - \frac{c^2}{\pi\sqrt{3}} a_l M_{c l\sigma, \sigma^2}(x)
\end{align*} 
and combining (\ref{eq:coeffbound1},\ref{eq:coeffbound2}, \ref{eq:coeffbound3}), we deduce that the above has $L^{\infty}$ norm at most
\[
\norm{ (10\eps') \sum_{l = -\infty}^{\infty} \frac{c}{\sqrt{2\pi}} M_{c l\sigma, \sigma^2}(x) }_{\infty} \leq 20\eps' \,.
\]
where we used Lemma \ref{lem:approx-constant}.  Thus, $\norm{f - g}_{\infty} \leq 22\eps'$ and we are done.
\end{proof}

We can now prove Theorem \ref{thm:approx-interval} by repeatedly applying Lemma \ref{lem:widergaussian}.

\begin{proof}[Proof of Theorem \ref{thm:approx-interval}]
Let $c = 0.01 (\log 1/\eps)^{-1/2}$.  Let $K = \lceil \frac{1 + 0.5\eps }{c\eps^2} \rceil$
\[
f_0(x) = \sum_{j = - K}^{K} \frac{c}{\sqrt{2\pi}} M_{c j \eps^2 l, \eps^4 l^2}(x) \,.
\]
Let $\eps' = \eps^{10}$.  Using Lemma \ref{lem:approx-constant}, (and basic tail decay properties of a Gaussian) we get that
\begin{itemize}
\item $0 \leq f_0(x) \leq 1 + \eps'$ for all $x$
    \item $1 - \eps' \leq f_0(x) \leq 1 + \eps'$ for all $x$ in the interval $[-l, l]$
    \item $0 \leq f_0(x) \leq \eps'$ for $x \geq (1+ \eps)l $ and $x \leq -(1 + \eps)l$
\end{itemize}
Now we can apply Lemma \ref{lem:widergaussian} to $f_0(x)$ to obtain 
\begin{align*}
f_1(x) = \sum_{j =  -K}^{-K + 2C}  \frac{w_j c}{\sqrt{ 2\pi}} M_{cj\eps^2l, \eps^4l^2}(x)  + \sum_{j = K - 2C }^{K } \frac{w_j c}{\sqrt{2\pi}} M_{cj\eps^2l, \eps^4l^2}(x) \\ + \sum_{j = -\lceil  (K - C)/2 \rceil}^{\lceil  (K - C)/2 \rceil} \frac{c}{\sqrt{2\pi}}M_{2cj\eps^2l, 4\eps^4l^2}(x)
\end{align*}
where $C = \lceil 10^2 c^{-1} \log (1/\eps)^{1/2} \rceil$, the $w_j$ are weights between $0$ and $1$, and 
\[
\norm{f_1 - f_0} \leq \eps' \,.
\]
Now we can apply Lemma \ref{lem:widergaussian} again on the last sum in the expression for $f_1$.  We have to do this at most $10 \log 1/\eps$ times before there are at most $O((\log 1/\eps)^2)$ components remaining.  It is clear that in this procedure, the $10 \sqrt{\log 1/\eps}$-significant intervals of all of the Gaussian multipliers always remains in $[-(1 + \eps)l, (1 + \eps)l]$.  Also, the total $L^{\infty}$ error incurred over all of the applications of Lemma \ref{lem:widergaussian} is at most $10\eps' \log 1/\eps \leq \eps^9$.  It is clear that all of the weights are always nonnegative and in the interval $[0,1]$.  Thus, the final function $f$ satisfies
\begin{itemize}
    \item $0 \leq f(x) \leq 1 + \eps$ for all $x$
    \item $1 - \eps \leq f(x) \leq 1 + \eps$ for all $x$ in the interval $[-l, l]$
    \item $0 \leq f(x) \leq \eps$ for $x \geq (1+ \eps)l $ and $x \leq -(1 + \eps)l$
\end{itemize}
and we are done.
\end{proof}

In light of Theorem \ref{thm:approx-interval}, we may make the following definition.
\begin{definition}\label{def:interval-function}
For parameters $\eps, l$, let $\mcl{I}_{\eps, l}$ denote the function computed in Theorem \ref{thm:approx-interval} for parameters $\eps, l/(1 + \eps)$ .  We will also use $\mcl{I}_{\eps, l}^{(a)}$ to denote the function  $\mcl{I}_{\eps, l}(x - a)$.
\end{definition}
\begin{remark}
We define $\mcl{I}_{\eps, l}$ as above because it will be convenient later to be able to say that the significant part of $\mcl{I}_{\eps, l}$ is contained in the interval $[-l , l]$.
\end{remark}

\section{Nearly-Properly Learning GMMs: Full Version}\label{sec:GMM-full}

In this section, we complete the proof of our main result for learning GMMs, Theorem \ref{thm:main-gmm}.  We localize the distribution by multiplying by a Gaussian multiplier $M_{\mu, \sigma^2}$.  Note that the product of two Gaussians is still a Gaussian so multiplying a GMM by a Gaussian multiplier results in a re-weighted mixture of Gaussians.  Roughly, we argue that the new weights on components of the mixture that are far away from the multiplier $M_{\mu, \sigma^2}$ are negligible so the resulting mixture is well-conditioned and we can then use Corollary \ref{coro:learn-tight-mixture} to reconstruct the localized distribution.   To reconstruct the entire distribution, we show that it suffices to sum together $\wt{O}(k)$ different localized reconstructions.

\subsection{Localizing with Gaussian Multipliers}

 Recall Claim~\ref{claim:mult-by-gaussian-intro} (restated below) which gives an explicit formula for what happens when we have a Gaussian $G_1 = N(\mu_1, \sigma_1^2)$ and we multiply it by a Gaussian multiplier $M_{\mu, \sigma^2}(x) $.
 
 \begin{claim}[Restated from Claim~\ref{claim:mult-by-gaussian-intro}] \label{claim:mult-by-gaussian}
We have the identity
 \[
 M_{\mu, \sigma^2}(x)N(\mu_1, \sigma_1^2) =   \frac{1}{\sqrt{1 + \frac{\sigma_1^2}{\sigma^2}}}e^{-\frac{(\mu_1 - \mu)^2}{2(\sigma_1^2 + \sigma^2)}} N\left( \frac{\mu \sigma_1^2 + \mu_1 \sigma^2}{\sigma_1^2 + \sigma^2} , \frac{\sigma_1^2 \sigma^2}{\sigma_1^2 + \sigma^2} \right) \,.
 \]
 \end{claim}
 
 \subsection{ Building Blocks}

We first consider reconstructing a GMM $\mcl{M} = w_1G_1 + \dots w_kG_k$ after multiplying by a Gaussian multiplier $M_{\mu, \sigma^2}$.  As a corollary of Claim \ref{claim:mult-by-gaussian}, we know that when the $C$-significant intervals (recall Definition \ref{def:sig-interval}) of a Gaussian $G_j$ and the multiplier $M_{\mu, \sigma^2}(x) $ are disjoint for large $C$, then the $L^1$ norm of their product is $e^{-\Omega(C^2)}$.  In particular this means that after multiplying by $M_{\mu, \sigma^2}$, the only components that remain relevant are those that have nontrivial overlap with the multiplier $M_{\mu, \sigma^2}$.  The only way these components will not form a well-conditioned mixture is if there is some $G_j$ that is very thin (i.e. $\sigma_j << \sigma$) and overlaps with $M_{\mu, \sigma^2}$.  As long as this doesn't happen, we can apply Corollary \ref{coro:learn-tight-mixture}.  We formalize this below. 
 
 
 \begin{corollary}\label{coro:reconstruct-multiplier}
 Let $\mcl{M} = w_1G_1 + \dots + w_kG_k$ be an arbitrary mixture of Gaussians where $G_i = N(\mu_i, \sigma_i^2)$.  Let $\eps > 0$ be some parameter and let $l = \lceil \sqrt{\log (1/\eps)} \rceil $.  Assume we are given access to a distribution $f$.  Let $M_{\mu, \sigma^2}$ be a Gaussian multiplier.  Assume that for all $i \in [k]$, either $\sigma_i \geq 4l\sigma$ or the $10 l$-significant intervals of $G_i$ and $M_{\mu, \sigma^2}$ do not intersect.  Then in $\poly(1/\eps)$ time and with high probability, we can compute a weighted sum $\wt{M}$ of at most $O(\log (1/\eps)) $ Gaussians such that 
 \[
 \norm{\wt{M} - M_{\mu, \sigma^2}f }_1 \leq \eps + \poly(\log (1/\eps)) \norm{ M_{\mu, \sigma^2} (\mcl{M} - f)}_1 \,.
 \]
 \end{corollary}
 \begin{proof}
 We compute $M_{\mu, \sigma^2}f $ and let $C  = \norm{M_{\mu, \sigma^2}f}_1$.  If $C \leq \eps$ then we may simple output $0$.  Otherwise, we will apply Corollary \ref{coro:learn-tight-mixture} on $M_{\mu, \sigma^2}f /C$ and multiply the result by $C$.  We must first verify the conditions of Corollary \ref{coro:learn-tight-mixture}.  Let $S \subset [k]$ be the indices such that the $10l$-significant intervals of $G_i$ and $M_{\mu, \sigma^2}$ intersect.  First for $i \notin S$, by Claim \ref{claim:mult-by-gaussian},
 \[
 \norm{G_i M_{\mu, \sigma^2}}_1 \leq e^{-\frac{(\mu_i - \mu)^2}{2(\sigma_i^2 + \sigma^2)}} \leq e^{-10l^2} \leq \eps^{10} \,.
 \]
  Let 
 \[
 \mcl{M}' = \sum_{i \in S} w_iG_i \,.
 \]
 Then we know
 \[
 \norm{ \frac{M_{\mu, \sigma^2}f}{C}  - \frac{M_{\mu, \sigma^2} \mcl{M}'}{C} }_1 \leq \eps^{9} + \frac{\norm{M_{\mu, \sigma^2}(f - \mcl{M} )}_1}{C}
 \]
 Next, for $i \in S$,
 \[
G_i M_{\mu, \sigma^2} = w_i' N\left( \frac{\mu \sigma_i^2 + \mu_i \sigma^2}{\sigma_i^2 + \sigma^2} , \frac{\sigma_i^2 \sigma^2}{\sigma_i^2 + \sigma^2} \right)
 \]
 for some weight $w_i'$ and since we must have $\sigma_i \geq 4l\sigma$, then
 \begin{align*}
 &\frac{\sigma^2}{2} \leq \frac{\sigma_i^2 \sigma^2}{\sigma_i^2 + \sigma^2} \leq \sigma^2 \\
 &\left\lvert \frac{\mu \sigma_i^2 + \mu_i \sigma^2}{\sigma_i^2 + \sigma^2} - \mu \right \rvert = \left \lvert  \frac{(\mu_i - \mu) \sigma^2}{\sigma_i^2 + \sigma^2} \right \rvert \leq \frac{l(\sigma + \sigma_i) \sigma^2}{\sigma_i^2 + \sigma^2} \leq \frac{\sigma}{2}
 \end{align*}
Let 
\[
\mcl{M}'' = \sum_{i \in S}\frac{w_i'}{\sum_{i \in S} w_i'} N\left( \frac{\mu \sigma_i^2 + \mu_i \sigma^2}{\sigma_i^2 + \sigma^2} , \frac{\sigma_i^2 \sigma^2}{\sigma_i^2 + \sigma^2} \right) \,.
\]
Then we deduce, since $\norm{M_{\mu, \sigma^2}f/C}_1 = 1$, that
\begin{align*}
\norm{ \frac{M_{\mu, \sigma^2}f}{C}  - \mcl{M}'' }_1 &\leq \norm{\mcl{M}'' - \frac{M_{\mu, \sigma^2} \mcl{M}'}{C} }_1  + \norm{ \frac{M_{\mu, \sigma^2}f}{C}  - \frac{M_{\mu, \sigma^2} \mcl{M}'}{C} }_1 \\& \leq  2 \norm{ \frac{M_{\mu, \sigma^2}f}{C}  - \frac{M_{\mu, \sigma^2} \mcl{M}'}{C} }_1 \\ &\leq \eps^{8} + 2\frac{\norm{M_{\mu, \sigma^2}(f - \mcl{M} )}_1}{C} 
\end{align*}
and further, after applying a suitable linear transformation (taking $(\mu, \sigma^2) \rightarrow (0,1)$) that the mixture $\mcl{M}''$ is $0.5$-well-conditioned.  Thus, we can apply Corollary \ref{coro:learn-tight-mixture} and compute a weighted sum of $O(\log (1/\eps))$ Gaussians, $\wt{\mcl{M}}$ such that 
\[
\norm{\frac{M_{\mu, \sigma^2}f}{C}  - \wt{\mcl{M}}}_1 \leq \poly(\log (1/\eps))\left( \eps^{8} +  \frac{\norm{M_{\mu, \sigma^2}(f - \mcl{M} )}_1}{C} \right) \,.
\]
Now we can simply output $C \wt{\mcl{M}}$ (which is still a weighted sum of $O(\log (1/\eps))$ Gaussians) and we are done.
\end{proof}

 Recall that Theorem \ref{thm:approx-interval} shows how to express an interval as a sum of Gaussian multipliers.  Combining Theorem \ref{thm:approx-interval} with Corollary \ref{coro:reconstruct-multiplier}, we show that we can approximate a GMM over an interval as long as the interval does not overlap with a component that is much thinner than it.
 
 \begin{lemma}\label{lem:reconstruct-interval}
 Let $\mcl{M} = w_1G_1 + \dots + w_kG_k$ be an arbitrary mixture of Gaussians where $G_i = N(\mu_i, \sigma_i^2)$.  Let $\eps > 0$ be some parameter and let $l = \lceil \sqrt{\log (1/\eps)} \rceil $.  Assume we are given access to a distribution $f$.  Let $I = [a,b]$ be an interval.  Assume that for all $i \in [k]$, either $\sigma_i \geq (b - a)$ or the $10 l$-significant interval of $G_i$ does not intersect $I$.  Then in $\poly(1/\eps)$ time and with high probability, we can compute a weighted sum $\wt{\mcl{M}}$ of at most $\poly(\log (1/\eps)) $ Gaussians such that 
 \[
 \norm{\wt{\mcl{M}} - f \cdot 1_{I} }_1 \leq  \poly(\log (1/\eps)) \left( \eps +  \norm{ 1_{I} (\mcl{M} - f)}_1 \right) 
 \]
 where $1_{I}$ denotes the indicator function of $I$.
 \end{lemma}
 \begin{proof}
 Consider the function $\mcl{I} = \mcl{I}_{\eps , (b-a)/2}^{(a + b)/2}$ (recall Definition \ref{def:interval-function}).  Now note that by Theorem \ref{thm:approx-interval}, $\mcl{I}$ can be written in the form
 \[
 \mcl{I}  = \wt{w_1} M_{\wt{\mu_1}, \wt{\sigma_1}^2} + \dots + \wt{w_n} M_{\wt{\mu_n}, \wt{\sigma_n}^2}
 \]
 where $n = O(\log^2 1/\eps)$.  Furthermore, for all $i \in [n]$, we have $0 \leq \wt{w_i} \leq 1$ and $\wt{\sigma_i} \leq (b - a)/(4l)$ and the $10l$-significant intervals of $M_{\wt{\mu_i}, \wt{\sigma_i}^2}$ are all contained in the interval $[a,b]$.  Thus we can apply Corollary \ref{coro:reconstruct-multiplier} on $M_{\wt{\mu_i}, \wt{\sigma_i}^2} f$ for all $i \in [n]$.  Adding the results with the corresponding weights $\wt{w_1}, \dots , \wt{w_n}$, we obtain a function $\wt{\mcl{M}}$ that is a weighted sum of at most $\poly(\log (1/\eps)) $ Gaussians such that
 \begin{align*}
 \norm{\wt{\mcl{M}} - f \mcl{I}}_1 &\leq \poly(\log (1/\eps)) \left( \eps +  \sum_{i = 1}^n \wt{w_i}\norm{ M_{\wt{\mu_i}, \wt{\sigma_i}^2} (\mcl{M} - f)}_1 \right) \\ & = \poly(\log (1/\eps)) \left( \eps +  \norm{ \mcl{I} (\mcl{M} - f)}_1 \right)  \,.
\end{align*}
 Thus,
 \begin{equation}\label{eq:l1bound1}
 \norm{\wt{\mcl{M}} - \mcl{M} \mcl{I}}_1 \leq \poly(\log (1/\eps)) \left( \eps +  \norm{ \mcl{I} (\mcl{M} - f)}_1 \right) \,.
 \end{equation}
 Next, by the properties in Theorem \ref{thm:approx-interval},
 \begin{align*}
 \norm{\mcl{M}(\mcl{I} - 1_I )}_1 &\leq \eps \int_{-\infty}^{\infty} \mcl{M} + \int_{a}^{a + \eps(b- a)} \mcl{M} + \int_{b - \eps(b - a)}^b \mcl{M} \\ &\leq \eps +  \left( \sum_{j=1}^k w_j  \int_{a}^{a + \eps(b- a)} G_j  + w_j\int_{b - \eps(b - a)}^b G_j \right)\,.
\end{align*}
Consider one of the component Gaussians $G_j$ where $j \in [k]$.  If the $10l$-significant interval of $G_j$ does not intersect $[a,b]$ then it is clear that the total mass of $G_j$ on the interval $[a,b]$ is at most $\eps$.  Otherwise, we know that the standard deviation of $G_j$ is at least $b-a$ which means that its mass on the set $[a, a+ \eps(b - a)] \cup [b - \eps(b - a), b]$ is at most $O(\eps)$.  Thus we conclude that 
\begin{equation}\label{eq:l1bound2}
 \norm{\mcl{M}(\mcl{I} - 1_I )}_1 \leq  O(\eps) \,.
 \end{equation}
Also note that by the properties in Theorem \ref{thm:approx-interval}
\begin{equation}\label{eq:l1bound3}
\norm{ \mcl{I} (\mcl{M} - f)}_1 \leq \eps \int_{-\infty}^{\infty} |\mcl{M} - f| + 2\int_{a}^b |\mcl{M} - f| \leq 2(\eps + \norm{1_I(\mcl{M} - f)}_1) \,.
\end{equation}
Putting together (\ref{eq:l1bound1}, \ref{eq:l1bound2}, \ref{eq:l1bound3}) , we conclude 
\[
\norm{\wt{\mcl{M}} - f \cdot 1_I}_1 \leq \norm{ 1_{I} (\mcl{M} - f)}_1  +  \norm{\wt{\mcl{M}} - \mcl{M}1_I}_1 \leq \poly(\log (1/\eps)) \left( \eps +  \norm{ 1_{I} (\mcl{M} - f)}_1 \right) 
\]
and we are done.
 \end{proof}

\subsection{Structural Properties}
Lemma \ref{lem:reconstruct-interval} allows us to reconstruct the unknown GMM $\mcl{M}$ over certain intervals.  However, it cannot be applied to an arbitrary interval (because an interval may overlap with a component that is too thin).  We will now prove several structural results that will imply that there exist $\wt{O}(k)$ intervals for which the conditions of Lemma \ref{lem:reconstruct-interval} are satisfied (i.e. these intervals do not overlap with components that are much thinner than themselves) and such that the union of these intervals contains most of the mass of $\mcl{M}$.  Then, to complete the proof of Theorem \ref{thm:main-gmm}, we show how to find such a set of $\wt{O}(k)$ intervals using a dynamic program.

First, we define a modified density function for a GMM $\mcl{M} = w_1G_1 + \dots + w_kG_k$ where we modify each component Gaussian by restricting it to its $10l$-significant interval (and making it $0$ outside).  It is clear that this modified function is close to $\mcl{M}$ in $L^1$-distance but it will be convenient to use in the analysis later on.
\begin{definition}
For a mixture of Gaussians $\mcl{M} = w_1G_1 + \dots + w_kG_k$ where $G_j = N(\mu_j, \sigma_j^2)$ and a parameter $l$, define the function $\mcl{M}_{\text{sig},l}(x)$ to be, at each point $x \in \R$, equal to the weighted sum of the components $G_j$ of $\mcl{M}$ such that $x$ is in the $10l$-significant interval of $G_j$.  Formally, 
\[
\mcl{M}_{\text{sig},l}(x) = \sum_{\substack{j \text{ such that}\\  |x  - \mu_j| \leq 10l\sigma_j}} w_j G_j(x) \,.
\]
\end{definition}
The following claim is immediate from the definition.
\begin{claim}\label{claim:significant-function}
Let $\eps > 0$ be some parameter and let $l = \lceil \sqrt{\log 1/\eps} \rceil$.  Then
\[
\norm{\mcl{M} - \mcl{M}_{\text{sig}, l}}_1 \leq \eps
\]
\end{claim}
\begin{proof}
The inequality holds because the total mass of a Gaussian outside of its $10l$-significant interval is at most $\eps$. 
\end{proof}

 We now present our first structural result.

\begin{claim}\label{claim:structural-characterization}
Let $\mcl{M} = w_1G_1 + \dots + w_kG_k$ be an arbitrary mixture of Gaussians where $G_i = N(\mu_i, \sigma_i^2)$.  Let $\eps > 0$ be some parameter and let $l = \lceil \sqrt{\log (1/\eps)} \rceil $.  There exist disjoint intervals $I_1, \dots , I_n$ with lengths, say $t_1, \dots , t_n$, where $n \leq 50kl$ with the following property:
\begin{itemize}
    \item For each interval $I_i$, for all $j \in [k]$ either the the $10l$-significant interval of $G_j$ is disjoint from $I_i$ or $\sigma_j \geq t_i$
    \item We have
    \[
    \norm{\mcl{M} - (1_{I_1} + \dots + 1_{I_n})\mcl{M}}_1 \leq \eps
    \]
\end{itemize}
\end{claim}
\begin{proof}
Sort the Gaussians by their standard deviations, WLOG $\sigma_1 \leq \dots \leq \sigma_k$.  Now we will create several intervals $A_1, A_2, \dots $ and we will also associate each interval with one of the Gaussians $G_1, \dots , G_k$ which we will call its parent.

First, set $A_1$ to be the $10l$-significant interval of $G_1$.  Next, we will process the Gaussians $G_2, \dots , G_k$ in order.  For $G_j$, assume that the intervals we have created so far are $A_1, \dots , A_m$ (which will be disjoint by construction).  Now consider the $10l$-significant interval of $G_j$, say $L_j$.   Note that removing the union of the intervals $A_1, \dots , A_m$ from $L_j$ divides $L_j$ into several (at most $m+1$) disjoint intervals.  We label these intervals $A_{m+1}, A_{m+2}, \dots $ and set all of their parents to be $G_j$.  We then move onto $G_{j+1}$ and repeat the above process.  The following properties are immediate from the construction:
\begin{enumerate}
    \item If the parent of $A_i$ is $G_j$ then the length of $A_i$ is at most $20l\sigma_j$
    \item The union of all of the $A_i$ whose parent is among $G_1, \dots , G_j$ contains the $10l$-significant intervals of all of $G_1, \dots , G_j$
    \item If the parent of $A_i$ is $G_j$, then $A_i$ is disjoint from the $10l$-significant intervals of $G_1, \dots , G_{j-1}$ 
\end{enumerate}
Now we claim that at the end of the algorithm, the total number of intervals is at most $2k$.  To see this, say that after processing $G_{j-1}$, the intervals we have created are $A_1, \dots A_m$.  Now consider the potential that is the number of intervals $m$ plus the number of connected components in $A_1 \cup \dots \cup A_m$.  This potential can increase by at most $2$ when processing $G_j$ so thus the total number of intervals at the end of the execution is at most $2k$.  We will now assume that the intervals at the end of the execution are $A_1, \dots , A_{2k}$ (if there are less than $2k$ intervals then add a bunch of dummy intervals of length $0$).

We now describe a post-processing step.  For each of $A_1, \dots , A_{2k}$, if its parent is $G_j$ then divide it into intervals of length at most $\sigma_j$ and assign $G_j$ as the parent of all of these intervals.  By property 1, we can ensure that this creates a total of at most $50kl$ intervals, say $I_1, \dots , I_{n}$ where $n \leq 50kl$.  We use $t_1, \dots , t_n$ to denote their lengths.  We now prove that this set of intervals satisfies the desired properties.  First, note that the following properties are immediate from the construction:
\begin{enumerate}
    \item If the parent of $I_i$ is $G_j$ then $I_i$ is contained in the $10l$-significant interval of $G_j$ and $t_i \leq \sigma_j$
    \item  The union of all of $I_1, \dots , I_{n}$  contains the $10l$-significant intervals of all of $G_1, \dots , G_k$
    \item If the parent of $I_i$ is $G_j$, then $I_i$ is disjoint from the $10l$-significant intervals of $G_1, \dots , G_{j-1}$
\end{enumerate}

The first of the desired properties is clear since by construction if the parent of $I_i$ is $G_j$, then $t_i \leq \sigma_j$ and it must be disjoint from the $10l$-significant intervals of all of $G_1, \dots , G_{j-1}$ (where recall $G_1, \dots , G_k$ are sorted in increasing order of their standard deviation).  Now it remains to verify the second property.  Consider the function $\mcl{M}_{\text{sig},l}(x)$.  Recall by Claim \ref{claim:significant-function},
\begin{equation}\label{eq:approxbound1}
\norm{\mcl{M} - \mcl{M}_{\text{sig}, l}}_1 \leq \eps
\end{equation}
Next observe that
\[
\mcl{M}_{\text{sig}, l} =(1_{I_1} + \dots + 1_{I_{n}}) \mcl{M}_{\text{sig}, l} \,.
\]
Combining the above, we have
\[
\norm{\mcl{M} - (1_{I_1} + \dots + 1_{I_{n}}) \mcl{M}_{\text{sig},l}}_1 \leq \eps \,.
\]
However, note that we have 
\[
(1_{I_1} + \dots + 1_{I_{n}}) \mcl{M}_{\text{sig},l} \leq (1_{I_1} + \dots + 1_{I_{n}}) \mcl{M} \leq \mcl{M}
\]
everywhere along the real line.  Thus, we immediately get the desired inequality.
\end{proof}

The next structural result shows that the intervals $I_1, \dots , I_n$ obtained in Claim \ref{claim:structural-characterization} are ``findable" in the sense that if we draw many samples from $\mcl{M}$ (or from a distribution $f$ that is close to $\mcl{M}$), then with high probability, there will be samples close to the endpoints of each of $I_1, \dots , I_n$.  This will mean that algorithmically, it suffices to draw sufficiently many samples and then only consider intervals whose endpoints are given by a pair of samples.
 
\begin{claim}\label{claim:structural-part2}
Let $\mcl{M} = w_1G_1 + \dots + w_kG_k$ be an arbitrary mixture of Gaussians where $G_i = N(\mu_i, \sigma_i^2)$.  Let $\eps > 0$ be some parameter and let $l = \lceil \sqrt{\log (1/\eps)} \rceil $.  Let $f$ be a distribution.  Assume we are given samples from $f$, say $x_1, \dots , x_Q$ for some sufficiently large $Q = \poly(k/\eps)$.  Then with high probability, there exists pairs $\{x_{a_1}, x_{b_1} \}  ,\dots , \{x_{a_n}, x_{b_n} \} $ such that 
\begin{itemize}
     \item The intervals $ J_1 = [x_{a_1}, x_{b_1}], \dots , J_n = [x_{a_n}, x_{b_n}] $ are disjoint  \item $n \leq 50kl$
     \item For each interval $J_i$, for all $j \in [k]$ either the the $10l$-significant interval of $G_j$ is disjoint from $J_i$ or $\sigma_j \geq |x_{b_i} - x_{a_i}|$
     \item 
      \[
        \norm{\mcl{M} -(1_{J_1} + \dots + 1_{J_n}) \mcl{M}  }_1 \leq  4(\eps + \norm{\mcl{M} - f}_1) \,.
        \]
 \end{itemize}
 \end{claim}
 \begin{proof}
 Let $I_1, \dots , I_n$ be the intervals computed in Claim \ref{claim:structural-characterization} applied to the mixture $\mcl{M}$ and assume that their lengths are $t_1, \dots , t_n$.  Let $C = \lceil (k/\eps)^2 \rceil$.  For each interval $I_i$, divide it into  $C$ subintervals $I_i^1, \dots , I_i^C$ of length $t_i/C$ and assume that these subintervals are sorted in order.  We say one of these subintervals is good if 
 \[
 \int_{I_i^j} f \geq (\eps/k)^{10} \,.
 \]
 For an index $i$, let $c_i, d_i$ be the smallest and largest index such that $I_i^{c_i}, I_i^{d_i}$ are good respectively.  Then with high probability for all $i$, there will be samples, say $x_{a_i}, x_{b_i}$ in $I_i^{c_i}$ and  $I_i^{d_i}$.  Now we will form the intervals $J_i = [x_{a_i}, x_{b_i}]$.  The first two of the desired properties are clear.  The third follows from the statement of Claim \ref{claim:structural-characterization}. It remains to verify the last.  Similar to the proof of Claim \ref{claim:structural-characterization}, we consider the function $\mcl{M}_{\text{sig},l}$.  Note that
 \begin{align*}
 &\int_{I_i \backslash J_i} \mcl{M}_{\text{sig},l} \leq \int_{I_i^1} \mcl{M}_{\text{sig},l} + \dots + \int_{I_i^{c_i}} \mcl{M}_{\text{sig},l} + \int_{I_i^{d_i}} \mcl{M}_{\text{sig},l} + \dots + \int_{I_i^C} \mcl{M}_{\text{sig},l} \\ &\leq \norm{1_{I_i}(\mcl{M}_{\text{sig},l} - f)}_1 + \int_{I_i^1} f + \dots + \int_{I_i^{c_i - 1}}f + \int_{I_i^{d_i + 1}} f + \dots + \int_{I_i^C} f  \\ &\quad + \int_{I_i^{c_i}} \mcl{M}_{\text{sig},l} + \int_{I_i^{d_i}} \mcl{M}_{\text{sig},l} \\ & \leq \norm{1_{I_i}(\mcl{M}_{\text{sig},l} - f)}_1  +(\eps/k)^2 + \int_{I_i^{c_i}} \mcl{M}_{\text{sig},l} + \int_{I_i^{d_i}} \mcl{M}_{\text{sig},l}
 \end{align*}
 where the last step follows by the minimality and maximality of $c_i, d_i$.  Now by construction, the only Gaussians among $G_1, \dots , G_k$ whose $10l$-significant intervals intersect $I_i$ must have standard deviation at least $t_i$.  Since  $I_i^{c_i}, I_i^{d_i}$ each have length $(\eps/k)^2t_i$, we conclude
 \[
  \int_{I_i \backslash J_i} \mcl{M}_{\text{sig},l}  \leq 3(\eps/k)^2 + \norm{1_{I_i}(\mcl{M}_{\text{sig},l} - f)}_1\,.
 \]
 Thus, we have
 \[
 \norm{(1_{I_1} + \dots + 1_{I_n}) \mcl{M}_{\text{sig},l} - (1_{J_1} + \dots + 1_{J_n}) \mcl{M}_{\text{sig},l}}_1 \leq \eps + \norm{(\mcl{M}_{\text{sig},l} - f)}_1  \leq 2\eps + \norm{(\mcl{M} - f)}_1\,.
 \]
 where we used Claim \ref{claim:significant-function} in the last step. Now, using the statement of Claim \ref{claim:structural-characterization}, we deduce
 \[
 \norm{ \mcl{M} - (1_{J_1} + \dots + 1_{J_n}) \mcl{M}_{\text{sig},l}}_1 \leq 4\eps + \norm{(\mcl{M} - f)}_1 \,.
 \]
 Since 
\[
(1_{J_1} + \dots + 1_{J_{n}}) \mcl{M}_{\text{sig},l} \leq (1_{J_1} + \dots + 1_{J_{n}}) \mcl{M} \leq \mcl{M}
\]
everywhere along the real line, we immediately get the desired inequality.
 \end{proof}
 
\subsection{Finishing the Proof} 
We can now prove the key lemma and then Theorem \ref{thm:main-gmm} will follow as an immediate consequence since we can use the improper learner in Corollary \ref{coro:improper-learner}.  The lemma states that given explicit access to a distribution $f$ that is $\eps$-close to a GMM, $\mcl{M}$, with $k$ components, we can output a GMM, $\wt{\mcl{M}}$, with $\wt{O}(k)$ components that is $\wt{O}(\eps)$-close to $f$.  At a high-level, the proof involves attempting to reconstruct $f$ over various intervals using Lemma \ref{lem:reconstruct-interval} and then using a dynamic program to find a union of $\wt{O}(k)$ such intervals that approximates the entire function.  We use Claim \ref{claim:structural-part2} to argue that such a solution exists so our dynamic program must find it.

\begin{lemma}\label{lem:reconstruct-Gaussian}
Let $\mcl{M} = w_1G_1 + \dots + w_kG_k$ be an arbitrary mixture of Gaussians where $G_i = N(\mu_i, \sigma_i^2)$.  Let $\eps > 0$ be some parameter.  Assume we are given access to a distribution $f$.  Then in $\poly(k/\eps)$ time and with high probability, we can compute a mixture $\wt{\mcl{M}}$ of at most $k\poly(\log (k/\eps)) $ Gaussians such that 
\[
\norm{\wt{\mcl{M}} - f  }_1 \leq  \poly(\log (k/\eps)) \left( \eps +  \norm{ \mcl{M} - f}_1 \right)  \,.
\]
 \end{lemma}
\begin{proof}
Note that it suffices to compute a weighted sum of Gaussians that satisfies the desired inequality since rescaling such a weighted sum to a mixture will at most increase the $L^1$ error by a factor of $2$.  Thus, from now on, we will not worry about ensuring the mixing weights sum to $1$.

Let $\gamma = \eps/k$ and $l = \lceil \sqrt{\log 1/\gamma}\rceil $. First draw $Q = \poly(k/\eps)$ samples $x_1, \dots , x_Q$ from $f$ for sufficiently large $Q$ that we can apply Claim \ref{claim:structural-part2} with $\eps \leftarrow \gamma$.   While we do not know what the intervals $J_1, \dots , J_n$ are, we will set up a dynamic program to find a set of at most $50kl $ intervals that we can reconstruct $f$ over each one using Lemma \ref{lem:reconstruct-interval} and such that these intervals contain essentially all of the mass of $f$.
\\\\
For each pair $x_a, x_b $ with $a,b \in \{1,2, \dots , Q \}$, apply Lemma \ref{lem:reconstruct-interval} with parameter $\eps \leftarrow \gamma$ to attempt to approximate $f$ on the interval $[x_a,x_b]$.  Let the output obtained be $\wt{\mcl{M}}_{x_a,x_b}$.  Note that sometimes the algorithm will fail to output a good approximation to $f$ (because the assumptions of Lemma \ref{lem:reconstruct-interval} fail) but we can ensure that the output is a weighted sum of at most $\poly(\log 1/\gamma)$ Gaussians.  In the proof we will only use the fact that when the assumptions of Lemma \ref{lem:reconstruct-interval} hold, then our approximation to $f$ restricted to the interval will be accurate.  
\\\\
Now we show how to set up the dynamic program.  WLOG the points $x_1, \dots , x_Q$ are sorted in nondecreasing order.  We also use the convention that $x_0 = -\infty, x_{Q+1} = \infty$.  Now, we maintain the following state for each index $0 \leq j \leq Q + 1$, and integer $c \leq 50kl$: the best approximation to $f$ from $(-\infty, x_j]$ using  a sum of $\wt{\mcl{M}}_{x_a,x_b}$ over at most $c$ intervals.  Formally,
\paragraph{Dynamic Program:} Let $DP_{j,c}$ be the minimum over all sets $S$ of $c$ pairs $(a^{(1)},b^{(1)}), \dots , (a^{(c)},b^{(c)})  \in [j] \times [j]$ such that $a^{(1)} < b^{(1)} \leq a^{(2)} < \dots < b^{(c)}$ of
\begin{align*}
&\norm{f\cdot 1_{(-\infty, x_j]} - \left(1_{[x_{a^{(1)}},x_{b^{(1)}}]}\wt{\mcl{M}}_{x_{a^{(1)}},x_{b^{(1)}}}  + \dots +   1_{[x_{a^{(c)}},x_{b^{(c)}}]} \wt{\mcl{M}}_{x_{a^{(c)}},x_{b^{(c)}}} \right) }_1 \\ &+ \norm{\wt{\mcl{M}}_{x_{a^{(1)}},x_{b^{(1)}}} - 1_{[x_{a^{(1)}},x_{b^{(1)}}]}\wt{\mcl{M}}_{x_{a^{(1)}},x_{b^{(1)}}}}_1 + \dots + \norm{\wt{\mcl{M}}_{x_{a^{(c)}},x_{b^{(c)}}} - 1_{[x_{a^{(c)}},x_{b^{(c)}}]}\wt{\mcl{M}}_{x_{a^{(c)}},x_{b^{(c)}}}}_1 \,.
\end{align*}
Note that the first term represents the approximation error compared to $f$.  We must truncate each function $\wt{\mcl{M}}_{x_{a^{(1)}},x_{b^{(1)}}}$ to its corresponding interval $[x_{a^{(1)}},x_{b^{(1)}}]$ in order for the problem to be solvable via dynamic programming because otherwise previous choices would affect later ones.  Thus, we also need to add the additonal terms that represent the error from truncation.  Note that the $L^1$ distance can be estimated using Claim \ref{claim:testL1}.  We can solve the dynamic program in polynomial time because from each state $DP_{j,c}$, we simply consider adding all possible intervals among $x_j, x_{j+1}, \dots , x_Q$ as the next one.  
\\\\
Now we prove that there is a good solution to the dynamic program for which the objective (for $j = Q+1$) is small.  Let $a_1,b_1, \dots , a_n, b_n$ be the indices obtained in Claim \ref{claim:structural-part2}.  We claim that setting  
\[
(a^{(1)},b^{(1)}) = (a_1, b_1) ,\dots ,  (a^{(n)},b^{(n)}) = (a_n, b_n) 
\]
results in the objective function being small.  Let $J_i = [a_i,b_i]$ for all $i$.   Let 
\[
\wt{\mcl{M}}^{\text{good}} = 1_{[x_{a_1}, x_{b_1}]} \wt{\mcl{M}}_{x_{a_1},x_{b_1}} + \dots + 1_{[x_{a_n}, x_{b_n}]} \wt{\mcl{M}}_{x_{a_n},x_{b_n}} \,.
\]
The guarantee from Lemma \ref{lem:reconstruct-interval} implies that for all $i$,
\begin{equation}\label{eq:trunc-bound}
\norm{ \wt{\mcl{M}}_{x_{a_i},x_{b_i}} -  1_{[x_{a_i}, x_{b_i}]} \wt{\mcl{M}}_{x_{a_i},x_{b_i}}}_1 \leq  \poly(\log 1/\gamma)(\gamma  + \norm{1_{J_i}(\mcl{M} - f)}_1)
\end{equation}
and thus, using the guarantee from Lemma \ref{lem:reconstruct-interval} again, we have
\begin{align*}
\norm{\wt{\mcl{M}}^{\text{good}} - (1_{J_1} + \dots + 1_{J_n} )f }  \leq \norm{\wt{\mcl{M}}_{x_{a_1},x_{b_1}} - 1_{J_1} \cdot f}_1 + \dots +  \norm{\wt{\mcl{M}}_{x_{a_n},x_{b_n}} - 1_{J_n} \cdot f}_1 \\ \quad + \norm{ \wt{\mcl{M}}_{x_{a_1},x_{b_1}} -  1_{[x_{a_1}, x_{b_1}]} \wt{\mcl{M}}_{x_{a_1},x_{b_1}}}_1 + \dots + \norm{ \wt{\mcl{M}}_{x_{a_n},x_{b_n}} -  1_{[x_{a_n}, x_{b_n}]} \wt{\mcl{M}}_{x_{a_n},x_{b_n}}}_1 \\ \leq n \gamma \poly(\log 1/\gamma) + \poly(\log 1/\gamma) \left( \norm{1_{J_1}(\mcl{M} - f)}_1 + \dots + \norm{1_{J_n}(\mcl{M} - f)}_1\right) \\ \leq  \poly(\log 1/\gamma)\left( \eps + \norm{\mcl{M} - f }_1 \right) \,.
\end{align*}
However also recall that by Claim \ref{claim:structural-part2},
\[
 \norm{\mcl{M} -(1_{J_1} + \dots + 1_{J_n}) \mcl{M}  }_1 \leq  4(\gamma +\norm{\mcl{M} - f }_1)  \,.
\]
Combining the above two inequalities, we get
\[
\norm{\wt{\mcl{M}}^{\text{good}} - f}_1 \leq \poly(\log 1/\gamma)\left( \eps + \norm{\mcl{M} - f }_1 \right) \,.
\]
Finally, combining the above with (\ref{eq:trunc-bound}) implies that the objective value of the dynamic program is at most $ \poly(\log 1/\gamma)\left( \eps + \norm{\mcl{M} - f }_1 \right)$.  Finally it remains to note that the objective of the dynamic program (for $j = Q+1$) is an upper bound on 
\[
\norm{f - \left(\wt{\mcl{M}}_{x_{a^{(1)}},x_{b^{(1)}}}  + \dots +   \wt{\mcl{M}}_{x_{a^{(c)}},x_{b^{(c)}}} \right) }_1
\]
so thus, we can simply output the solution $\wt{\mcl{M}} = \wt{\mcl{M}}_{x_{a^{(1)}},x_{b^{(1)}}}  + \dots +   \wt{\mcl{M}}_{x_{a^{(c)}},x_{b^{(c)}}} $ and we are guaranteed to have
\[
\norm{f -\wt{\mcl{M}} }_1 \leq  \poly(\log 1/\gamma)\left( \eps + \norm{\mcl{M} - f }_1 \right) \,.
\]
It is clear that $\wt{\mcl{M}}$ is a weighted sum of at most $k \poly(\log 1/\gamma)$ Gaussians (since for each interval there are $\poly(\log 1/\gamma)$ Gaussians and there are at most $50kl$ total intervals).  It is clear that all of the steps run in $\poly(k/\eps)$ time and we are done.
\end{proof}

Now we can complete the proof of our main theorem, Theorem \ref{thm:main-gmm}.

\begin{proof}[Proof of Theorem \ref{thm:main-gmm}]
We can apply Corollary \ref{coro:improper-learner} to learn a distribution $f$ such that $d_{\TV}(\mcl{M} , f) \leq O(\eps)$.  We can then apply Lemma \ref{lem:reconstruct-Gaussian} using $f$.  Note that since $f$ is a piecewise polynomial, we can perform all of the explicit computations with the density function that are used in the proof of Lemma \ref{lem:reconstruct-Gaussian}.  It is immediate that the output of Lemma \ref{lem:reconstruct-Gaussian} must satisfy 
\[
d_{\TV}(\wt{\mcl{M}}, \mcl{M} ) \leq \wt{O}(\eps) 
\]
so we are done.
\end{proof}

\section{Well-Conditioned Case: Sparse Fourier Reconstruction}\label{sec:fourier-well-conditioned}

We now move onto our results on sparse Fourier reconstruction.  As with GMMs, we will first consider the well-conditioned case.  Here, this means a function that has its Fourier support contained in one interval that is not too long i.e., all of its Fourier mass is not too spread out.  Note that WLOG, we may assume that this interval is centered at $0$ since otherwise we can multiply by a suitable exponential to shift the Fourier support to be around $0$.  We prove the following statement:

\begin{lemma}\label{lem:Fourier-simplecase}
Let $0 < \eps < 0.1$ be a parameter and let $l \geq \lceil \log 1/\eps \rceil$ be some parameter.  Let $\mcl{M}$ be a function such that $\wh{\mcl{M}}$ is supported on $[-l, l]$ and such that $\norm{\wh{\mcl{M}}}_1 \leq 1$.  Also assume that we have access to a function $f$ such that 
\[
\int_{-1}^1 |f(x) - \mcl{M}(x)|^2 dx \leq \eps^2 \,.
\]
There is an algorithm that runs in $\poly(l)$ time and outputs a function $\wt{\mcl{M}}$ such that $\wt{\mcl{M}}$ is $(O(l), O(l))$-simple, has Fourier support contained in $[-l, l]$, and
\[
\int_{-1}^1 |\wt{\mcl{M}}(x) - f(x)|^2 dx \leq 16\eps^2 \,. 
\]
\end{lemma}
\begin{remark}
Note that in this case, we do not need any constraint on the Fourier sparsity of $\mcl{M}$ to guarantee that the output of our algorithm is $O(\log 1/\eps)$-Fourier sparse.  Also, unlike our full result, Theorem \ref{thm:fourier-main}, our output in this case is guaranteed to be a good approximation over the entire interval (instead of a subinterval).
\end{remark}

Our proof will be separated into two parts.  The first step will be proving the existence of a function $\wt{\mcl{M}}$ of the desired form.  The second step will be developing an algorithm to actually compute it.
\subsubsection{Existence of a Sparse Approximation}

First, we will prove that under the assumptions of Lemma \ref{lem:Fourier-simplecase}, an approximation $\wt{\mcl{M}}$ satisfying the desired properties exists.  We will also prove that independent of the problem instance, it suffices to only consider a fixed set of $O(l)$ distinct frequencies given by the Chebyshev points (with suitable rescaling).  

The proof relies on first taking the Taylor series of an exponential $e^{2\pi i \zeta x}$ and arguing that we only need to keep the first $O(\log 1/\eps)$ terms.  This essentially lets us represent such an exponential with the coefficients of its Taylor series, which are (up to rescaling) $(1, \zeta, \zeta^2 , \dots )$.  We then use Corollary \ref{coro:convexhull} to argue that an arbitary linear combination of such vectors can be replaced with a sparse combination with similarly sized coefficients.   

\begin{lemma}\label{lem:Fourier-existence-simplecase}
Let $0 < \eps < 0.1$ be a parameter and let $l \geq \lceil \log 1/\eps \rceil$ be some parameter.  Let $t_0, \dots , t_{10^2l}$ be the degree-$10^2l$ Chebyshev points.  Let $\mcl{M}$ be a function such that $\wh{\mcl{M}}$ is supported on $[-l, l]$ and such that $\norm{\wh{\mcl{M}}}_1 \leq 1$.  Then there is a function 
\[
h(x) = \sum_{j = 0}^{10^2l} c_j e^{2\pi i (lt_j) x}
\]
where $c_0, \dots , c_{10^2l}$ are complex numbers such that $\sum_{j = 0}^{10^2l} |c_j| \leq 200l$ and 
\[
\int_{-1}^1 (h(x) - \mcl{M}(x))^2 dx \leq \eps^2 \,.
\]
\end{lemma}
\begin{proof}
Note that $\mcl{M}$ can be written as $\mcl{M}(x) = \int_{-l}^l \wh{\mcl{M}}(\zeta) e^{2 \pi i x \zeta}  d\zeta $.  Now consider the Taylor expansion of 
\[
e^{2 \pi i x \zeta}  = \sum_{j = 0}^{\infty} \frac{(2 \pi i x \zeta)^j}{j!}
\]
Note that since $ -l \leq \zeta \leq l$, we have
\[
\sum_{j = 10^2l + 1}^{\infty} \left \lvert \frac{(2 \pi i \zeta)^j}{j!} \right \rvert  \leq (\eps/l)^3 \,.
\]
In particular, if we define 
\[
g_{\zeta}(x) = \sum_{j = 0}^{10^2l} \frac{(2 \pi i x \zeta)^j}{j!}
\]
then over the interval $x \in [-1,1]$
\begin{equation}\label{eq:taylor-bound}
|e^{2 \pi i x \zeta} - g_{\zeta}(x)| \leq (\eps/l)^3 \,.
\end{equation}
Next, for each $\zeta \in [-l,l]$, by Corollary \ref{coro:convexhull}, we can write the vector
\[
\mcl{V}_{10^2l}(\zeta) = w_0(\zeta) \mcl{V}_{10^2l}(t_0 l) + \dots + w_{10^2l}(\zeta) \mcl{V}_{10^2l}(t_{10^2l} l)
\]
for some real numbers (depending on $\zeta$) $w_0(\zeta), \dots , w_{10l}(\zeta)$ with $\sum |w_j(\zeta)| \leq 200l$.  Thus,
\[
g_{\zeta}(x) = w_0(\zeta) g_{t_0 l}(x) + \dots + w_{10^2l}(\zeta) g_{t_{10^2l} l}(x)
\]
for the same weights.  Now note that by (\ref{eq:taylor-bound}), for all $x \in [-1,1]$,
\[
\left \lvert \mcl{M}(x) - \sum_{j = 0}^{10^2l} g_{t_j l}(x) \left(\int_{-l}^l \wh{\mcl{M}}(\zeta) w_j(\zeta)  d\zeta \right) \right \rvert = \left \lvert \mcl{M}(x) - \int_{-l}^l \wh{\mcl{M}}(\zeta) g_{\zeta}(x)  d\zeta \right \rvert  \leq 2l \cdot (\eps/l)^3 \,.
\]

Note that 
\[
\sum_{j = 0}^{10^2l}\left \lvert \int_{-l}^l \wh{\mcl{M}}(\zeta) w_j(\zeta)  d\zeta \right \rvert \leq  200l
\]
so by (\ref{eq:taylor-bound}), for all $x \in [-1,1]$,
\[
\left \lvert \sum_{j = 0}^{10^2l} \left(g_{t_j l}(x) - e^{2\pi i x (t_j l)} \right) \left(\int_{-l}^l \wh{\mcl{M}}(\zeta) w_j(\zeta)  d\zeta \right) \right \rvert \leq 200l (\eps/l)^3
\]
so therefore for all $x \in [-1,1]$, we have 
\[
\left \lvert \mcl{M}(x) - \sum_{j = 0}^{10^2l}  e^{2\pi i x (t_j l)} \left(\int_{-l}^l \wh{\mcl{M}}(\zeta) w_j(\zeta)  d\zeta \right) \right \rvert \leq 202l \cdot (\eps/l)^3
\]
and setting
\[
h(x) = \sum_{j = 0}^{10^2l}  e^{2\pi i x (t_j l)} \left(\int_{-l}^l \wh{\mcl{M}}(\zeta) w_j(\zeta)  d\zeta \right)
\]
immediately leads to the desired conclusion.
\end{proof}

\subsubsection{Completing the Proof of Lemma \ref{lem:Fourier-simplecase}}

By combining Lemma \ref{lem:Fourier-existence-simplecase} and Lemma \ref{lem:linear-regression}, we can complete the proof of Lemma \ref{lem:Fourier-simplecase}.

\begin{proof}[Proof of Lemma \ref{lem:Fourier-simplecase}]
We can separate $f$ into its real and imaginary parts, say $f_{\textsf{re}},f_{\textsf{im}}$ and we can separate $\mcl{M}$ into its real and imaginary parts $\mcl{M}_{\textsf{re}},\mcl{M}_{\textsf{im}}$.  Now consider the Chebyshev points of degree $10^2l$, say $t_0, \dots , t_{10^2l}$.  We will now apply Lemma \ref{lem:linear-regression} where we consider the set of functions 
\[
\{f_1, \dots , f_n \} = \{ \pm \cos(2\pi t_0 x), \pm \sin(2\pi t_0 x), \dots ,  \pm \cos(2\pi t_{10^2l} x), \pm \sin(2\pi t_{10^2l} x) \} \,.
\]
The distribution $\mcl{D}$ is the uniform distribution on $[-1,1]$ and by Lemma \ref{lem:Fourier-existence-simplecase}, there are coefficients $a_1, \dots , a_n \geq 0$ with $a_1 + \dots + a_n \leq O(l)$ (note we can split the complex coefficients $c_j$ into their real and imaginary parts and then split into positive and negative parts) such that 
\begin{align*}
\int_{-1}^1 (f_{\textsf{re}}(x) - (a_1f_1(x) + \dots + a_n f_n(x)))^2 dx  \leq  2\int_{-1}^1 (\mcl{M}_{\textsf{re}}(x) - (a_1f_1(x) + \dots + a_n f_n(x)))^2 dx \\  + 2\int_{-1}^1 (f_{\textsf{re}}(x) - \mcl{M}_{\textsf{re}}(x))^2 dx \leq 4\eps^2
\end{align*}
and similar for the imaginary part of $f$.  Applying Lemma \ref{lem:linear-regression} to both the real and imaginary part (after rescaling by $1/(O(l))$), adding the results, and rewriting the trigonometric functions using complex exponentials (note the set $\{t_0, \dots , t_{10^2l} \}$ is symmetric around $0$ so we can do this) completes the proof.
\end{proof}

We can slightly extend Lemma \ref{lem:Fourier-simplecase} to work even if we do not know the desired accuracy $\eps$ but only a lower bound on it.  It suffices to run the algorithm for Lemma \ref{lem:Fourier-simplecase} and repeatedly decrease the target accuracy until our algorithm fails to find the optimal accuracy within a constant factor. 

\begin{corollary}\label{coro:fourier-simplecase}
Let $l, \eps$ be parameters given to us such that $l \geq \lceil \log 1/\eps \rceil$.  Let $\mcl{M}$ be a function such that $\wh{\mcl{M}}$ is supported on $[-l, l]$ and $\norm{\wh{\mcl{M}}}_1 \leq 1$.  Assume that we have access to a function $g$ defined on $[-1,1]$.  There is an algorithm that runs in $\poly(l)$ time and outputs a function $\wt{\mcl{M}}$ such that $\wt{\mcl{M}}$ is $(O(l), O(l))$-simple, has Fourier support contained in $[-l, l]$, and
\begin{align*}
\int_{-1}^1 |\wt{\mcl{M}}(x) - f(x)|^2 dx \leq  20\left( \eps^2 + \int_{-1}^1 |f(x) - \mcl{M}(x)|^2 dx \right)  \,. 
\end{align*}
\end{corollary}
\begin{proof}
For a target accuracy $\gamma > \eps$ we run the algorithm in Lemma \ref{lem:Fourier-simplecase} to get a function $\wt{\mcl{M}}_{\gamma}(x)$.  We then check whether 
\[
\int_{-1}^1 |\wt{\mcl{M}}_{\gamma}(x) - f(x)|^2 dx \leq 16\gamma^2 \,.
\]
Note that the above can be explicitly computed.  If the above check passes, we then take $\gamma \leftarrow 0.99 \gamma $.  Taking the smallest $\gamma$ for which the above succeeds, the guarantee from Lemma \ref{lem:Fourier-simplecase} ensures that we have a function $\wt{\mcl{M}}$ such that
\[
\int_{-1}^1 |\wt{\mcl{M}}(x) - f(x)|^2 dx \leq 20 \left(\eps^2 +  \int_{-1}^1 |f(x) - \mcl{M}(x)|^2 dx \right) \,.
\]
It is clear that we run the routine from Lemma \ref{lem:Fourier-simplecase} at most $O(l)$ times so we are done.
\end{proof}

\section{Sparse Fourier Reconstruction: Full Version}\label{sec:sparse-fourier-full}

In this section, we complete the proof of our main result on sparse Fourier reconstruction, Theorem \ref{thm:fourier-main}.  The high-level outline of the proof is similar to the proof of Theorem \ref{thm:main-gmm}.  The key lemma that goes into the proof is stated below.  At a high level, the lemma states that if we know roughly where the Fourier support of the unknown Fourier-sparse signal $\mcl{M}$ is located, then we can successfully reconstruct it.

\begin{lemma}\label{lem:fourier-manycluster}
Assume we are given $N, k, \eps , c$ with $0 < \eps < 0.1$.  Let $l = \lceil \log kN / (\eps c)  \rceil$ be some parameter.  Let $\mcl{M}$ be a function that is $(k,1)$-simple.  Also, assume that we are given a set $T \subset \R$  of size $N$ such that all of the support of $\wh{\mcl{M}}$ is within distance $1$ of $N$.  Further, assume we are given access to a function $f$ such that 
\[
\int_{-1}^1 |f(x) - \mcl{M}(x)|^2 dx \leq \eps^2 \,.
\]
There is an algorithm that runs in $\poly(N,k,l , 1/c)$ time and outputs a function $\wt{\mcl{M}}$ that is $(k\poly(l/c), k\poly(l/c) )$-simple and
\[
\int_{-1 + c}^{1 - c} |\wt{\mcl{M}}(x) - f(x)|^2 dx \leq \eps^2 \poly(l) \,. 
\]
\end{lemma}

The proof of Lemma \ref{lem:fourier-manycluster} will involve localizing the frequencies and then using Corollary \ref{coro:fourier-simplecase} to reconstruct after localizing.  We will do this for $\poly(N,k, \log 1/\eps)$ different localizations (based on the set $T$ that we are given).  We will then select at most $\wt{O}(k)$ of these localized reconstructions to add together and output.  The intuition behind why we can find such a set of $\wt{O}(k)$ localized reconstructions and ignore the rest is that $\mcl{M}$ is $k$-Fourier sparse so localizations that are far away from the frequencies of $\mcl{M}$ can essentially be ignored.

The localization procedure will involve convolving $f$ by a Gaussian times an exponential (technically we will convolve by a function that approximates a Gaussian times an exponential).  Note that this is equivalent to multiplying the Fourier transform by a Gaussian multiplier.  This will ensure that frequencies too far away from a certain target frequency will only contribute negligibly and we only need to worry about reconstructing the frequencies that are close to the target frequency.

\subsection{Properties of Localization}\label{sec:local-fourier-properties}
In this section, we formalize the localization step and prove several inequalities that will be used in the proof of Lemma \ref{lem:fourier-manycluster}.  The way we would like to localize the frequencies is by multiplying by a Gaussian multiplier in Fourier space since afterwards, we would be able to essentially neglect any frequencies that are far away from the center of the Gaussian multiplier.  This is equivalent to convolving by a Gaussian times an exponential, i.e. a function of the form $G(x)e^{2 \pi i \theta x}$, in real space.  For technical reasons, we will actually define two types of functions, which we call kernels, that approximate functions of the form $G(x)e^{2 \pi i \theta x}$.  The reason that we will need to work with both is that the first type can be computed efficiently while the second is easier to use in the analysis of our algorithm.   
\\\\
We begin with a few definitions.
\begin{definition}
For a function $f: \R \rightarrow \C$ and any $l > 0$, define $f^{\trunc(l)}$ to be the function that is equal to $f$ on $[-l,l]$ and $0$ otherwise.
\end{definition}

\begin{definition}\label{def:localization-type2}
For parameters $\mu, l , c$. we define the function 
\[
\mcl{K}_{\mu, l ,c} = (1/c)\mcl{P}_{l}^{\trunc(l)}(x/c) \cdot e^{2 \pi i \mu x}
\]
where $\mcl{P}_l$ is as defined in Definition \ref{def:poly-approx-Gaussian}. 
\end{definition}
\begin{remark}
We call functions of the above form truncated polynomial kernels.  The fact that such functions are truncated polynomials will make it easy to explicitly compute convolutions.
\end{remark}

\begin{definition}\label{def:localization-type1}
For parameters $\mu, l,a$, we define the function $\mcl{T}_{\mu,l ,a}$ as follows.  First define   
\[
\mcl{S}_{ l , a}  = M_{0,1}^{\trunc(l)} ( x /a )  
\]
(recall Definition \ref{def:gaussian-multiplier}) and then define
\[
\mcl{T}_{\mu,l ,a}(x) = \wh{\mcl{S}_{l, a}}(x) e^{2 \pi i \mu x} \,.
\]
\end{definition}
\begin{remark}
We call functions of the above form truncated Gaussian kernels.  Note that truncated Gaussian kernels are compactly supported in Fourier space, which will be a convenient property in the analysis of our algorithm later on.
\end{remark}

Note that both $\mcl{K}_{0, l ,c} $  and $\mcl{T}_{0,l , 2\pi / c}$ are meant to approximate the Gaussian $N(0, c^2)$.  The fact that $\mcl{K}_{0, l ,c} $ approximates $N(0,c^2)$ is clear from the definition (and Lemma \ref{lem:Gaussian-taylorseries}).  To see why $\mcl{T}_{0,l , 2\pi / c}$ approximates the Gaussian, note that if in the definition of $\mcl{T}$, we did not truncate before taking the Fourier transform, then we would get exactly $N(0, c^2)$.  
\\\\
We will now prove several inequalities relating to how convolving with the kernels $\mcl{K}$ and $\mcl{T}$ affect a function.  The first set of bounds are an immediate consequence of Lemma \ref{lem:Gaussian-taylorseries}.
\begin{claim} \label{claim:easyL2bound}
Let $l, \eps > 0$ be parameters such that $l \geq \lceil \log 1/\eps \rceil$.  Let $0 < c < 1$ be some constant.  Then
\begin{align*}
&\norm{\mcl{K}_{0,l,c}(x)  - N(0, c^2)(x) }_1 \leq O(\eps l)  \\
&\norm{\mcl{K}_{0,l,c}(x)  - N(0, c^2)(x) }_2^2 \leq O( l \eps^2 / c) \,.
\end{align*}
\end{claim}
\begin{proof}
We know by Lemma \ref{lem:Gaussian-taylorseries} that 
\[
|(1/c) \mcl{P}_{l}(x/c) - N(0,c^2)(x)| \leq \eps /c
\]
for all $x \in [-lc, lc]$.  Thus, since $G$ decays rapidly outside the interval $[-lc, lc]$ we have
\[
\norm{\mcl{K}_{0,l,c}(x) - N(0,c^2)(x) }_1 \leq O(\eps l) \,.
\]
The second inequality follows by a similar argument.
\end{proof}

The next claim formalizes the intuition that $\mcl{K}_{\mu, l ,c} $  and $\mcl{T}_{\mu,l , 2\pi / c}$ must be close because they both approximate the same function (the function $N(0,c^2) e^{2\pi i \mu x}$).
\begin{claim}\label{claim:fourier-localization}
Let $l$ be some parameter and let $\eps > 0$ be such that $l \geq \lceil \log 1/\eps \rceil$.  Let $g$ be a function such that $\norm{\wh{g}}_1 \leq 1$.  Then for any $\mu, c $,
\[
\norm{\mcl{K}_{\mu,l,c} * g -  \mcl{T}_{\mu,l,2\pi/c} * g}_{\infty} \leq O( \eps l ) \,,
\]
where $*$ denotes convolution.  
\end{claim}
\begin{proof}
First, note that it suffices to prove the above for $\mu = 0$.  Now let $G = N(0, c^2)$.  By Claim \ref{claim:easyL2bound}, we have
\[
\norm{\mcl{K}_{0,l,c}(x) - G(x) }_1 \leq O(\eps l) \,.
\]
Since $\norm{\wh{g}}_1 \leq 1$, we know that $\norm{g}_{\infty} \leq 1$ and thus for all $x$,
\begin{equation}\label{eq:convolutionbound}
|\mcl{K}_{0,l,c} * g(x) - G * g(x)| \leq O(\eps l) \,.
\end{equation}
Finally, observe that the Fourier transform of $G * g$ is equal to $\wh{G} \wh{g} $.  Note that 
\[
\wh{G}(x) = M_{0,(2\pi/ c)^2}(x) = M_{0,1}(cx/(2 \pi )) \,.
\]
By construction, 
\[
\wh{\mcl{T}_{0,l,2\pi/c} * g} = M_{0,1}^{\trunc(l)} ( c x /(2\pi) ) \wh{g}  
\]
which is equal to $\wh{G} \wh{g} $ restricted to the interval $[-2 \pi l/c, 2 \pi l/c]$.  Using the fact that  $\wh{G} = M_{0,(2\pi/ c)^2}$ decays rapidly outside $[-2 \pi l/c, 2 \pi l/c]$, we have that
\[
\norm{\wh{\mcl{T}_{0,l,2\pi/c} * g} - \wh{G} \wh{g} }_1 \leq \eps \,.
\]
Thus, $|\mcl{T}_{0,l,2\pi/c} * g(x) - G * g(x)| \leq \eps$ for all $x$ and combining with (\ref{eq:convolutionbound}), we are done.
\end{proof}

\subsection{Decoupling}
In our full algorithm, we will reconstruct frequency-localized versions of a function independently for different frequencies $\theta$ that we localize around.  We will then combine our localized reconstructions by adding them.  In this section, we prove several inequalities that will allow us to analyze what happens to our estimation error when we add different localized reconstructions together.  Recall that convolving by a truncated polynomial kernel $\mcl{K}_{\mu, l ,c}$ or truncated Gaussian kernel $\mcl{T}_{\mu, l ,a}$ is approximately equivalent to multiplying the Fourier transform by a Gaussian multiplier centered around $\mu$.  Lemma \ref{lem:approx-constant} implies that adding up evenly spaced Gaussian multipliers approximates the constant function.  Thus, we expect that convolving by an expression of the form $\sum_{\mu} \mcl{K}_{\mu, l ,c}$ or $\sum_{\mu} \mcl{T}_{\mu, l ,a}$ where the sum is over evenly spaced $\mu$ should roughly recover the original function.  The first two claims here formalize this intuition.

In the first claim, we analyze what happens when we add several localizations obtained by convolving with various truncated polynomial kernels.
\begin{claim}\label{claim:convolution-bound}
Let $l$ be some parameter and let $0 < \eps < 0.1$ be such that $l \geq \lceil \log 1/\eps \rceil$.  Let $g$ be a function.  Let $c$ be some constant.  Let $S$ be a set of integer multiples of $2 \pi/(cl)$.  Then
\[
\int_{-1}^1 \left \lvert \frac{1}{l\sqrt{2\pi}}\sum_{\mu \in S }\mcl{K}_{\mu, l ,c} * g  \right \rvert^2 \leq (1 + O(\eps |S|))^2 \int_{-(1 + lc)}^{1 + lc} |g|^2 \,,
\]
\end{claim}
\begin{proof}
Note that $\mcl{K}_{\mu,l,c}$ is supported on the interval $[-lc, lc]$ so we can restrict $g$ to be supported on $[-(1 + lc), 1 + lc]$ and $0$ outside the interval.  Define the function 
\[
V(x) =  \frac{1}{l\sqrt{2\pi}}\sum_{\mu \in S } \mcl{K}_{\mu,l,c} \,.
\]
Then 
\[
\wh{V}(x) = \frac{1}{l\sqrt{2\pi}}\sum_{\mu \in S} \wh{\mcl{K}_{0,l,c}}(x - \mu) \,.
\]
Next, let $G$ denote the Gaussian $G = N(0, c^2)$.  Recall by Claim \ref{claim:easyL2bound}, $\norm{\mcl{K}_{0,l,c}(x) - G(x) }_1 \leq O(\eps l)$ so $\norm{\wh{\mcl{K}_{0,l,c}} - \wh{G}}_{\infty} \leq O(\eps l)$.  Let 
\[
U(x) = \frac{1}{l\sqrt{2\pi}}\sum_{\mu \in S} \wh{G}(x - \mu) \,.
\]
Note that since $\wh{G} =  M_{0, (2\pi/c)^2}$, for any $x$,
\[
|U(x)| =  \frac{1}{l\sqrt{2\pi}}\sum_{\mu \in S}  M_{0, (2\pi/c)^2}(x - \mu) \leq  \frac{1}{l\sqrt{2\pi}}\sum_{\mu \in (2\pi)/(cl) \Z}  M_{0, (2\pi/c)^2}(x - \mu) \leq 1 + \eps 
\]
where the last inequality follows from Lemma \ref{lem:approx-constant}.  Also $\norm{\wh{K_{0,l,c}} - \wh{G}}_{\infty} \leq O(\eps l)$ so for all $x$,
\[
|\wh{V}(x)| \leq 1 + O(\eps  |S|) \,.
\]

We conclude
\[
\int_{-1}^1 |V * g(x)|^2 dx \leq \norm{V  * g}_2^2 = \norm{\wh{V}\wh{g}}_2^2 \leq ( 1 + O(\eps  |S|))^2\norm{\wh{g}}_2^2 =( 1 + O(\eps  |S|))^2\norm{g}_2^2
\]
To complete the proof recall that we restricted $g$ to be supported on $[-(1 + lc), 1 + lc]$ and we are done.
\end{proof}

The next claim is similar to the previous one except we analyze what happens when we add several localizations obtained by convolving with various truncated Gaussian kernels.
\begin{claim}\label{claim:sum-approximation}
Let $a.l, \eps$ be parameters such that $0 < \eps < 0.1$ and $ l \geq \lceil \log 1/\eps \rceil$. Let $g$ be a function whose Fourier support is contained in a set $S_0 \subset \R$ and such that $\norm{\wh{g}}_1 \leq 1$.  Let $S$ be a set of integer multiples of $a/l$ that contains all multiples within a distance $l \cdot a$ of $S_0$.  Then
\[
\norm{g - \frac{1}{l \sqrt{2\pi}} \sum_{\mu \in S} \mcl{T}_{\mu, l, a} * g }_{\infty}  \leq 2\eps  \,.
\]
\end{claim}
\begin{proof}
Consider the function
\[
A(x) = \sum_{\mu \in S} \frac{1}{l\sqrt{2\pi}} M_{\mu, a^2}(x) \,.
\]
For all $x \in S_0$, we claim that
\[
|A(x) - 1| \leq \eps \,.
\]
To see this, note that by Lemma \ref{lem:approx-constant},
\[
\left \lvert \sum_{\mu \in (a/l)\Z} \frac{1}{l\sqrt{2\pi}} M_{\mu,a^2}(x) - 1 \right \rvert \leq 0.1\eps
\]
for all $x$.  By assumption, the set $S$ contains all integers that are within $la$ of the set $S_0$ so for any $x \in S_0$, 
\[
\sum_{\mu \in (a/l)\Z \backslash S } \frac{1}{l\sqrt{2\pi}} M_{\mu, a^2}(x) \leq 0.1 \eps \,,
\]
and we conclude that we must have $|A(x) - 1| \leq \eps$.  Next, we claim that if we define
\[
B(x) = \sum_{\mu \in S} \frac{1}{l\sqrt{2\pi}} M_{0, 1}^{\textsf{trunc}(l)}((x - \mu)/a) \,,
\]
then we have for all $x$,
\[
|B(x) - A(x) | \leq \eps \,. 
\]
To see this, first note that $M_{0, 1}((x - \mu)/a) = M_{\mu, a^2}(x)$.  Next, using Gaussian tail decay, we have for all $x$,
\[
|B(x) - A(x) | \leq \sum_{\substack{\mu \in (a/l)\Z \\ |x - \mu| \geq la}}  \frac{1}{l\sqrt{2\pi}} M_{\mu, a^2}(x) \leq \eps \,.
\]
Thus, we have
\[
|B(x) - 1| \leq 2\eps
\]
for $x \in S_0$.  Note that by definition,
\[
\wh{g}(x) - \frac{1}{l\sqrt{2\pi}} \sum_{\mu \in S} \wh{\mcl{T}_{\mu, l, a} * g }(x) = \wh{g}(x)  - \wh{g}(x)\sum_{\mu \in S} \frac{1}{l\sqrt{2\pi}} M_{0, 1}^{\textsf{trunc}(l)}((x - \mu)/a)  = (1 - B(x)) \wh{g}(x)
\]
so therefore
\[
\norm{\wh{g} - \frac{1}{l\sqrt{2\pi}} \sum_{\mu \in S} \wh{\mcl{T}_{\mu, l, a} * g }}_{1} \leq 2 \eps
\]
and we conclude
\[
\norm{ g - \frac{1}{l \sqrt{2\pi}} \sum_{\mu \in S}\mcl{T}_{\mu, l, a} * g }_{\infty}  \leq 2\eps 
\]
as desired.
\end{proof}

The last result in this section will allow us to decouple errors from summing over different localizations.  Note that naively, if we add together $n$ estimates with $L^2$ errors $\eps_1, \dots , \eps_n$, then the resulting $L^2$ error of the sum could be as large as $\eps_1 + \dots + \eps_n$.  If the estimates were ``independent" on the other hand, we would expect the $L^2$ error of the sum to only be $\sqrt{\eps_1^2 + \dots + \eps_n^2}$.  We prove that when adding together functions that are frequency-localized at different locations, the error essentially matches the latter bound (up to logarithmic factors).  This tighter bound will be necessary in the proof of Lemma \ref{lem:fourier-manycluster}.

Note that if we have functions $g_1, \dots , g_n$ whose Fourier supports are disjoint, then it is immediate that 
\[
\norm{g_1 + \dots + g_n}_2^2 = \norm{g_1}_2^2 + \dots + \norm{g_n}_2^2 \,.
\]
However, in our setting, we need to restrict the functions to the interval $[-1,1]$ first, which causes the Fourier supports to no longer be disjoint.  Through a few additional arguments we are able to prove an analogous statement for bounding $\int_{-1}^1|g_1 + \dots + g_n|^2$.  We do pay some additional losses both in the inequality itself and the bounds of the integral i.e. we need to integrate the individual functions over a slightly larger interval.

\begin{claim}\label{claim:decoupling-pre}
Let $\alpha, l, \eps$ be parameters such that $\alpha > 1$ and $ l \geq \lceil \log \alpha n/\eps \rceil$.  Let $I_1, \dots , I_n$ be intervals of length at least $\alpha$ and assume that for any $x \in \R$, at most $l$ of the intervals contain $x$.  Let $g_1, \dots , g_n$ be functions such that for all $j \in [n]$, $\norm{\wh{g_j}}_1 \leq 1$ and $\wh{g_j}$ is supported on $I_j$.  Then 
\[
\int_{-1}^1 |g_1 + \dots + g_n|^2 \leq \poly(l) \left( \eps^5 + \int_{-(1 + \alpha^{-1})}^{(1 + \alpha^{-1})} |g_1|^2 + \dots + \int_{-(1 + \alpha^{-1})}^{(1 + \alpha^{-1})} |g_n|^2 \right)\,.
\]
\end{claim}
\begin{proof}
Consider the Gaussian multiplier $M = M_{\mu, \alpha^{-2} l^{-100}}$ for some $\mu \in [-1,1]$.  Now first, we bound 
\[
\int_{-\infty}^{\infty} M(x)^2 |g_1(x) + \dots + g_n(x)|^2 dx \,.
\]
Define the functions $h_j = \wh{M g_j}$ for all $j$.  Then by Plancherel,
\begin{equation}\label{eq:bound1}
\int_{-\infty}^{\infty} M(x)^2 |g_1(x) + \dots + g_n(x)|^2 dx = \int_{-\infty}^{\infty}|h_1 + \dots + h_n|^2 dx
\end{equation}
On the other hand, note that $h_j = \wh{M} * \wh{g_j}$.  Let $J_j$ denote the interval containing all points within distance at most $10 \alpha l^{60}$ of the interval $I_j$.  Let $h_j'$ be the function $h_j$ restricted to $J_j$ (and equal to $0$ outside).  Recall that the support of $\wh{g_j}$ is contained within $I_j$.  Then we claim that 
\[
\int_{-\infty}^{\infty} |h_j - h_j'|^2 \leq (\eps/(\alpha n))^{100} \,.
\]
This follows because $|\wh{M}| = N(0, (2 \pi \alpha l^{50})^2)$ and $\norm{\wh{g_j}}_1 \leq 1$ so for a point $x$ that is distance $d$ away from the interval $I_j$, we have
\[
|h_j(x)| \leq \max_{ y \in I_j} |\wh{M}(x - y)| \leq  N(0, (2 \pi \alpha l^{50})^2)(d) \,.
\]
Also note that$\norm{\wh{g_j}}_1 \leq 1$, implies $\norm{g_j}_{\infty} \leq 1$ so
\[
\norm{h_j}_2^2 = \norm{Mg_j}_2^2 \leq \norm{M}_2^2  \leq 1 \,.
\]
Combining the previous two inequalities over all $j$, we have
\begin{align*}
\left \lvert \norm{h_1 + \dots + h_n}_2 - \norm{h_1' + \dots + h_n'}_2 \right \rvert \leq (\eps/(\alpha n))^{49} \\
\norm{h_1 + \dots + h_n}_2 + \norm{h_1' + \dots + h_n'}_2 \leq 3n
\end{align*}
which implies
\begin{equation}\label{eq:bound2}
\int_{-\infty}^{\infty}|h_1 + \dots + h_n|^2 dx \leq 0.1 (\eps/(\alpha n))^{10} + \int_{-\infty}^{\infty}|h_1' + \dots + h_n'|^2 dx \,.
\end{equation}
Now note that since not too many of the intervals $I_j$ may contain the same point $x \in \R$, not too many of the extended intervals $J_j$ can contain the same point $x \in \R$.  In particular, at most $O(l^{70})$ of the extended intervals can contain the same point $x \in \R$.  In other words, each point $x \in \R$ is in the support of at most $O(l^{70})$ of the $h_1', \dots , h_n'$.  Thus, by Cauchy Schwarz,
\begin{equation}\label{eq:bound3}
\int_{-\infty}^{\infty}|h_1' + \dots + h_n'|^2 \leq O(l^{70}) \left( \int_{-\infty}^{\infty}|h_1'|^2 + \dots + \int_{-\infty}^{\infty}|h_n'|^2  \right) \,.
\end{equation}
Now we bound 
\begin{equation}\label{eq:bound4}
\int_{-\infty}^{\infty}|h_j'|^2  \leq \int_{-\infty}^{\infty} |h_j|^2 = \int_{-\infty}^{\infty} M(x)^2 |g_j(x)|^2 dx \leq 0.1 (\eps/(\alpha n))^{10} + \int_{- ( 1 + \alpha^{-1})}^{1 + \alpha^{-1}} M(x)^2 |g_j(x)|^2 dx
\end{equation}
where the last step holds because $\norm{g_j}_{\infty} \leq 1$ and the multiplier $M(x)$ is always at most $1$ and decays rapidly outside the interval $[- ( 1 + \alpha^{-1}), 1 + \alpha^{-1}]$ since $\mu \in [-1,1]$.  Putting everything together (\ref{eq:bound1}, \ref{eq:bound2}, \ref{eq:bound3}, \ref{eq:bound4}), we get
\[
\int_{-\infty}^{\infty} M(x)^2 |g_1(x) + \dots + g_n(x)|^2 dx \leq \poly(l) \left( (\eps/(\alpha n))^5 + \sum_{j = 1}^n \int_{- ( 1 + \alpha^{-1})}^{1 + \alpha^{-1}} M(x)^2 |g_j(x)|^2 dx  \right)\,.
\]

Now summing the above over different multipliers $M =  M_{\mu, \alpha^{-2} l^{-100}}$ i.e. with $\mu$ uniformly spaced on $[-1,1]$ with spacing $\alpha^{-1} l^{-50}$, we conclude 
\begin{align*}
\int_{-1}^{1}  |g_1(x) + \dots + g_n(x)|^2 dx \leq 10 \sum_{\mu}\int_{-1}^{1}    M_{\mu, \alpha^{-2} l^{-100}}(x)^2 |g_1(x) + \dots + g_n(x)|^2 dx \\ \leq \poly(l)  \eps^5 + \poly(l) \left(\sum_{\mu}\int_{-(1 + \alpha^{-1})}^{(1 + \alpha^{-1})}  M_{\mu, \alpha^{-2} l^{-100}}(x)^2(|g_1|^2 + \dots +  |g_n|^2) \right) \\ \leq \poly(l)  \eps^5 + \poly(l) \left(\int_{-(1 + \alpha^{-1})}^{(1 + \alpha^{-1})}  |g_1|^2 + \dots +  \int_{-(1 + \alpha^{-1})}^{(1 + \alpha^{-1})}|g_n|^2) \right)\,.
\end{align*}
\end{proof}

\subsubsection{Completing the Proof of Lemma \ref{lem:fourier-manycluster}}

In this section, we will complete the proof of Lemma \ref{lem:fourier-manycluster}.  First, we need to introduce some notation.  We will carry over all of the notation from the statement of Lemma \ref{lem:fourier-manycluster}.  We also use the following conventions:
\begin{itemize}
    \item Let $S_0 = \{\theta_1, \dots , \theta_k \}$ be the frequencies in the Fourier support of $\mcl{M}$
    \item  Let $\gamma > 0$ be parameter to be chosen later and let $l' = \lceil \log 1/ \gamma \rceil$ (we will ensure $\gamma$ is sufficiently small i.e. $\gamma < (\eps c/(k  N))^{K}$ for some sufficiently large absolute constant $K$  ) 
    
    \item Let the function $r(x)$ be defined as $r(x) = f(x) - \mcl{M}(x)$ on the interval $[-1,1]$ and $r(x) = 0$ outside the interval.

\end{itemize}

Recall that the way we will reconstruct the function is by attempting to localize around each of the points in the given set $T$ and reconstructing the localized function using Corollary \ref{coro:fourier-simplecase}.  We then find $k \poly(l'/c)$ of these localized reconstructions that we can combine to approximate the entire function.

In the first claim, we bound the error of our reconstruction using Corollary \ref{coro:fourier-simplecase} for a given localization.  Recall the two types of kernels, the truncated polynomial kernel and the truncated Gaussian kernel, defined in Section \ref{sec:local-fourier-properties}.  Consider the kernels $\mcl{K}_{\mu,l',c/l'}$ and $\mcl{T}_{\mu, l',2\pi l'/c } $ (which, recall, are approximately the same). In the next lemma, we will bound the distance between   $\mcl{K}_{\mu,l',c/l'} * f$ and $\mcl{T}_{\mu, l',2\pi l'/c }  * \mcl{M}$ in terms of $r(x)$.  The reason we care about these two functions is that the first is something that we can compute since we are given explicit access to $f$.  On the other hand, the second is Fourier sparse and has bounded Fourier support so it can be plugged into Corollary \ref{coro:fourier-simplecase} (as the unknown function $\mcl{M}$).

\begin{claim}\label{claim:contraction-bound}
For any real number $\mu$, 
\[
\int_{-1 + c}^{1 - c} | \mcl{K}_{\mu,l',c/l'} * f - \mcl{T}_{\mu, l',2\pi l'/c }   * \mcl{M}|^2 \leq \poly(l'/c)\gamma^2 + 4\int_{-\infty}^{\infty} |M_{0, (2\pi l'/c)^2}(x - \mu) \wh{r}(x)|^2 dx \,.
\]
\end{claim}

\begin{proof}
Note that since $\mcl{K}_{\mu,l',c/ l'}$ is supported on $[-c, c]$, 
\[
\int_{-1 + c}^{1 - c} \left \lvert \left(\mcl{K}_{\mu,l',c/l'} * f \right)(x) - \left(\mcl{K}_{\mu,l',c/l'} * \mcl{M} \right)(x) \right \rvert^2 dx \leq \int_{-1}^1   \lvert  \left(\mcl{K}_{\mu,l',c/l'} * r \right)(x)  \rvert^2  dx
\]
Now the Fourier transform of $\mcl{K}_{\mu,l',c/l'} * r$ is $\wh{\mcl{K}_{0,l',c/l'}}(x - \mu) \wh{r}(x )$ so
\[
\int_{-1}^1  \lvert  \left(\mcl{K}_{\mu,l',c/l'} * r \right)(x)  \rvert^2  dx \leq \int_{-\infty}^{\infty} |\wh{\mcl{K}_{0,l',c/l'}}(x - \mu) \wh{r}(x )|^2 dx
\]
We deduce that
\begin{align*}
\int_{-1 + c}^{1 - c} |\mcl{K}_{\mu,l',c/l'} * f - \mcl{T}_{\mu, l',2\pi l'/c }   * \mcl{M}|^2 \leq 2 \int_{-1 + c}^{1 - c}  \left\lvert \mcl{K}_{\mu,l',c/l'} * \mcl{M} - \mcl{T}_{\mu, l',2\pi l'/c }   * \mcl{M}\right\rvert^2 \\ + 2\int_{-\infty}^{\infty} |\wh{\mcl{K}_{0,l',c/l'}}(x - \mu) \wh{r}(x )|^2 dx \\ \leq \poly(l'/c) \gamma^2 + 2\int_{-\infty}^{\infty} |\wh{\mcl{K}_{0,l',c/l'}}(x - \mu) \wh{r}(x )|^2 dx 
\end{align*}
where the last inequality follows from Claim \ref{claim:fourier-localization}.
\\\\
Note since $r$ is supported on $[-1,1]$ and $\norm{r}_2^2 \leq \eps^2 \leq 0.1$, we must have $\norm{r}_1 \leq 1$ which then implies $\norm{\wh{r}}_{\infty} \leq 1$.  Together with Claim \ref{claim:easyL2bound}, if we let $G = N(0, (c/l')^2)$ then we have
\begin{align*}
\int_{-\infty}^{\infty} |\wh{\mcl{K}_{0,l',c/l'}}(x - \mu) \wh{r}(x)|^2 dx \leq 2\int_{-\infty}^{\infty} |\wh{G}(x - \mu) \wh{r}(x )|^2 dx + 2 \norm{\wh{r}}_{\infty}^2 \norm{\mcl{K}_{0,l',c/l'} - G}_2^2 \\ \leq \poly(l'/c)\gamma^2 + 2\int_{-\infty}^{\infty} |M_{0, (2\pi l'/c)^2}(x - \mu) \wh{r}(x )|^2 dx \,. 
\end{align*}
and combining with the previous inequality, we get the desired result.
\end{proof}

We are now ready to complete the proof of Lemma \ref{lem:fourier-manycluster}.
\begin{proof}[Proof of Lemma \ref{lem:fourier-manycluster}]
First, let $T'$ be the set of integer multiples of $2\pi/c$ that are within distance $(10l')^2/c$ of the set $T$.  For all $\mu \in T'$, do the following.  We compute the function $f^{(\mu)} = \mcl{K}_{\mu, l', c/l'} * f$.  Next we apply Corollary \ref{coro:fourier-simplecase} (with appropriate rescaling) to compute a function $h^{(\mu)}$ in $\poly(l'/c)$ time that has Fourier support in $[ \mu - 2\pi l'^2/c, \mu + 2\pi l'^2/c]$, is $(\poly(l'/c),  \poly(l'/c))$-simple and such that
\begin{equation}\label{eq:piecewise-estimates}
\int_{-1 + c}^{1 - c} |f^{(\mu)} - h^{(\mu)}|^2 \leq 20\left( \gamma^2 +   \int_{-1 + c}^{1 - c} |f^{(\mu)} - \mcl{T}_{\mu, l' , 2\pi l'/c} * \mcl{M}|^2 \right) \,.
\end{equation}
To see why we can do this, note that 
\[
\wh{ \mcl{T}_{\mu, l' , 2\pi l'/c} * \mcl{M}} = M_{0,1}^{\trunc(l')} ( c(x - \mu) /(2 \pi l')) \wh{\mcl{M}}
\]
is supported on $[ \mu - 2\pi l'^2/c, \mu + 2\pi l'^2/c]$.  Also it is clear that $\norm{\wh{ \mcl{T}_{\mu, l' , 2\pi l'/c} * \mcl{M}}}_1 \leq \norm{\wh{\mcl{M}}}_1 \leq 1$.
\\\\
Now we choose a set $U \subset T'$ with $|U| \leq k(10l')^2$ such that the following quantity is minimized:
\[
E_{U} = \sum_{\mu \in U}  \int_{-1 + c}^{1 - c} |f^{(\mu)} - h^{(\mu)}|^2 + \sum_{\mu \in T' \backslash U} \int_{-1 + c}^{1 - c} |f^{(\mu)} |^2 \,.
\]
Note that this can be done using a simple greedy procedure.  First, we obtain a bound on the value $E_U$ that we compute.  Let $U_0$ be the set of all integer multiples of $2\pi/c$ that are within distance $10l'^2/c$ of $S_0$ (recall that $S_0$ is the Fourier support of $\mcl{M}$ which consist of $k$ points).  By assumption, we know that $U_0 \subset T'$ and it is clear that $|U_0| \leq k(10l')^2$.  Note that by definition, for any $\mu \notin U_0$, the function $\mcl{T}_{\mu, l' , 2\pi l'/c} * \mcl{M}$ is identically $0$.  Now using (\ref{eq:piecewise-estimates}), then Claim \ref{claim:contraction-bound},
\begin{align*}
E_{U_0} &\leq 20\sum_{\mu \in U_0}\left( \gamma^2 +   \int_{-1 + c}^{1 - c} |f^{(\mu)} -  \mcl{T}_{\mu, l' , 2\pi l'/c} * \mcl{M}|^2 \right) + \sum_{\mu \in T'\backslash U_0} \int_{-1 + c}^{1 - c} |f^{(\mu)} |^2 \\
& \leq 80   \sum_{\mu \in T'} \left( \poly(l'/c)\gamma^2 + \int_{-\infty}^{\infty} |M_{0, (2\pi l'/c)^2}(x - \mu) \wh{r}(x)|^2 dx\right)
\\ & \leq |T'| \poly(l'/c) \gamma^2 +  80\int_{-\infty}^{\infty} | \wh{r}(x )|^2  \sum_{\mu \in (2\pi/c)\Z} M_{0, (2\pi l'/c)^2}(x - \mu)^2dx
\\ & \leq \gamma + \poly(l') \norm{r}_2^2 \,.
\end{align*}
Note that we used the fact that $\gamma$ is sufficiently small and the tail decay properties of the Gaussian multipliers in the last step.  Thus, we can ensure that the error that we compute satisfies $E_U \leq  \gamma + \poly(l') \norm{r}_2^2$.  Now we output the function
\[
\wt{\mcl{M}} = \sum_{\mu \in U} \frac{1}{l'\sqrt{2 \pi}}h^{(\mu)} \,.
\]
It remains to bound the error between $\wt{\mcl{M}}$ and $f$.  First we apply Claim \ref{claim:decoupling-pre} to decouple over all $\mu \in T'$.  Note that $h^{(\mu)}$ and $\mcl{T}_{\mu, l' , 2\pi l'/c} * \mcl{M}$ both have  Fourier support contained in the interval $[ \mu - 2\pi l'^2/c, \mu +2\pi l'^2/c]$.  For distinct $\mu$ that are integer multiples of $2\pi/c$, there are at most $O(l'^2)$ intervals that contain any point.  Also, note that for all $\mu$, $\norm{h^{(\mu)}}_1 \leq \poly(l'/c)$ and $\norm{\wh{ \mcl{T}_{\mu, l' , 2\pi l'/c} * \mcl{M}}}_1 \leq 1$.  Thus, by Claim \ref{claim:decoupling-pre} (with appropriate rescaling of the functions and the interval),
\begin{align*}
&\int_{-1 + 2c}^{1 - 2c} \left \lvert \wt{\mcl{M}} - \sum_{\mu \in T'} \frac{1}{l'\sqrt{2 \pi}} \mcl{T}_{\mu, l' , 2\pi l'/c} * \mcl{M} \right \rvert \\ &\leq \poly(l'/c) \gamma  + \poly(l') \left( \sum_{\mu \in U} \int_{-1 + c}^{1 - c} |h^{(\mu)} - \mcl{T}_{\mu, l' , 2\pi l'/c} * \mcl{M} |^2  + \sum_{\mu \in T'\backslash U} \int_{-1 + c}^{1 - c} |\mcl{T}_{\mu, l' , 2\pi l'/c} * \mcl{M} |^2  \right)
\\ &\leq \poly(l'/c)\gamma + \poly(l') \left(   \sum_{\mu \in U} \int_{-1 + c}^{1 - c} |h^{(\mu)} - f^{(\mu)}|^2  + \sum_{\mu \in T'\backslash U} \int_{-1 + c}^{1 - c} |f^{(\mu)}|^2  \right)
\\ &\quad + \poly(l') \sum_{\mu \in T'}\int_{-1 + c}^{1 - c} |f^{(\mu)} - \mcl{T}_{\mu, l' , 2\pi l'/c} * \mcl{M} |^2 
\\ &\leq \poly(l'/c) \gamma + \poly(l')\left(E_U +  \sum_{\mu \in T'} \left( \poly(l'/c)\gamma^2 + \int_{-\infty}^{\infty} |M_{0, (2\pi l'/c)^2}(x - \mu) \wh{r}(x )|^2 dx\right)\right)
\\ &\leq \poly(l'/c) \gamma + \poly(l')( E_U + \norm{r}_2^2)
\\& \leq \poly(l') \eps^2
\end{align*}
where we used that $\gamma = (\eps c/(kN))^{O(1)}$ is sufficiently small and Claim \ref{claim:contraction-bound} and the last two inequalities follow from the same argument as in the bound for $E_{U_0}$.   Next, by Claim \ref{claim:sum-approximation} (and the definition of $T'$), we have that
\[
\norm{\mcl{M} - \sum_{\mu \in T'} \frac{1}{l'\sqrt{2 \pi}} \mcl{T}_{\mu, l' , 2\pi l'/c} * \mcl{M}  }_{\infty} \leq O(\eps)
\]
so overall, we conclude
\[
\int_{-1 + 2c}^{1 - 2c} |\wt{\mcl{M}} - \mcl{M}|^2 \leq \poly(l') \eps^2
\]
from which we immediately deduce
\[
\int_{-1 + 2c}^{1 - 2c} |\wt{\mcl{M}} - f|^2 \leq \poly(l') \eps^2 \,.
\]
It is also clear that $\wt{\mcl{M}}$ is $(k \poly(l'/c), k \poly(l'/c))$-simple (since it is a sum of at most $k(10l')^2$ functions that are $( \poly(l'/c), \poly(l'/c))$-simple).  Now we are done because $l' = O(l)$. 
\end{proof}

Similar to obtaining Corollary \ref{coro:fourier-simplecase} from Lemma \ref{lem:Fourier-simplecase}, we can extend Lemma \ref{lem:fourier-manycluster} to work even when we do not know the target accuracy but only a lower bound on it.
\begin{corollary}\label{coro:fourier-manycluster}
Assume we are given $N, k, \eps , c$ with $0 < \eps < 0.1$.  Let $l = \lceil \log kN / (\eps c)  \rceil$ be some parameter.  Let $\mcl{M}$ be a function that is $(k,1)$-simple.  Also, assume that we are given a set $T \subset \R$  of size $N$ such that all of the support of $\wh{\mcl{M}}$ is within distance $1$ of $N$.  Further, assume we are given access to a function $f$.  There is an algorithm that runs in $\poly(N,k,l , 1/c)$ time and outputs a function $\wt{\mcl{M}}$ that is $(k\poly(l/c), k\poly(l/c) )$-simple and such that
\[
\int_{-1 + c}^{1 - c} |\wt{\mcl{M}}(x) - f(x)|^2 dx \leq \eps^2 + \poly(l) \int_{-1}^1 |f(x) - \mcl{M}(x)|^2 dx\,. 
\]
\end{corollary}
\begin{proof}
This will be the exact same argument as the proof of Corollary \ref{coro:fourier-simplecase}.  For a target accuracy $\eps' > \eps$ we run the algorithm in Lemma \ref{lem:fourier-manycluster} to get a function $\wt{\mcl{M}}_{\eps'}(x)$.  We then check whether 
\[
\int_{-1 + c}^{1 - c} |\wt{\mcl{M}}_{\eps'}(x) - f(x)|^2 dx \leq \poly(l) \eps'^2 \,.
\]
If the check passes, we take $\eps' \leftarrow 0.99 \eps'$ and repeat the above until we find the smallest $\eps'$ (up to a constant factor) for which the check passes.  The guarantee of Lemma \ref{lem:fourier-manycluster} implies that for this $\eps'$, we can just output $\wt{\mcl{M}}_{\eps'}(x)$  and it is guaranteed to satisfy the desired inequality.
\end{proof}

\subsection{Proof of Main Theorem}

Using Theorem \ref{thm:sparse-fourier-old} and Corollary \ref{coro:fourier-manycluster}, we can prove our main theorem, Theorem \ref{thm:fourier-main}.  The main thing that we need to prove is that the frequencies in the function $f'$ computed by Theorem \ref{thm:sparse-fourier-old} cover (within distance $\poly(k, \log 1/\eps)$) all of the frequencies in $\mcl{M}$.  This will then let us use the frequencies in $f'$ to construct a set $T$ of size $\poly(k, \log 1/\eps)$ that covers all frequencies in $\mcl{M}$ to within distance $1$ that we can then plug into Corollary \ref{coro:fourier-manycluster}.  We will need the following technical lemma from \cite{chen2016fourier}.
\begin{lemma}\label{lem:extension-bound}[Lemma 5.1 in \cite{chen2016fourier}]
For any $k$-Fourier sparse signal $g: \R \rightarrow \C$,
\[
\max_{x \in [-1,1]} |g(x)|^2 \leq O(k^4 \log^3 k) \int_{-1}^1 |g(x)|^2 dx \,.
\]
\end{lemma}
The above lemma roughly says that the mass of a $k$-Fourier sparse function cannot be too concentrated.  We now finish the proof of Theorem \ref{thm:fourier-main}.

\begin{proof}[Proof of Theorem \ref{thm:fourier-main}]
We first apply Theorem \ref{thm:sparse-fourier-old} to compute a function $f'$. such that
\begin{enumerate}
    \item $f'$ is $(\poly(k, \log 1/\eps), \exp( \poly(k, \log 1/\eps)) )$- simple
    \item 
    \[
    \int_{-1}^1 |f' - f|^2 \leq O\left(\eps^2 + \int_{-1}^1 |f - \mcl{M}|^2 \right)  \,.
    \]
\end{enumerate}

Let $L = (k \log 1/\eps)^{K}$ for some sufficiently large absolute constant $K$.  Let $\gamma = e^{-L}$.  Now apply Claim \ref{claim:sum-approximation} on the function $f'$ with parameters $a \leftarrow L, l \leftarrow L, \eps \leftarrow \gamma $. Let $S \subset L \Z$ be the set of all integer multiples of $L$ that are within distance $L^3$ of the Fourier support of $f'$.  We have
\begin{equation}\label{eq:infbound1}
\norm{f' - \frac{1}{L\sqrt{2\pi}} \sum_{\mu \in S} \mcl{T}_{\mu, L , L^2} * f' }_{\infty} \leq 2 \norm{\wh{f'}}_1 \gamma 
\end{equation}

Next, we apply Claim \ref{claim:fourier-localization} on the function $ (\mcl{M} - f')$.  We deduce that for any $\mu$,
\begin{equation}\label{eq:infbound2}
\norm{\mcl{K}_{\mu, L, 2\pi/L^2} * (\mcl{M} - f')-  \mcl{T}_{\mu, L , L^2} *(\mcl{M} - f')}_{\infty} \leq O\left(\gamma L \left(\norm{\wh{f'}}_1  +\norm{ \wh{\mcl{M}}}_1 \right) \right)  \,.
\end{equation}
Finally, by Claim \ref{claim:convolution-bound} applied to the function $ (\mcl{M} - f')$ (with parameters $\eps \leftarrow \gamma ,l \leftarrow L, c \leftarrow (2\pi)/L^2$), we have
\begin{align*}
\int_{-1}^1 \left \lvert  \frac{1}{L\sqrt{2\pi}} \sum_{\mu \in S}\mcl{K}_{\mu, L, 2\pi/L^2} * (\mcl{M} - f')\right \rvert^2 \leq (1 + O\left(\gamma |S|\right) )^2\int_{-1 - 2\pi/L}^{1 + 2\pi /L} | \mcl{M} - f'|^2   \leq  2 \int_{-1}^1 | \mcl{M} - f'|^2 \,.
\end{align*}
Note that in the last step we use Lemma \ref{lem:extension-bound} and the fact that $ \mcl{M} - f'$ is $\poly(k, \log 1/\eps)$-Fourier sparse so choosing $L = (k \log 1/\eps)^{O(1)}$ sufficiently large ensures that
\[
\int_{-1 - 2\pi/L}^{1 +2\pi /L} | \mcl{M} - f'|^2  \leq 1.1 \int_{-1}^1 | \mcl{M} - f'|^2 \,.
\]
Define the functions
\begin{align*}
&A(x) = f' - \frac{1}{L\sqrt{2\pi}} \sum_{\mu \in S}\mcl{T}_{\mu,L,L^2} * \mcl{M} \\
&B(x) = \frac{1}{L\sqrt{2\pi}} \sum_{\mu \in S}\mcl{K}_{\mu, L, 2\pi/L^2} * (  \mcl{M} - f')
\end{align*}
Note $\norm{A}_{\infty}, \norm{B}_{\infty} \leq |S|\left(\norm{\wh{f'}}_1 + \norm{\wh{\mcl{M}}}_1 \right) $.  Combining (\ref{eq:infbound1}, \ref{eq:infbound2}) we have
\[
\int_{-1}^1 |A(x)|^2 - \int_{-1}^1 |B(x)|^2  \leq \gamma \poly\left(L, |S|, \norm{\wh{f'}}_1 + \norm{\wh{\mcl{M}}}_1 \right)
\]
However, we proved that $\int_{-1}^1 |B(x)|^2 \leq 2 \int_{-1}^1 |\mcl{M} - f'|^2  $ so choosing $L = (k \log 1/\eps)^{O(1)}$ sufficiently large and since $\gamma = e^{-L}$, we conclude
\begin{equation}\label{eq:true-errorbound}
\int_{-1}^1 |A(x)|^2 \leq \eps^2 +  2\int_{-1}^1 |\mcl{M} - f'|^2 \leq O\left(\eps^2  +  \int_{-1}^1 |\mcl{M} - f|^2 \right) 
\end{equation}
where we are using the guarantee from Theorem \ref{thm:sparse-fourier-old}.  Now note that the function 
\[
\mcl{M}' = \frac{1}{L\sqrt{2\pi}} \sum_{\mu \in S}\mcl{T}_{\mu,L,L^2} *\mcl{M}
\]
is $k$-Fourier sparse and has $\norm{\wh{\mcl{M}'}}_1 \leq 2\norm{\wh{\mcl{M}}}_1$ by Lemma \ref{lem:approx-constant}.   Furthermore, all of its Fourier support is within distance $\poly(L)$ of the Fourier support of $f'$ (by the construction of the set $S$).  Thus, we can apply Corollary \ref{coro:fourier-manycluster} on $f'$ (where we treat the unknown function as $\mcl{M}'$) with $N = \poly(L) = \poly(k , \log 1/\eps)$ and recover a function $\wt{\mcl{M}}$ such that 
\begin{align*}
\int_{-1 + c}^{1 - c} |\wt{\mcl{M}} - f'|^2  \leq \eps^2 + \poly(\log k/(\eps c)) \int_{-1}^1 |f' - \mcl{M}'| = \eps^2 + \poly(\log k/(\eps c)) \int_{-1}^1 |A(x)|^2   \\ =  \poly(\log k/(\eps c)) \left(\eps^2  +  \int_{-1}^1 |\mcl{M} - f|^2 \right) 
\end{align*}
where the last step is from (\ref{eq:true-errorbound}).  Since $\int_{-1}^1 |f' - f|^2 \leq O\left(\eps^2 + \int_{-1}^1 |f - \mcl{M}|^2 \right)$, the above implies
\[
\int_{-1 + c}^{1 - c} |\wt{\mcl{M}} - f|^2  \leq \poly(\log k/(\eps c)) \left(\eps^2  +  \int_{-1}^1 |\mcl{M} - f|^2 \right) 
\]
and we are done.
\end{proof}

\subsection{Implementation of Computations}\label{sec:implementation}

In the proof of Theorem \ref{thm:fourier-main}, we use the result from \cite{chen2016fourier} to obtain an approximation $f'$ that is written as a sum of $\poly(k, \log 1/\eps)$ exponentials and has coefficients bounded by $\exp(\poly (k , \log 1/\eps))$.  We then perform explicit computations using this function in our algorithm to eventually compute a sparser approximation with smaller coefficients.  Here we briefly explain why these explicit computations can all be implemented efficiently.  Note that all of the functions that we perform computations on can be written as sums of polynomials multiplied by exponentials i.e.
\begin{equation}\label{eq:func-form}
P_1(x)e^{2\pi i \theta_1 x} + \dots + P_n(x)e^{2\pi i \theta_n x}
\end{equation}
where there are at most $\poly(k, \log 1/\eps)$ terms in the sum, all of the polynomials have degree at most $\poly(k, \log 1/\eps)$ and all of the coefficients are bounded by $\exp(\poly (k , \log 1/\eps))$.  To see this, note that convolving by a polynomial $P(x)$ truncated to an interval (recall the truncated polynomial kernel in Definition \ref{def:localization-type2} ) preserves a function of the form in (\ref{eq:func-form}) (only increasing the degrees of the polynomials by $\deg(P)$).  All other computations that we need such as computing the exact value, adding and multiplying and integrating over some interval can clearly be done explicitly in $\poly (k , \log 1/\eps)$ time and to $\exp(\poly (k , \log 1/\eps))^{-1}$ accuracy for functions of the form specified in (\ref{eq:func-form}).

\section{Basic Tools}\label{appendix:basic}
In this section, we have a few basic tools that are used repeatedly throughout the paper.
\subsection{Chebyshev Polynomials}
Here we will introduce several basic results about the  Chebyshev polynomials, which have algorithmic applications in a wide variety of settings \cite{rivlin2020chebyshev, guruswami2016robust}.

\begin{definition}[Chebyshev Polynomials]
The Chebyshev Polynomials are a family of polynomials defined as follows: $T_0(x) = 1, T_1(x) = x$ and for $n \geq 2$,
\[
T_n(x) = 2xT_{n-1}(x) - T_{n-2}(x) \,.
\]
\end{definition}

\begin{fact}\label{fact:chebyshev-basic}
The Chebyshev polynomials satisfy the following property:
\[
T_n(\cos \theta) = \cos n \theta \,.
\]
\end{fact}

As an immediate consequence of the above, we have a few additional properties.
\begin{fact}\label{fact:chebyshev-basic2}
The Chebyshev polynomials satisfy the following properties:
\begin{enumerate}
    \item $T_n(x)$ has degree $n$ and leading coefficient $2^{n-1}$
    \item For $x \in [-1,1]$, $T_n(x) \in [-1,1 ]$
    \item $T_n(x)$ has $n$ zeros all in the interval $[-1,1]$
    \item There are $n+1$ values of $x$ for which $T_n(x) = \pm 1$, all in the interval $[-1,1]$
\end{enumerate}
\end{fact}

\noindent In light of the previous properties, we make the following definition.

\begin{definition}
For an integer $n$, we define the Chebyshev points of degree $n$, say $t_0, \dots , t_n$, as the points in the interval $[-1,1]$ where the Chebyshev polynomial satisfies $T_n(t_j) = \pm 1$.  Note that the Chebyshev points are exactly
\[ 
\{ \cos 0, \cos \frac{\pi}{n}, \dots \cos \frac{(n-1)\pi}{n}, \cos \pi \} \,.
\]
\end{definition}

Next, we have a result saying that if we have a bound on the value of a degree-$n$ polynomial at all of the degree $n$ Chebyshev points, then we can bound the value over the entire interval $[-1,1]$.  Similar results are used in \cite{rivlin2020chebyshev, guruswami2016robust}, but there does not appear to be a directly usable reference.

\begin{claim}\label{claim:chebyshev-bound}
Let $P(x)$ be a polynomial of degree at most $n$ with real coefficients.  Let $t_0, \dots , t_n$ be the Chebyshev points of degree $n$.  Assume that $|P(t_j)| \leq 1 $  for $j = 0,1, \dots , n$.  Then $|P(x) | \leq 2n$ for all $x \in [-1,1]$.
\end{claim}
\begin{proof}
By Lagrange interpolation, we may write
\[
P(x) = \frac{P(t_0) (x - t_1) \cdots (x - t_n)}{(t_0 - t_1) \cdots (t_0 - t_n)}  + \dots + \frac{P(t_n) (x - t_0) \cdots (x - t_{n-1})}{(t_n - t_0) \cdots (t_n - t_{n-1})} \,.
\]
Thus, it suffices to upper bound the quantity
\[
F(x) = \left\lvert \frac{(x - t_1) \cdots (x - t_n)}{(t_0 - t_1) \cdots (t_0 - t_n)} \right\rvert  + \dots + \left\lvert \frac{ (x - t_0) \cdots (x - t_{n-1})}{(t_n - t_0) \cdots (t_n - t_{n-1})} \right\rvert 
\]
on the interval $[-1,1]$.  Note that by Lagrange interpolation on $T_n(x)$, we have
\[
T_n(x) = \frac{T_n(t_0) (x - t_1) \cdots (x - t_n)}{(t_0 - t_1) \cdots (t_0 - t_n)}  + \dots + \frac{T_n(t_n) (x - t_0) \cdots (x - t_{n-1})}{(t_n - t_0) \cdots (t_n - t_{n-1})} \,.
\]
Also note that $T_n(t_j) = (-1)^{n - j} $ which has the same sign as $(t_j - t_0) \cdots (t_j - t_{j-1}) (t_j - t_{j+1}) \cdots (t_j - t_n)$.  Thus, 
\[
\left\lvert \frac{1}{(t_0 - t_1) \cdots (t_0 - t_n)} \right\rvert  + \dots + \left\lvert \frac{ 1}{(t_n - t_0) \cdots (t_n - t_{n-1})} \right\rvert  = 2^{n-1} \,,
\]
since the leading coefficient of $T_n(x)$ is $2^{n-1}$.  Now we will upper bound
\[
M = \max\left( |(x - t_1) \cdots (x - t_n)|, \dots , |(x - t_0) \cdots (x - t_{n-1})| \right) 
\]
and once we do this, we will have a bound on $F(x)$ since $F(x) \leq 2^{n-1}M$.  Define the polynomial
\[
Q(x) = (x-t_0)(x-t_1) \cdots (x - t_n) = \frac{\sqrt{x^2 - 1}}{2^n} \left( (x+ \sqrt{x^2 - 1})^n  - (x -  \sqrt{x^2 - 1})^n \right) \,.
\]
To see why the last equality is true, note that the RHS has roots at $t_0, \dots , t_n$ and is a monic polynomial of degree $n+1$ so it must be equal to $(x - t_0) \cdots (x - t_n)$.  Now, 
\[
M = \max\left( \left \lvert \frac{Q(x)}{x - t_0} \right \rvert, \dots , \left \lvert \frac{Q(x)}{x - t_n} \right \rvert \right) \leq \max |Q'(x)| 
\]
where the last step holds by the mean value theorem (because $Q(t_j) = 0$ for all $j$).  Now note that 
\[
Q(\cos \theta) = -\frac{\sin \theta \sin(n\theta)}{2^{n-1}}
\]
so 
\[
Q'(\cos \theta) = \frac{n\cos n \theta}{2^{n-1}}  + \frac{\cos \theta \sin(n \theta)}{\sin(\theta) 2^{n-1}} 
\]
and from the above it is clear that
\[
|Q'(\cos \theta)| \leq \frac{n}{2^{n-1}} + \frac{n}{2^{n-1}} = \frac{n}{2^{n-2}} \,.
\]
Now we are done because
\[
\max_{x \in [-1,1]} |P(x)| \leq F(x) \leq 2^{n-1}M \leq 2n \,.
\]
\end{proof}

It turns out that we can restate the above result in terms of convex hulls of points on the moment curve.  This reformulation is the version that is useful in our algorithms.
\begin{definition}
For a real number $x$, we define the moment vector $\mcl{V}_n(x) = (1,x, \dots , x^{n})$.
\end{definition}

\begin{corollary}\label{coro:convexhull}
Let $t_0, t_1, \dots , t_n$ be the Chebyshev points of degree $n$.  Then for any $x \in [-1,1]$, the point $\mcl{V}_n(x)$ is contained in the convex hull of the points 
\[
\{ \pm 2n\mcl{V}_n(t_0), \dots , \pm 2n\mcl{V}_n(t_n) \} \,.
\]
\end{corollary}
\begin{proof}
Assume for the sake of contradiction that the above is not true.  Then there must be a separating hyperplane.  Assume that this hyperplane is given by $a \cdot x = b$ where $a$ is a vector and $b$ is a real number.  Now WLOG $b \geq 0$ and we must have
\begin{align*}
& a \cdot \mcl{V}_n(x) \geq b \\
& |a \cdot \mcl{V}_n(t_j) | \leq \frac{b}{2n}  \quad \forall j 
\end{align*}
However, applying Claim \ref{claim:chebyshev-bound} with $P(x) = \frac{2n (a \cdot \mcl{V}_n(x)))}{b}$ gives a contradiction.  Thus, no separating hyperplane can exist and we are done.
\end{proof}

\subsection{Approximating a Gaussian with a Polynomial}
We will also need to approximate Gaussians with polynomials.  This is a somewhat standard result which we state below. 
\begin{lemma}\label{lem:Gaussian-taylorseries}
Let $G = N(0, 1)$ be the standard Gaussian.  Let $l$ be some parameter.  Then we can compute a polynomial $P(x)$ of degree $(10l)^2$ such that for all $x \in [-2l, 2l]$, 
\[
|G(x) - P(x)| \leq e^{-l} \,. 
\]
\end{lemma}
\begin{proof}
Write
\[
G(x) = \frac{1}{ \sqrt{2\pi}} e^{-x^2/2} \,.
\]
and now we can write the Taylor expansion
\[
e^{-\frac{x^2}{2} } = \sum_{m = 0}^\infty \frac{\left(-\frac{x^2}{2}\right)^m }{m!} = \sum_{m =0}^{\infty}\frac{(-1)^m x^{2m}}{2^m m!}
\]

Now define 
\[
P(x) = \sum_{m = 0}^{(10l)^2} \frac{(-1)^m x^{2m}}{2^m m!}  \,.
\]
For $x \in [-2l,2l]$, we have
\begin{align*}
|G(x) - P(x)| \leq \left \lvert \sum_{m =(10l)^2 + 1}^\infty   \frac{(-1)^m x^{2m}}{2^m m!} \right \rvert \leq \sum_{m =10^2l + 1}^\infty \frac{ (2l)^{2m}}{2^m m!} \leq \sum_{m =10^2l + 1}^\infty  \left(\frac{(2l)^2}{2m/3} \right)^{m} \\ \leq \sum_{m =10^2l + 1}^\infty \frac{1}{2^{m}}\leq e^{-l} \,.
\end{align*}

\end{proof}

In light of the above, we use the following notation.
\begin{definition}\label{def:poly-approx-Gaussian}
We will use $\mcl{P}_{l}(x)$ to denote the polynomial computed in Lemma \ref{lem:Gaussian-taylorseries} for parameter $l$.  Note that $\mcl{P}_l$ is a polynomial of degree $(10l)^2$ and for $G = N(0,1)$, we have
\[
|G(x) - \mcl{P}_l(x)| \leq e^{-l}
\]
for $x \in [-2l, 2l]$.
\end{definition}

\subsection{Linear Regression}
Recall that at the core of the problems we are studying, we are given some function $f$ and want to approximate it as a weighted sum $a_1 f_1 + \dots + a_nf_n$ of some functions $f_1, \dots , f_n \in \mcl{F}$ for some family of functions $\mcl{F}$.  The result below allows us to solve the problem of computing the coefficients if we already know the components $f_1, \dots , f_n$ that we want to use.  The precise technical statement is slightly more complicated in order to incorporate the various types of additional constraints that we may want to impose on the coefficients $a_1, \dots , a_n$.

\begin{lemma}\label{lem:linear-regression}
Let $\mcl{D}$ be a distribution on $\R$ that we are given.  Also assume that we are given functions $f ,f_1, \dots , f_n, g, g_1, \dots , g_n : \R \rightarrow \R$. 
Assume that there are nonnegative coefficients $a_1, \dots , a_n$ such that $a_1 + \dots + a_n \leq 1$ and
\[
\int_{-\infty}^{\infty} (f(x) - a_1f_1(x) - \dots - a_nf_n(x))^2 \mcl{D}(x)dx  + \int_{-\infty}^{\infty} (g(x) - a_1g_1(x) - \dots - a_ng_n(x))^2 \mcl{D}(x)dx\leq \eps^2
\]
for some parameter $\eps > 0$.  Then there is an algorithm that runs in $\poly(n, \log 1/\eps)$ time and outputs nonnegative coefficients $b_1, \dots , b_n$ such that $b_1 + \dots + b_n \leq 1$ and 
\[
\int_{-\infty}^{\infty} (f(x) - b_1f_1(x) - \dots - b_nf_n(x))^2 \mcl{D}(x)dx + \int_{-\infty}^{\infty} (g(x) - b_1g_1(x) - \dots - b_ng_n(x))^2 \mcl{D}(x)dx \leq 2\eps^2 \,.
\]
\end{lemma}
\begin{proof}
Let $v = (1,b_1, \dots , b_n)$.  Note that we can write 
\[
\int_{-\infty}^{\infty} (f(x) - b_1f_1(x) - \dots - b_nf_n(x))^2 \mcl{D}(x)dx + \int_{-\infty}^{\infty} (g(x) - b_1g_1(x) - \dots - b_ng_n(x))^2 \mcl{D}(x)dx = v^TMv
\]
where $M$ is a matrix whose entries are $\int_{-\infty}^{\infty} \left(f(x) f_j(x) + g(x)g_j(x)\right) \mcl{D}(x)dx$ in the first row and column and the other entries are $\int_{-\infty}^{\infty} \left( f_i(x) f_j(x) + g_i(x)g_j(x) \right) \mcl{D}(x)dx$.  Since all of these functions are given to us, we can explicitly compute $M$.  Also note that clearly $M$ is positive semidefinite.  Thus, we can compute its positive semidefinite square root, say $N$.  Now 
\[
v^TMv = \norm{Nv}_2^2
\]
so it remains to solve $\min_{v} \norm{Nv}_2^2 $ which is a convex optimization problem that we can solve efficiently (the size of the problem is $\poly(n)$).
\end{proof}

\end{document}